\documentclass[3p]{elsarticle}

\journal{arXiv}

\usepackage{amsfonts}
\usepackage{amsmath, amscd, amssymb, bm, amsbsy, epsf, amsthm}
\usepackage{enumerate}
\usepackage{paralist}
\usepackage{graphicx,color}
\usepackage{subfigure}
\usepackage{mathrsfs}
\usepackage{verbatim}
\usepackage{inputenc}
\usepackage{diagbox}
\usepackage{algorithm}
\usepackage{algorithmicx}
\usepackage{algpseudocode}
\usepackage{hyperref}
\usepackage{pifont}
\usepackage{inputenc}
\usepackage{bbm, dsfont}
\usepackage{multicol}
\usepackage{balance}
\usepackage{verbatim, booktabs}

\usepackage{multirow, url}
\usepackage[table]{xcolor}
\usepackage{diagbox}
\usepackage{tikz}

\usepackage{cmbright}

\DeclareMathAlphabet{\mathpzc}{OT1}{pzc}{m}{it}
\DeclareMathAlphabet{\mathpzc}{OT1}{pzc}{m}{it}

\newtheorem{theorem}{Theorem}

\newtheorem{proposition}[theorem]{Proposition}
\newdefinition{definition}{Definition}
\newdefinition{hypothesis}{Hypothesis}
\newdefinition{problem}{Problem}
\newdefinition{remark}{Remark}
\newdefinition{example}{Example}

\def\N{\boldsymbol{\mathbbm{N}}}
\def\R{\boldsymbol{\mathbbm{R}}}
\def\defining{\overset{\mathbf{def}}=}

\def\B{\mathcal{B}}

\def\D{\mathcal{D}}
\def\A{\mathcal{A}}
\def\S{\mathcal{S}}

\def\Exp{\boldsymbol{\mathbb{E}}}
\def\prob{\boldsymbol{\mathbbm{P} } }

\def\x{\boldsymbol{x} }
\def\y{\boldsymbol{y} }

\def\p{\boldsymbol{p}}
\def\g{\boldsymbol{g}}
\def\capac{\boldsymbol{c} }

\DeclareMathAlphabet{\mathpzc}{OT1}{pzc}{m}{it}

\DeclareMathOperator*{\talg}{\textit{T-alg}}
\DeclareMathOperator*{\dist}{\textit{dist}}
\DeclareMathOperator*{\eff}{\textit{eff}}

\newcommand\X{ \mathsf{X} } 
\newcommand\Y{ \mathsf{Y} } 
\newcommand\Z{ \mathsf{Z} } 
\newcommand\ups{ \mathsf{S} } 
\newcommand\uph{ \mathsf{H} } 
\newcommand\upc{ \mathsf{C} } 

\begin{document}

\begin{frontmatter}

\title{Analysis of Divide \& Conquer strategies for the \\0-1 Minimization Knapsack Problem}
\tnotetext[mytitlenote]{This material is based upon work supported by project HERMES 45713 from Universidad Nacional de Colombia,
Sede Medell\'in.}

\author[mymainaddress]{Fernando A Morales} 
\cortext[mycorrespondingauthor]{Corresponding Author}
\ead{famoralesj@unal.edu.co}

\author[mymainaddress]{Jairo A Mart\'inez}




\address[mymainaddress]{Escuela de Matem\'aticas
Universidad Nacional de Colombia, Sede Medell\'in \\
Carrera 65 \# 59A--110, Bloque 43, of 106,
Medell\'in - Colombia}


\begin{abstract}
We introduce and asses several Divide \& Conquer heuristic strategies aimed to solve large instances of the \c0-1 Minimization Knapsack Problem. The method subdivides a large problem in two smaller ones (or recursive iterations of the same principle), to lower down the global computational complexity of the original problem, at the expense of a moderate loss of quality in the solution. Theoretical mathematical results are presented in order to guarantee an algorithmically successful application of the method and to suggest the potential strategies for its implementation. In contrast, due to the lack of theoretical results, the solution's quality deterioration is measured empirically by means of Monte Carlo simulations for several types and values of the chosen strategies. Finally, introducing parameters of efficiency we suggest the best strategies depending on the data input.
\end{abstract}

\begin{keyword}
Divide and Conquer, Minimization Knapsack Problem, Monte Carlo simulations, method's efficiency.
\MSC[2010] 90C59 \sep 90C06 \sep 90C10 \sep 65C05 \sep 68U01
\end{keyword}

\end{frontmatter}



%
%
%
%
%
%
%
%
%
\section{Introduction}   
%
%
%
%
The \textit{Knapsack Problem} (KP) is, beyond dispute, one of the fundamental problems in integer optimization for three main reasons. First, due to its simplicity with respect to a general linear integer (or mixed integer) optimization problem. Second, because of its occurrence as a subproblem of an overwhelming number of optimization problems, including a wide variety of real life situations which can be modeled by KP. Third, because it belongs to the NP hard class problems which makes it relevant from the theoretical perspective. As a natural consequence, there is a vast literature dedicated to the KP solution, comprising a broad spectrum between exact algorithms such as Dynamic Programming (DP) and Branch \& Bound (B\&B) techniques \cite{kellerer2005knapsack}, metaheuristic schemes such as Genetic Algorithms (GA), Ant Colony Algorithms (ACO's) and hybrid algorithms, including \textit{matheuristics} and \textit{symheuristics} \cite{blum2011hybrid}, \cite{blum2003metaheuristics}, \cite{juan2015review}, \cite{voss2009matheuristics}, \cite{wilbaut2008survey}; an early review of non-standard versions of KP is found in \cite{lin1998bibliographical}, a detailed review of some versions is found in the texts \cite{kellerer2005knapsack} and \cite{martello1990knapsack}. The convergence analysis for some of the aforementioned algorithms is presented in \cite{GLOVER20023}, \cite{Gutjahr2003}, \cite{Gutjahr2010}, \cite{Hanafi2001}.
As with any optimization problem, for the KP solution it is crucial to exploit the trade-off between the quality of the solution in terms of the value of the objective function, and the computational effort required to obtain it.\\

\noindent Both, exact methods and metaheuristic algorithms have disadvantages. Exact algorithms such as DP and B\&B usually are insufficient to address large instances: all dynamic programming versions for KP are pseudo-polynomial, i.e. time and memory requirements are dependent on the instance size. Commonly, the computational complexity of the algorithms B\&B cannot be explicitly described, as it is not possible to estimate \textit{a priori} the number of search tree nodes required (see \cite{kellerer2005knapsack}, \cite{pisinger2005hard}). On the other hand, most metaheuristics lack sufficient theoretical justification. Despite the widespread success of such techniques, among researchers there is little understanding about the key aspects of their design, including the identification of search space characteristics that influence the difficulty of the problem. There are some theoretical results related to the convergence of algorithms under appropriate probabilistic hypotheses, however these are not useful from the practical point of view. Moreover, it is not possible to argue that any of the particular metaheuristics is on average superior to any other, so the choice of a metaheuristics to address a specific optimization problem depends largely on the user's experience \cite{voss2009matheuristics}.\\

\noindent As a consequence of the KP's relevance, it is natural that any proposed method for solving integer optimization problems: theoretical, empirical or mixed, is usually first tested on a Knapsack Problem.
This is the case of the present work, where we introduce a Divide and Conquer (D\&C) strategy aimed to solve large instances of the 0-1 Minimization Knapsack Problem \ref{Pblm Integer Problem} below (from now on 0-1 MKP). 
%
The main goal of the proposed approach is to reduce the computational complexity of the 0-1 MKP by subdividing the original/initial problem in two smaller subproblems, at the price of giving up (to some extent) quality of the solution. Moreover, using multiple recursive D\&C iterations the initial problem can be decomposed on several subproblems of suitable size (at the price of further deterioration in the solution's quality), in a multilevel scheme, see \textsc{Figures} \ref{Fig Tree Generated by Head-Left Algorithm}, \ref{Fig Tree Generated by Head-Left Algorithm Biased} and \ref{Fig Tree Generated by Balanced Left-Right Algorithm}. 
The multilevel paradigm is not a metaheuristic in itself, on the contrary, it must act in collaboration with some solution strategy, be it an exact or approximate procedure. For the method to be worthy, the loss of quality vs. the reduction of computational time must lie within an acceptable range. Consequently, the present work first introduces the technique, together with several strategies for its implementation. Next, the quality is defined using several parameters of efficiency. Finally, since no theoretical results can be mathematically shown for measuring the efficiency of the method, we proceed empirically using Monte Carlo simulations and the Law of Large Numbers (see \textsc{Theorem} \ref{Th the Law of Large Numbers} below) to identify which strategy will likely be the best, when the data input of the problem are regarded as random variables with known probabilistic distribution.\\

\noindent We close this section mentioning that different authors have reported the increased performance of metaheuristic techniques when used in conjunction with a multilevel scheme on large instances. The multilevel paradigm has been used mainly in mesh construction, \textit{Graph Partition Problem} (GPP), \textit{Capacitated Multicommodity Network Design} (CMND), \textit{Covering Design} (CD), \textit{Graph Colouring } (GC), \textit{Graph Ordering} (GO), \textit{Traveling Salesman Problem} (TSP) and \textit{Vehicle Routing Problem} (VRP) \cite{banos2003multilevel}, \cite{walshaw2008multilevel}. To the Authors' best knowledge, the use of a multilevel D\&C scheme for solving the 0-1 Minimization Knapsack Problem has not been reported.\\

\section{Preliminaries}\label{Sec Preliminaries}
%
%
%
%
%
%
In this section the general setting and preliminaries of the problem are presented. We start introducing the mathematical notation. For any natural number $ N\in \N $, the symbol $ [N] \defining \{ 1, 2, \ldots , N \} $ indicates the set/window of the first $ N $ natural numbers. For any set $ E $ we denote by $ \vert E \vert $ its cardinal and $ \wp(E) $ its power set. 
A particularly important set is $ \mathcal{S}_{N} $, where $ \mathcal{S}_{N} $ denotes the collection of all permutations in $ [N] $, its elements will be usually denoted by $ \pi, \sigma, \tau $, etc. Random variables will be represented with upright capital letters, e.g. $ \X, \Y, \Z, ... $ and its respective expectations with $ \Exp(\X),  \Exp(\Y), \Exp(\Z), ... $.
Vectors are indicated with bold letters, namely $ \p, \g, ... $ etc. Particularly important collections of objects will be written with calligraphic characters, e.g. $ \mathcal{A}, \mathcal{D}, \mathcal{E} $ to add emphasis. For any real number $ x \in \R $ the floor and ceiling function are given (and denoted) by $ \lfloor x \rfloor \defining \max\{k: \ell\leq x, \, k \text{ integer}\} $, $ \lceil x \rceil \defining \max\{k: k\geq x, \, k \text{ integer}\} $, respectively.
%
%
\subsection{The Problem}\label{Sec The Problem}
%
%
Now we introduce the 0-1 Minimization Knapsack Problem. 
%
%
\begin{problem}[0-1 Minimization KP]\label{Pblm Integer Problem}
	\begin{subequations}\label{Eqn Integer Problem}
		\begin{equation}\label{Eqn Integer Problem Objective Function}
		\min \sum\limits_{i\, \in \, [N] } p_{i}x_{i} ,
		\end{equation}
		subject to
		\begin{equation}\label{Eqn Integer Problem Capacity Constraint}
		\sum\limits_{i\, \in \, [N] } c_{i}x_{i} \geq D,
		\end{equation}
		\begin{align}\label{Eqn Integer Problem Choice Constraint}
		& x_{i} \in \{0,1\}, &
		& \text{for all } i \,\in \,[N] .
		\end{align}	
	\end{subequations}
	Here, $  \x \defining \big(x_{i}: i \in [N] \big) $ is the list of binary valued decision variables. In addition, the capacity coefficients $ \capac \defining \big(c_{i}: i \in [N]\big) $ as well as the costs $ \p \defining \big(p_{i}: i \in [N]\big) $
	, are all positive integers. In the sequel, the feasible set is denoted by 
	\begin{equation}\label{Eqn Feasible Set}
	S \defining \big\{ \x\in \{0,1\}^{N}: \capac\cdot \x \geq D \big\}
	\end{equation}
	and the problem can be written concisely as
	\begin{equation}\label{Eqn Concise Integer Problem}
	z^{*} \defining \min\big\{ \p\cdot \x : \x\, \in \, S
	\big\} ,
	\end{equation}
	where $ z_{*} $ denotes the optimal solution value.	
\end{problem}
%
%
In general, the 0-1 MKP can be understood as the problem of buying items (buses, aircraft, ships fleet), denoted by the index $ i = 1, \ldots, N $, with corresponding costs $ p_{i} $ and capacities $ c_{i} $. Therefore, the natural question is to choose a set of items to minimize its total cost but whose overall capacity satisfies a minimum threshold demand $ D $.
Observe that the solution of \textsc{Problem} \ref{Pblm Integer Problem} above can be found using the solution of the following Knapsack Problem
\begin{problem}\label{Pblm Equivalent Knapsack Problem}
	%
	%
	\begin{subequations}\label{Eqn Equivalent Knapsack Problem}
		\begin{equation}\label{Eqn Equivalent Knapsack Problem Objective Function}
		\max \sum\limits_{i\, \in \, [N] } p_{i}\xi_{i} ,
		\end{equation}
		subject to
		\begin{equation}\label{Eqn Equivalent Knapsack Problem Capacity Constraint}
		\sum\limits_{i\, \in \, [N] } c_{i}\xi_{i} \leq \sum\limits_{i\, \in \, [N]}c_{i} -  D,
		\end{equation}
		\begin{align}\label{Eqn Equivalent Knapsack Problem Choice Constraint}
		& \xi_{i} \in \{0,1\}, &
		& \text{for all } i \,\in \,[N] .
		\end{align}	
	\end{subequations}
	%
	%
\end{problem}
\begin{proposition}\label{Thm Equivalence Knapsack Integer}
	Let $ \boldsymbol{\xi} = \big(\xi_{i}: i\in [N] \big) \in \{0, 1\}^{N} $ be a solution to \textsc{Problem} \ref{Pblm Equivalent Knapsack Problem} and define $ x_{i} \defining 1 - \xi_{i} $ for all $ i\in [N] $ then, the vector $ \x = \big(x_{i}: i\in [N] \big)  \in \{0, 1\}^{N} $ is a solution to \textsc{Problem} \ref{Pblm Integer Problem}.
\end{proposition}
\begin{proof}
	The proof uses the well-known classic transformation of complementary binary variables, $ x_{i} = 1 - \xi_{i} \in \{0, 1\} $ for all $ i\in [N] $, to relate the problems \ref{Pblm Integer Problem} and \ref{Pblm Equivalent Knapsack Problem} (see Section 13.3.3 in \cite{kellerer2005knapsack} for details).
	\qed
\end{proof}
%
%
%
%
\subsection{Greedy Algorithm vs Linear Optimization Relaxation}
%
%
In this section, we explore the relationship between the solution of the natural linear relaxation of \textsc{Problem} \ref{Pblm Integer Problem} and the solution provided by the natural Greedy Algorithm. First we introduce the following definitions
\begin{definition}\label{Def Specific Weight}
	Let $ \capac= \big(c_{i}:i \in [N] \big) $,  $ \p = \big(p_{i}:i \in [N] \big) $ be the lists of capacities and prices 
	respectively, we define the list of \textbf{specific weights} by
	\begin{align}\label{Eqn Specific Weight}
	& \gamma_{i} \defining \frac{c_{i}}{p_{i}}\, , &
	& \text{for all } i \in [N] .
	\end{align}
\end{definition}
Consider now the Greedy Algorithm \ref{Alg Greedy Algorithm} to find a feasible solution to \textsc{Problem} \ref{Pblm Integer Problem}. 
%
%
\begin{algorithm} 
	\caption{Greedy Algorithm, returns feasible solution to \textsc{Problem} \ref{Pblm Integer Problem} $ \big(y_{i}: i \in [N]\big) $ and corresponding value of objective function $ \sum\big\{ p_{i}y_{i} : i\, \in \, [N] \big\}$ }
	\label{Alg Greedy Algorithm}
	\begin{algorithmic}[1]
		\Procedure{Greedy Algorithm}{Prices: $ \p = \big(p_{i}: i \in [N]\big) $,
			Capacities: $ \capac= \big(c_{i}: i \in [N]\big) $,
			Demand: $ D $}
		
		
		\If{$ \sum\limits_{i\, \in \, [N] } c_{i} < D $}
		\textbf{print} ``Feasible region is empty" \Comment{Checking if the problem is infeasible}
		\Else
		
		\State \textbf{compute} list of specific weights $ \big(\gamma_{i}: i \in [N]\big) $ \Comment{Introduced in \textsc{Definition} \ref{Def Specific Weight}.}
		
		\State \textbf{sort} the list $ \big(\gamma_{i}: i \in [N]\big) $ in descending order
		
		\State \textbf{denote} by $ \sigma \in \S[N] $ the associated ordering permutation, i.e., 
		\begin{align}\label{Eqn Sorting Permutation}
		& \gamma_{\sigma(i)} \geq \gamma_{\sigma(i + 1)}, &
		& \text{for all } i \in [N - 1] .
		\end{align}

		\State $ y_{i} = 0 $ for all $ i \in [N] $, $ capacity = 0 $ \Comment{Initializing feasible solution and capacity}
		
		%
		
		\State i = 1
		\While{$ capacity \geq D $}{
			$ y_{\sigma(i) } = 1 $, \;
			$ capacity = capacity +  c_{\sigma(i) } $,\;	
			$ i = i + 1 $\;		
		}
		\EndWhile
		
		\State \Return 
		
		$ \big(y_{i}: i \in [N]\big) $, 
		$ \sum\limits_{i\, \in \, [N] } p_{i}y_{i} $ 
		\Comment{Feasible solution and corresponding value of objective function}
		
		%
		
		\EndIf
		\EndProcedure
	\end{algorithmic}
\end{algorithm}
Observe that due to the condition ($ \sum\limits_{i\, \in \, [N] } c_{i} \geq D $) for the loop to start, it will stop after a finite number of iterations. Next we introduce 
\begin{definition}\label{Pblm Natural LOP Problem}
	The natural linear relaxation of \textsc{Problem} \ref{Pblm Integer Problem}, is given by 
	\begin{problem}[Linear Relaxation, 0-1 Minimization KP]\label{Pblm Integer Problem LOP Relaxation}
		%
		\begin{subequations}\label{Eqn Integer Problem LOP Relaxation}
			\begin{equation}\label{Eqn Integer Problem Objective Function LOP}
			\min \sum\limits_{i\, \in \, [N] } p_{i}\xi_{i} ,
			\end{equation}
			subject to
			\begin{equation}\label{Eqn Integer Problem Capacity Constraint LOP}
			\sum\limits_{i\, \in \, [N] } c_{i}\xi_{i} \geq D,
			\end{equation}
			\begin{align}\label{Eqn Integer Problem Choice Constraint LOP}
			& 0 \leq \xi_{i} \leq 1, &
			& \text{for all } i \,\in \,[N] ,
			\end{align}	
		\end{subequations}
		i.e., the decision variables $ \big(\xi_{i}: i \in [N] \big) $ are are now real-valued. 
	\end{problem}
\end{definition}
Next we introduce a convenient notation and recall a classic result
\begin{definition}\label{Def Zero One sets}
	Let $ \boldsymbol{\xi} = \big(\xi_{i}:i \in [N]\big) $ be a solution of  Problem \ref{Pblm Integer Problem LOP Relaxation} 
	define the index sets
	\begin{align}\label{Eqn Zero One sets}
	& P \defining \big\{i \in [N]: \xi_{i} > 0 \big\}, &
	& Z \defining \big\{i \in [N]: \xi_{i} = 0 \big\} .
	\end{align}
	Define its \textbf{associated integer solution}
	$ \x^{\boldsymbol{\xi}} = \big(x_{i}^{\boldsymbol{\xi}}:i \in [N]\big) $ by
	\begin{equation}\label{Eqn From Simplex to Greedy Solution}
	x_{i}^{\boldsymbol{\xi}} \defining \begin{cases}
	1 & i \in P, \\
	0 & i \in  Z .
	\end{cases}
	\end{equation}
\end{definition}
\begin{theorem}\label{Thm Simplex and Greedy}
	Let $ \boldsymbol{\xi} = \big(\xi_{i}:i \in [N]\big) $ be an optimal solution of Problem \ref{Pblm Integer Problem LOP Relaxation} 
	and let $ \x^{\boldsymbol{\xi}} $ be as in \textsc{Definition} \ref{Def Zero One sets} above.
	%
	Then, $ \x^{\boldsymbol{\xi}} = \big(x_{i}^{\boldsymbol{\xi}}:i \in [N]\big) $ is the solution furnished by the \textsc{Greedy Algorithm} \ref{Alg Greedy Algorithm}.
\end{theorem}
\begin{proof}
	See Theorem 2.1 in \cite{kellerer2005knapsack}.
	\qed
\end{proof}
\begin{remark}\label{Rem Simplex and Greedy}
	It is important to stress that the \textsc{Greedy Algorithm} \ref{Alg Greedy Algorithm} may not produce an optimal solution as the following example shows. Consider \textsc{Problem} \ref{Pblm Integer Problem} for $ D = 40 $ and the following data.\\
	\begin{table}[h!]
		\begin{centering}
			\rowcolors{2}{gray!25}{white}
			\begin{tabular}{cccc}
				\hline
				\rowcolor{gray!50}
				Item & Capacity: $ \capac$ & Price: $ \p $ & Specific Weigh: $ \boldsymbol{\gamma} $ \\
				\hline
				1 & 100 & 4  & 25 \\
				2 & 40 & 2 & 20 \\
				\hline
			\end{tabular}
			\caption{\textsc{Remark} \ref{Rem Simplex and Greedy} Data}
			\label{Tbl Greedy Algorithm Counterexample}
		\end{centering}
	\end{table}
	
	\noindent Clearly, the \textsc{Greedy Algorithm} \ref{Alg Greedy Algorithm} would choose the solution $ \x = \big( 1, 0\big) $ with $ \p\cdot \x = 4 $ while $ \y = \big(0, 1 \big) $ gives $ \p\cdot \y = 2 $ and $ \x $ is not optimal. Moreover, the linear relaxation of this problem would yield $ \boldsymbol{\xi} =  \big( 0.4, 0\big) $ with $ \p\cdot \boldsymbol{\xi} = 1.6 $ and associated integer solution $ \x^{\boldsymbol{\xi}} = \big( 1, 0\big) $ i.e, the solution produced by the \textsc{Greedy Algorithm} \ref{Alg Greedy Algorithm}. 
	%
\end{remark}
%
%
%
%
%
\subsection{Introducing a Price-Capacity Rate}
%
%
In the sequel we adopt a relationship between capacities $ \capac$ and prices $ \p $ as it usually holds in real life scenarios.
\begin{definition}[Rate Price Capacity]\label{Def Rate Proctors Students}
	Let $r\in \big[ 1 ,  \max\limits_{i} c_{i} \big] $ be a fixed price-capacity increase threshold, then
	\begin{align}\label{Eq Rate Proctors Students}
	& p_{i} \defining \Big\lceil \frac{c_{i}}{r} \Big\rceil , &
	& \text{for all } i \in [N] .
	\end{align}
	In the following, we refer to $ r $ as the \textbf{price-capacity rate}. 
\end{definition}
Next we recall the main result of this part
\begin{theorem}\label{Thm Quality of the Greedy Algorithm}
	Let $ \big( c_{i} : i \in [N] \big) $ be a given list of capacities, let the list of prices $ \big( p_{i} : i \in [N] \big) $ be computed by the map \eqref{Eq Rate Proctors Students} and let $ r $ be the price-capacity rate introduced in \textsc{Definition} \ref{Def Rate Proctors Students}. 
	\begin{enumerate}[(i)]
		\item The Greedy Algorithm \ref{Alg Greedy Algorithm} produces the exact solution for $ r \geq \max\limits_{i}c_{i} $.
		
		
		\item Let $ r \vert c_{i} $ for all $ i\in [N] $ (i.e, a common divisor of all the capacities). Then, the effectiveness of the Greedy Algorithm \ref{Alg Greedy Algorithm} is entirely random.
		
		\item Let $ r \vert c_{i} $ for all $ i\in [N] $ (i.e, a common divisor of all the capacities). Let $ \boldsymbol{\xi} = \big(\xi_{i}:i \in [N]\big) $ be an optimal solution of Problem \ref{Pblm Integer Problem LOP Relaxation} furnished by the Simplex Algorithm and let $ \x^{\boldsymbol{\xi}} $ be as in \textsc{Definition} \ref{Def Zero One sets}. Then, $ \x^{\boldsymbol{\xi}}$ is a random element of the set
		\begin{equation}\label{Eq Minimal Cover}
		K \defining \Big\{ \x \in S: \sum\limits_{i\,\in\, A} c_{i}x_{i} < D, \;
		\forall\, A \subsetneq\{i\in [N]: x_{i} = 1\}  \Big\},
		\end{equation}
		where, $ S $ is the set of feasible solutions to \textsc{Problem}\ref{Pblm Integer Problem}, introduced in \textsc{Expression} \eqref{Eqn Feasible Set}.
		
	\end{enumerate}
\end{theorem}
\begin{proof}
	See \cite{kellerer2005knapsack}.
	\qed
\end{proof}
\begin{remark}\label{Rem Capacity as Greed Function}
	\begin{enumerate}[(i)]
		\item Observe that if $ r \vert c_{i} $ for all $ i \in [N] $, the \textsc{Problem} \eqref{Pblm Integer Problem} becomes
		\begin{align*}
		& \min\frac{1}{r}\sum\limits_{i\,\in\,[N]} c_{i}x_{i} , & 
		& S = \Big\{ \x \in \{0,1\}^{N} : \sum\limits_{i\,\in\,[N]} c_{i}x_{i} \geq D\Big\} .
		\end{align*}
		Hence, it reduces to a problem of approximating and integer from above using an integer partition of $ \sum\limits_{i}c_{i} $ in $ N $ blocks. 
		
		\item Notice that if $ r = \dfrac{d}{q} $ with $ d $ a common divisor of the capacities, the conclusion of   \textsc{Theorem} \ref{Thm Quality of the Greedy Algorithm} part (ii) holds.
		
		\item Given that the Greedy Algorithm effectiveness becomes entirely random when $ r $ is a common divisor of the capacities, we would like to use another criterion to distinguish the eligible items. To this end, the only possibility is to sort them according to its capacities. However, using the capacity as Greed function may not produce the exact solution as the Greedy Algorithm produced for the case $ r\geq \max\limits_{i} c_{i} $. Consider the following example 
		\begin{table}[h!]
			\begin{centering}
				\rowcolors{2}{gray!25}{white}
				\begin{tabular}{c c c c c c}
					\hline
					\rowcolor{gray!50}
					Item & Capacity & \multicolumn{2}{ c }{$ D = 11 $}  & \multicolumn{2}{ c }{$ D = 15 $}  \\
					\rowcolor{gray!50}
					& $ \capac$ & Greedy Decreasing & Optimal & Greedy Increasing & Optimal\\
					\hline
					1 & 10 & 1 & 0 & 0 & 0 \\
					2 & 9  & 1 & 0 & 0 & 0 \\
					3 & 8  & 0 & 0 & 1 & 1 \\
					4 & 7  & 0 & 1 & 1 & 1\\
					5 & 6  & 0 & 1 & 1 & 0 \\
					\hline
				\end{tabular}
				\caption{\textsc{Remark} \ref{Rem Capacity as Greed Function} Data}
				\label{Tbl Counterexample when Capacity is Greedy Function}
			\end{centering}
		\end{table}

	\end{enumerate}
\end{remark}
%
%
\section{A Divide \& Conquer Approach}
%
%
%
%
In the present section we introduce the Divide and Conquer method together with some theoretical results to assure the successful implementation of the method, from the algorithmic point of view. We begin with the following definition
\begin{definition}[Divide \& Conquer pairs and trees]\label{Def Divide and Conquer Setting}
	\begin{enumerate}[(i)]
		\item Let $ \capac= \big(c_{i}: i \in [N]\big) $ and $ \p = \big(p_{i}: i \in [N]\big) $ be the data associated to \textsc{Problem} \ref{Pblm Integer Problem}. A \textbf{subproblem} of \textsc{Problem} \ref{Pblm Integer Problem} is an integer problem with the following structure
		%
		\begin{align*} 
		& \min \sum\limits_{i\, \in \, A } p_{i}x_{i} , &
		& A\subseteq [N] ,
		\end{align*}
		subject to
		\begin{align*} 
		& \sum\limits_{i\, \in \, A } c_{i}x_{i} \geq D^{A},& 
		& D^{A}\leq D ,
		\end{align*}
		\begin{align*} 
		& x_{i} \in \{0,1\}, &
		& \text{for all } i \,\in \,A .
		\end{align*}	
		%
		
		\item Let $ (A^{0}, A^{1}) $ be a set partition of $ [N] $ and let $ (D^{0}, D^{1}) $ be an integer partition of $ D $ i.e., $ D = D^{0} + D^{1} $. 
		%
		We say a Divide and Conquer instance of \textsc{Problem} \ref{Pblm Integer Problem} is the pair of subproblems $ \big(\Pi^{b}: \, b \in \{0,1\} \big) $, defined by
		\begin{problem}[$ \Pi^{b}, b = 0, 1 $]\label{Pblm DC subproblems}
			%
			\begin{subequations}\label{Eqn Left-Right Integer Subproblem}
				\begin{equation}\label{Eqn Left-Right Integer Subproblem Objective Function}
				\min \sum\limits_{i\, \in \, A^{b} } p_{i}x_{i} ,
				\end{equation}
				subject to
				\begin{equation}\label{Eqn Left-Right Integer Subproblem Capacity Constraint}
				\sum\limits_{i\, \in \, A^{b} } c_{i}x_{i} \geq D^{b},
				\end{equation}
				\begin{align}\label{Eqn Left-Right Integer Subproblem Choice Constraint}
				& x_{i} \in \{0,1\}, &
				& \text{for all } i \,\in \,A^{b} .
				\end{align}	
			\end{subequations}
		\end{problem}
		In the sequel, we refer to $ \big(\Pi^{b}, b = 0, 1 \big) $ as a \textbf{D\&C pair}. Defining
		\begin{align*}
		& c^{b}_{i} \defining \begin{cases}
		c_{i} & i\in A^{b},\\
		0 & i \notin A^{b} , 
		\end{cases} &
		& p^{b}_{i} \defining \begin{cases}
		p_{i} & i\in A^{b},\\
		0 & i \notin A^{b} , 
		\end{cases}
		\end{align*}
		the corresponding feasible sets and the D\&C pair can be respectively written
		\begin{align}\label{Pblm Concise DC subproblems}
		z^{b}_{*} \defining\min\big\{\p^{b}\cdot \y: \y\in S^{b} \} &
		& \text{with} &
		& S
		^{b} \defining \big\{ \y\in \{0,1\}^{N}: \capac^{b} \cdot \y \geq D^{b} \big\} , &
		\end{align}
		where $ z^{b}_{*} $ denotes the optimal solution value of the problem $ \Pi^{b} $.	
		
		\item 
		A \textbf{D\&C tree} (see \textsc{Figures} \ref{Fig Tree Generated by Head-Left Algorithm}, \ref{Fig Tree Generated by Head-Left Algorithm Biased} and \ref{Fig Tree Generated by Balanced Left-Right Algorithm} below) for \textsc{Problem} \ref{Pblm Integer Problem} is a binary tree satisfying the following 
		\begin{enumerate}[(a)]
			\item Every \textbf{vertex} of the tree is in \textbf{bijective correspondence} with a \textbf{subproblem} of \textsc{Problem} \ref{Pblm Integer Problem}.
			
			\item The \textbf{root} of the tree is associated with \textsc{Problem} \ref{Pblm Integer Problem} itself.
			
			\item Every \textbf{internal vertex} $ V $ (which is not a leave) has a \textbf{left} and \textbf{right children}, $ V_{l}, V_{r} $ respectively, whose associated subproblems make a D\&C pair for the subproblem associated to $ V $. 
		\end{enumerate} 
		%
		%
		%
	\end{enumerate}
\end{definition}
\begin{remark}\label{Rem D&C trees' vertex and subproblems identification}
	\begin{enumerate}[(i)]
		\item Observe that, due to property (iii) a D\&C tree is, in particular, a \textbf{complete binary tree} (see \cite{GrossYellen} pg 127). 
		
		\item In the same way that in the knapsack problem the eligible items are identified with their corresponding labels $ j = 1, \ldots, N $, from now on, in order to ease notation, we \textbf{identify} every vertex of a D\&C tree with its associated subproblem. More specifically, a vertex/node of a D\&C tree will also act as the label of a subproblem of \textsc{Problem} \ref{Pblm Integer Problem}. Given that the vertex-subproblem assignment is a bijective map, such identification will introduce no confusion, 
		see \textsc{Figure} \ref{Fig Tree Generated by Head-Left Algorithm}-\textsc{Table} \ref{Tbl Tree Generated by Head-Left Algorithm} and \textsc{Figure} \ref{Fig Tree Generated by Head-Left Algorithm Biased}-\textsc{Table} \ref{Tbl Balanced Tree for Particular Example} for concrete examples; see also \textsc{Figure} \ref{Fig Tree Generated by Balanced Left-Right Algorithm} below.
	\end{enumerate}
\end{remark}
\begin{theorem}\label{Thm DC feasibility}
	Suppose that \textsc{Problem} \ref{Pblm Integer Problem} is feasible, then
	\begin{enumerate}[(i)]
		\item A feasible solution $ \y $ of \textsc{Problem} \ref{Pblm Integer Problem} can be infeasible for at most one problem of the D\&C pair. 
		
		\item At most one problem of the D\&C pair is infeasible.
		
		\item Let $ \big(A^{b}: b \in \{0,1\}\big) $ be a fixed partition of $ [N] $ then, both \textsc{Problems} \ref{Pblm DC subproblems}, $ \big(\Pi^{b}: b \in \{0,1\} \big) $ are feasible if and only if 
		\begin{align}\label{Stmt Distribution Feasibily DC}
		& D - \sum\limits_{ i\, \in \, A^{1 - b} } c_{i} \leq
		D^{b} \leq \sum\limits_{ i\, \in \, A^{b} } c_{i}, & 
		& \text{for } b = 0, 1.
		\end{align} 
		\item Let $ \big(A^{b}: b \in \{0,1\}\big) $ be a fixed partition of $ [N] $ and define
		\begin{align}\label{Eqn Demand Fraction}
		& D^{0} \defining  \bigg\lfloor \frac{D}{\sum\{c_{i}: i \in [N]\} }
		\sum\limits_{i \, \in \, A^{0}}c_{i}\bigg\rfloor , &
		& D^{1} \defining D - D^{0} .
		\end{align}
		Then, if 
		\begin{equation}\label{Ineq  Control on the Complement}
		\frac{ D }{\sum\{c_{i}: i \in [N] \}  } \, \sum\limits_{i\,\in\, A^{1}}c_{i}+ 1 
		\leq \sum\limits_{i\,\in\, A^{1}}c_{i} \, , 
		\end{equation}
		both \textsc{Problems} \ref{Pblm DC subproblems}, $ \big(\Pi^{b}: b \in \{0,1\} \big) $ are feasible.
		
		\item The following inclusions for the feasible sets $ S^{0}, S^{1}, S $ hold
		\begin{align}\label{Eqn Feasible Set Inclusions}
		& S \subseteq S^{0} \cup S^{1} , &
		& S^{0} \cap S^{1} \subseteq S .
		\end{align}
	\end{enumerate}
\end{theorem}
\begin{proof}
	\begin{enumerate}[(i)]
		\item Let $ \y $ be a feasible solution of \textsc{Problem} \ref{Pblm Integer Problem}, then $ \sum\limits_{i\,\in\, [N]} c_{i}y_{i}  \geq D $; equivalently
		\begin{equation*}
		\sum_{b \, \in \, \{0, 1\}}\sum\limits_{i\,\in\, A^{b}} c_{i}y_{i} \geq D^{0} + D^{1} .
		\end{equation*}
		Hence, if $ \y $ is $ \Pi^{b} $-infeasible we have $ \sum\limits_{i\,\in\, A^{b}} c_{i}y_{i} < D^{b}  $ and the expression above writes
		\begin{equation*}
		\sum\limits_{i\,\in\, A^{1 - b}} c_{i}y_{i} \geq D^{1 - b} + 
		D^{b} - \sum\limits_{i\,\in\, A^{b}} c_{i}y_{i} > D^{1 - b},
		\end{equation*}
		i.e., $ \y $ is $ \Pi^{1 - b} $-feasible. Since $ b \in \{0, 1\} $ was arbitrary, the claim of this part follows.
		
		\item Since \textsc{Problem} \ref{Pblm Integer Problem} is feasible, the vector $ \y \in \{0,1\}^{N} $ having all its entries equal to one is also feasible, due to the previous part the result follows.
		
		\item Fix $ b\in \{0,1\} $ arbitrary, then it is trivial to see that the second inequality in \eqref{Stmt Distribution Feasibily DC} is necessary and sufficient condition for the problem $ \Pi^{b} $ to be feasible, as well as the condition $ D^{1 - b} \leq \sum\limits_{ i\, \in \, A^{1 - b} } c_{i} $ is necessary and sufficient for $ \Pi^{1 - b} $ to be feasible. Recalling that $ D^{b} = D - D^{1 - b} $, the first inequality in \eqref{Stmt Distribution Feasibily DC} follows.
		
		\item Since \textsc{Problem} \ref{Pblm Integer Problem} is feasible then $ D \leq \sum\{c_{i}: i \in [N]\}  $, therefore
		\begin{equation*}
		D^{0} = \bigg\lfloor \frac{D}{\sum\{c_{i}: i \in [N]\} }\sum\limits_{i \, \in \, A^{0}}c_{i}\bigg\rfloor
		\leq \bigg\lfloor \sum\limits_{i \, \in \, A^{0}}c_{i}\bigg\rfloor =  \sum\limits_{i \, \in \, A^{0}}c_{i}, 
		\end{equation*}
		i.e., the problem $ \Pi^{0} $ is feasible. 
		On the other hand,
		\begin{equation*}
		D^{0} = \bigg\lfloor \frac{D}{\sum\{c_{i}: i \in [N]\} }\sum\limits_{i \, \in \, A^{0}}c_{i}\bigg\rfloor
		\geq \frac{D}{\sum\{c_{i}: i \in [N]\} }\sum\limits_{i \, \in \, A^{0}}c_{i} - 1 .
		\end{equation*}
		Since $ D^{1} = D - D^{0} $, we have
		\begin{align*}
		D^{1} & \leq D - \frac{D}{\sum\{c_{i}: i \in [N]\} }\sum\limits_{i \, \in \, A^{0}}c_{i} + 1 \\
		& = \frac{D}{\sum\{c_{i}: i \in [N]\} } \sum\limits_{i \, \in \, A^{1}}c_{i} + 1 \\
		& \leq \sum\limits_{i\, \in\, A^{1}}c_{i} ,
		\end{align*}
		where the last bound holds due to \textsc{Inequality} \ref{Ineq  Control on the Complement}.
		Hence, the problem $ \Pi^{1} $ is also feasible.
		
		\item Due to the first part, if $ \y \in S $ then it must be $ \Pi^{0} $ or $ \Pi^{1} $-feasible. Equivalently, it belongs to $ S ^{0} $ or $ S^{1} $, i.e. $ \y \in S^{0} \cup S^{1} $.
		
		Finally, if $ \y \in S^{0} \cap S^{1} $ then $ \sum\limits_{i\,\in\, A^{b}} c_{i}y_{i} \geq D^{b} $ for $ b = 0, 1 $. Adding both inequalities yields
		\begin{equation*}
		\capac\cdot \y = \sum\limits_{i\,\in\, [N]} c_{i}y_{i} =
		\sum\limits_{i\,\in\, A^{0}} c_{i}y_{i}  
		+\sum\limits_{i\,\in\, A^{1}} c_{i}y_{i} \geq 
		D^{0} + D^{1} = D,
		\end{equation*}
		i.e., $ y $ belongs to the set $ S $ and the proof is complete.
		\qed
	\end{enumerate}
\end{proof}
\begin{remark}\label{Rem Demand Computation}
	Observe that \textsc{Inequality} \ref{Ineq  Control on the Complement} in (iv) from \textsc{Theorem} \ref{Thm DC feasibility}, is a mild hypothesis. It is equivalent to
	\begin{equation}\label{Ineq Equivalent Control on the Complement}
	D
	\leq \frac{\sum \{c_{i}: i\,\in\, A^{1}\} - 1}{\sum \{c_{i}: i\,\in\, A^{1}\}} 
	\sum\limits_{i\, \in\, [N]} c_{i} ,
	\end{equation}
	i.e., \textsc{Inequality} \ref{Ineq  Control on the Complement} demands a reasonable slack $  \sum\limits_{i\, \in\, [N]} c_{i} - D $  between total capacity and demand.
\end{remark}
%
%
\begin{proposition}\label{Thm Feasibility of DC}
	Let $ \capac= \big(c_{i}: i \in [N]\big) $, $ \p = \big(p_{i}: i \in [N]\big) $ be the data associated to \textsc{Problem} \ref{Pblm Integer Problem} and let $ \big(\Pi^{0}, \Pi^{1}\big) $ be a D\&C pair. 
	\begin{enumerate}[(i)]
		
		\item Let $ \x $ be an optimal solution to \textsc{Problem} \ref{Pblm Integer Problem} and let $ \y^{0}, \y^{1} $ be optimal solutions to \textsc{Problems} \ref{Pblm DC subproblems} $ \Pi^{0}, \Pi^{1} $-respectively. Then
		\begin{equation}\label{Ineq DC optimal pair}
		z_{*} = \sum\limits_{i\,\in \, [N]} p_{i} x_{i} 
		\leq \sum\limits_{j\,\in \, A^{0}} p_{j} y_{j}^{0} + 
		\sum\limits_{j\,\in \, A^{1}} p_{j} y_{j}^{1} = z^{0}_{*} + z^{1}_{*},   
		\end{equation}
		where $ z_{*}, z^{0}_{*}, z^{1}_{*} $ denote the optimal solution values for the problems \ref{Pblm Integer Problem}, $ \Pi^{0} $ and $ \Pi^{1} $ respectively. 
		
		\item Let $ \x $ be an optimal solution to \textsc{Problem} \ref{Pblm Integer Problem} which is both $ \Pi^{b} $ and $ \Pi^{1 - b} $-feasible, then $ \x $ is a $ \Pi^{b} $ and $ \Pi^{1-b} $-optimal solution.
	\end{enumerate}
\end{proposition}
\begin{proof}
	\begin{enumerate}[(i)]
		\item Since $ \y^{b} $ is an optimal solution of $ \Pi^{b} $, define the vector $ \y \in \{0,1\}^{N} $ by
		\begin{equation*}
		y_{i} \defining \begin{cases}
		y^{0}_{i} & i \in A^{0}, \\
		y^{1}_{i} & i \in A^{1}.
		\end{cases}
		\end{equation*}
		Then, $ \p\cdot \y =  \sum\limits_{j\,\in \, A^{0}} p_{j} y_{j}^{0} + 
		\sum\limits_{j\,\in \, A^{1}} p_{j} y_{j}^{1} $ and $ y $ is both $ \Pi^{0} $ and $ \Pi^{1} $-feasible i.e., $ \y \in S^{0}\cap S^{1} $. Recalling the feasible sets inclusion \eqref{Eqn Feasible Set Inclusions} and that $ \x $ is optimal, we have $ \p\cdot \x = \min \big\{\p \cdot \boldsymbol{\xi}: \boldsymbol{\xi}\in S \big\} \leq \p\cdot \y $ i.e., the result follows.
		\item Let $ \x $ be an optimal solution to \textsc{Problem} \ref{Pblm Integer Problem} which is also $ \Pi^{b} $-feasible for $ b \in \{0,1\} $ fixed. Suppose that $ \x $ is not an optimal solution of \textsc{Problem} $ \Pi^{b} $ and let $ \y^{b} $ be its optimal solution, therefore $ \sum\limits_{j\,\in \, A^{b}} p_{j} y_{j}^{b} < \sum\limits_{j\,\in \, A^{b}} p_{j} x_{j}^{b} $. Define $ \y \in \{0,1\}^{N} $ by
		\begin{equation*}
		y_{i} \defining \begin{cases}
		y^{b}_{i} & i \in A^{b} , \\
		x_{i} & i \in A^{1 - b} .
		\end{cases}
		\end{equation*}
		Observe that 
		\begin{equation*}
		\capac\cdot \y = 
		\sum\limits_{j\,\in \, A^{b}} c_{j} y_{j}^{b} + 
		\sum\limits_{j\,\in \, A^{1 - b}} c_{j} x_{j} \geq
		D^{b} + D^{1 - b} .
		\end{equation*}
		Here, the inequality holds because $ \y^{b} $ is $ \Pi^{b} $-feasible and $ \x $ is $ \Pi^{1 - b} $-feasible. Therefore, $ \y $ is feasible for \textsc{Problem} \ref{Pblm Integer Problem}; but then
		\begin{equation*}
		\p\cdot \y = 
		\sum\limits_{j\,\in \, A^{b}} p_{j} y_{j}^{b} + 
		\sum\limits_{j\,\in \, A^{1 - b}} p_{j} x_{j} < 
		\sum\limits_{j\,\in \, A^{b}} p_{j} x_{j} + 
		\sum\limits_{j\,\in \, A^{1 - b}} p_{j} x_{j} =
		\p\cdot \x
		\end{equation*}
		and $ \x $ would not be an optimal solution, which is a contradiction. Since the above holds for any $ b \in \{0, 1\} $ the proof is complete.	
		\qed
	\end{enumerate}
\end{proof}
\begin{remark}\label{Rem Subproblems Optimal Solutions}
	Notice that in \textsc{Proposition} \ref{Thm Feasibility of DC} (ii) the hypothesis requiring the optimal solution $ \x $ being both $ \Pi^{b} $ and $ \Pi^{1 - b} $-feasible can not be relaxed as the following example shows. Consider \textsc{Problem} \ref{Pblm Integer Problem} for $ D = 150 $ and the following data
	\begin{table}[h!]
		\begin{center}
			\rowcolors{2}{gray!25}{white}
			\begin{tabular}{ccc}
				\hline
				\rowcolor{gray!50}
				Item & Capacity: $ \capac$ & Price: $ \p $\\
				\hline
				1 & 100 & 2 \\
				2 & 50 & 1 \\
				3 & 100 & 2 \\
				4 & 50 & 1 \\
				\hline
			\end{tabular}\caption{\textsc{Remark} \ref{Rem Subproblems Optimal Solutions} Data}
		\end{center}
	\end{table}

	\noindent An optimal solution is given by $ \x = \big( 1, 1, 0, 0\big) $ with $ \p\cdot\x = 3 $. Consider $ A^{0} \defining \{1, 2 \} $,  $ A^{1} \defining \{ 3, 4 \} $ with $ D^{0} = 50, D^{1} = 100 $. Then, $ \x $ is $ \Pi^{0} $-feasible but it is not $ \Pi^{1} $-feasible, moreover $ \x $ is not $ \Pi^{0} $-optimal because $ \y = \big(0, 1, 1, 0\big) $ is $ \Pi^{0} $-feasible and 
	\begin{equation*}
	\sum\limits_{ i\, \in\, A^{0} } p_{i}y_{i} =
	1 < 3 =
	\sum\limits_{ i\, \in\, A^{0} } p_{i}x_{i}.		
	\end{equation*}
	Consequently, the optimal solution has to be both $ \Pi^{0} $, $ \Pi^{1} $-feasible to guarantee that \textsc{Proposition} \ref{Thm Feasibility of DC} (ii) holds. \\
	
	\noindent On the other hand if we take the previous setting but replacing $ D^{0} = 60 $, $ D^{1} = 90 $, then $ \y^{0} = \big( 1, 0, 0, 0\big) $, $ \y^{1} = \big( 0, 0, 1, 0\big) $ are $ \Pi^{0} $ and $ \Pi^{1} $ optimal solutions, however
	\begin{equation*}
	\sum\limits_{ i\, \in\, A^{0} } p_{i}x_{i} =  
	2 + 1 < 
	2 + 2 =
	\sum\limits_{ i\, \in\, A^{0} } p_{i}y_{i}^{0} +
	\sum\limits_{ i\, \in\, A^{1} } p_{i}y_{i}^{1} ,		
	\end{equation*}
	i.e., a global optimal solution can not be derived from the local solutions of the D\&C pair.	
	Finally, if we choose $ A^{0} = \{ 1, 2, 3 \} $, $ A^{1} = \{ 4 \} $, $ D^{0} = 90 $, $ D^{1} = 60 $, the problem $ \Pi^{1} $ is not feasible.
\end{remark}
\begin{remark}\label{Rem Computational Time}
	The introduction of a D\&C pair is of course aimed to reduce the computational complexity of the \textsc{Problem} \ref{Pblm Integer Problem} given that \textsc{Problem} \ref{Pblm DC subproblems} $ \big(\Pi^{b}\big) $ can be regarded as a problem in $ \{0,1\}^{\vert A^{b} \vert} $ with $ b\in \{0, 1\} $, instead of a problem in  $ \{0,1\}^{N} $, which reduces the order of complexity (see \textsc{Section} \ref{Sec Computational Time} for details). However, from the discussion above, it follows that the choice of $ D^{0}, D^{1} $ is crucial when designing the pair $ \big( \Pi^{0}, \Pi^{1}\big) $. Ideally, \textsc{Inequality} \eqref{Ineq DC optimal pair} would be an equality for the optimal solutions $ \x $, $ \y^{0} $, $ \y^{1} $, this observation motivates the definition \ref{Def DC efficiency} introduced below.
\end{remark}
\begin{definition}\label{Def DC efficiency}
	Let $ \capac= \big(c_{i}: i \in [N]\big) $ and $ \p = \big(p_{i}: i \in [N]\big) $ be the data associated to \textsc{Problem} \ref{Pblm Integer Problem}. Let $ \big(A^{b}: b \in \{0,1\}\big) $,  $ \big(D^{b}: b \in \{0,1\}\big) $ be partitions of $ [N] $ and $ D $ respectively
	\begin{enumerate}[(i)]		
		\item We say the demands are \textbf{partition-dependent} if both satisfy the relationship \eqref{Eqn Demand Fraction} and we denote this dependence by
		\begin{align}\label{Def Fraction Demands Notation}
		& D^{b} = D^{b}(A^{0}, \, A^{1}), & 
		& b = 0, 1.
		\end{align}
		\item The D\&C pair $ \big(\Pi^{b}: b\in \{0,1\}\big) $ is said to be a \textbf{feasible pair} if both \textsc{Problems} \ref{Pblm DC subproblems} are feasible.
		
		\item If the D\&C pair $ \big(\Pi^{b}: b\in \{0,1\}\big) $ is feasible we define its \textbf{efficiency} as
		\begin{equation}\label{Eqn DC efficiency}
		\eff\big((A^{0}, A^{1}), (D^{0}, D^{1}) \big) \defining 100 \times
		\frac{z_{*}^{0} + z_{*}^{1} - z_{*}}{ z_{*} },
		\end{equation}
		where $ z_{*}, z^{0}_{*}, z^{1}_{*} $ denote the optimal solution values for the problems \ref{Pblm Integer Problem}, $ \Pi^{0} $ and $ \Pi^{1} $ respectively. 
		
		\item Given $ \A \defining \big(A^{j}: j \in [J] \big) $ and $ \D \defining \big(D^{j}: j \in [J] \big) $ be partitions of $ [N] $ and $ D $ respectively such that $ \Pi^{j} $ is feasible for all $ j \in [J] $, then, the its associated efficiency is defined by 
		\begin{equation}\label{Eqn Partition efficiency}
		\eff\big(\A, \D \big) \defining 100\times
		\frac{ \sum \{z_{*}^{j}: j\, \in\, [J]\} - z_{*}}{z_{*} } ,
		\end{equation}
		where $ z_{*} $ is the optimal solution value for \textsc{Problem} \ref{Pblm Integer Problem} and $ z_{*}^{j} $ indicates the optimal solution value for the subproblem  analogous to \textsc{Problem} \ref{Pblm Integer Problem}, whose input data are the demand $ D_{j} $ and items $ A^{j} $. 
	\end{enumerate}
\end{definition}
\begin{remark}
	Notice that for any feasible pair, it holds that $ 0 < \eff\big((A^{0}, A^{1}), (D^{0}, D^{1}) \big) $, due to \textsc{Inequality} \eqref{Ineq DC optimal pair}. 
	Additionally, the notion of efficiency that we are defining is nothing but the relative error introduced by the D\&C approximation of the solution. Finally, for general partitions $ \A $ and $ \D $, 
	an inequality analogous to \eqref{Ineq DC optimal pair} can be derived using induction on the cardinal of $ \A $.  
\end{remark}
Before introducing the definition of efficiency for D\&C trees we recall a classic definition from Graph Theory (see \textit{Section 2.3} in \cite{GrossYellen})
\begin{definition}\label{Def Subtree Definition}
	Let $ T = (V, E) $ be a tree and let $ U \subseteq V $ be a subset of vertices. The \textbf{subtree induced} on $ U $, denoted by $ T(U) $, is the tree whose vertices are $ U $ and whose edge-set consists on all those edges in $ E $ such that both endpoints are contained in $ U $.
	
	%
\end{definition}
\begin{definition}\label{Def DC Tree efficiency}
	Let $ \capac= \big(c_{i}: i \in [N]\big) $, $ \p = \big(p_{i}: i \in [N]\big) $ be the data associated to \textsc{Problem} \ref{Pblm Integer Problem}, let $ DCT $ be a D\&C tree associated. Let $ H $ be the height and $ V_{0} $ be the root of the DCT tree, where $ V_{0} $ is associated to the original problem \ref{Pblm Integer Problem} itself.
	\begin{enumerate}[(i)]
		\item The tree $ DCT $ is said to be \textbf{feasible} if all its nodes are feasible problems.
		
		\item Let  $ h \in [H] $ arbitrary, the \textbf{tree pruned} at height $ h $ is given by
		\begin{multline}\label{Eqn Pruned Tree}
		DCT_{h} = \text{subtree of } DCT \text{ induced on the set }\\
		\{V \text{vertex of } DCT: \text{ height} (V) \leq h\} .
		\end{multline}
		We denote by $ L(DCT_{h}) $ the set of leaves of the tree $ DCT_{h} $ i.e., those vertices whose degree is equal to one.
		
		\item We say that a set of leaves $ L(DCT_{h}) $ for a given $ h \in [H] $ is an \textbf{instance} of the D\&C approach applied to the problem \ref{Pblm Integer Problem}.   
		
		\item 
		Let $ DCT $ be feasible with $ H$, the \textbf{global} and \textbf{stepwise efficiencies} of the tree are defined by
		\begin{subequations}
			\begin{align}\label{Eqn Global Efficiencies}
			& GbE(h) \defining 100 \times
			\frac{ \sum \{z_{*}^{V}: V \in L(DCT_{h}) \} - z_{*}^{V_{0}} }{ z_{*}^{V_{0}} }, &
			& h \in [H] .
			\end{align}
			\begin{multline}\label{Eqn Stepwise Efficiencies}
			SwE(h) \defining 100\times
			\frac{\sum \{ z_{*}^{V}: V \in L(DCT_{h}) \}- 
				\sum \{ z_{*}^{V}: V \in L(DCT_{h - 1}) \} }
			{\sum \{ z_{*}^{V}: V \in L(DCT_{h - 1}) \} }, \\
			h \in [H] - \{0\} .
			\end{multline}
			Here, $ z_{*}^{V} $ indicates the optimal solution of the problem associated to the vertex $ V $ in the $ DCT $ tree and $ L(DCT_{h}) $ stands for the set of leaves in the tree $ DCT_{h} $, (see \textsc{Figures} \ref{Fig Tree Generated by Head-Left Algorithm}, \ref{Fig Tree Generated by Head-Left Algorithm Biased} and \ref{Fig Tree Generated by Balanced Left-Right Algorithm} below). 
		\end{subequations}

		\item 
		Let $ DCT $ be feasible with $ H$, the \textbf{global} and \textbf{stepwise relative computational times} of the tree are defined by
		\begin{subequations}
			\begin{align}\label{Eqn Global Times}
			& GbT(h) \defining 100 \times
			\frac{ \sum \{t^{V}: V \in L(DCT_{h}) \} }{ t^{V_{0}} }, &
			& h \in [H] .
			\end{align}
			\begin{align}\label{Eqn Stepwise Times}
			& SwT(h) \defining 100\times
			\frac{\sum \{ t^{V}: V \in L(DCT_{h}) \} \} }
			{\sum \{ t^{V}: V \in L(DCT_{h - 1}) \} }, &
			h \in [H] - \{0\} .
			\end{align}
			Here, $ t^{V} $ indicates the absolute computational time needed for the solution of the vertex $ V $ in the $ DCT $ tree and $ L(DCT_{h}) $ stands for the set of leaves in the tree $ DCT_{h} $, (see \textsc{Figures} \ref{Fig Tree Generated by Head-Left Algorithm}, \ref{Fig Tree Generated by Head-Left Algorithm Biased} and \ref{Fig Tree Generated by Balanced Left-Right Algorithm} below). 
		\end{subequations}
	\end{enumerate}
\end{definition}
Next we prove that the definition \ref{Def DC Tree efficiency} above makes sense.
\begin{theorem}\label{Thm Partitions in Tree}
	Let $ DCT $ be a D\&C tree with height $ H $, root $ V_{0} $ and let $ DCT_{h} $, $ L(DCT_{h}) $ for $ h \in [H] $ be as in \textsc{Definition} \ref{Def DC Tree efficiency} (ii) above. Then, $ \{A^{V}: V\in L(DCT_{h}) \} $ is a partition of $ [N] $, where $ A^{V} $ is the set of eligible items for the subproblem associated to the node $ V $. 
\end{theorem}
\begin{proof}
	We proceed by induction on the height of the tree. For $ H = 0 $ the result is trivial and for $ H = 1 $ the tree merely consists of $ V_{0} $ and its left and right children $ V_{l}, V_{r} $ which by definition, are associated to a D\&C pair for \textsc{Problem} \ref{Pblm Integer Problem}; in particular, the sets $ A^{0} $, $ A^{1} $ are a partition of $ [N] $. Now assume that the result is true for $ H \leq  k $ and let $ DCT $ be such that its height is $ k + 1 $. Consider $ DCT_{k} $ and $ L(DCT_{k}) $, given that the result is true for heights less or equal than $ k $ we have that $ \{A^{V}: V\in L(DCT_{k})\} $ is a partition of $ [N] $. We classify this set as follows
	\begin{equation}\label{Eqn Leaves at heigh k}
	\{A^{V}: V\in L(DCT_{k})\} = 
	\{A^{V}: V\in L(DCT_{k})\cap L(DCT_{k + 1})\} \cup 
	\{A^{V}: V\in L(DCT_{k}) -  L(DCT_{k + 1})\} .
	\end{equation}
	However, if $ V  \in L(DCT_{k}) -  L(DCT_{k + 1}) $ it means that its left and right children $ V_{l}, V_{r} $ belong to $ L(DCT_{k + 1}) $. Moreover, since $ (V_{l}, V_{r}) $ are associated to a D\&C pair for the subproblem associated to $ V $, then $ (A^{V_{l}}, A^{V_{r}}) $ is a partition of $ A^{V} $, i.e., 
	\begin{multline}\label{Eqn Leaves at heigh k + 1}
	\{A^{V}: V\in L(DCT_{k}) -  L(DCT_{k + 1})\} = \\
	\{A^{V_{l}}: V_{l} \text{ left child of } V\in L(DCT_{k}) -  L(DCT_{k + 1})\} \cup\\
	\{A^{V_{r}}: V_{r} \text{ right child of } V\in L(DCT_{k}) -  L(DCT_{k + 1})\} .
	\end{multline}
	Putting together \textsc{Expressions} \eqref{Eqn Leaves at heigh k} and \eqref{Eqn Leaves at heigh k + 1} the result follows.
	\qed
\end{proof}
\begin{remark}\label{Rem Consistency of DCT efficiency definition}
	Clearly, due to \textsc{Theorem} \ref{Thm Partitions in Tree} a set of leaves $ L(DCT_{h}) $ for $ h \in [H] $ is a potential instance of the D\&C method applied to \textsc{Problem} \ref{Pblm Integer Problem} as the definition \ref{Def DC Tree efficiency} (iii) states. It is also direct to see that the global and stepwise efficiencies $ GbE(h) $ and $ SwE(h) $ respectively, introduced in (iv) \textsc{Definition} \ref{Def DC Tree efficiency} compute the ratios adding the solution values found for different partitions of the set of eligible items.
\end{remark}
In view of the previous discussion a natural question is how to choose D\&C efficiency-optimal pairs (at least for one step and not for a full D\&C tree) however, allowing complete independence between the pairs $ (A^{0}, A^{1}) $ and $ (D^{0}, D^{1}) $ i.e., the partitions of $ [N] $ and $ D $ respectively would introduce an overwhelmingly vast search space. 
Consequently, from now on, we limit our study to partition-dependent demands, see \textsc{Definition} \ref{Def DC efficiency} (i). \\

\noindent In the next section several ways to generate partitions $ (A^{0}, A^{1}) $ will be introduced, which will be regarded as \textbf{strategies} to implement the D\&C approach. However, other strategies will be explored, such as the price-capacity rate $ r $ and the demand-capacity fraction $ o \defining \frac{1}{D}\, \sum\{ c_{i}: i\,\in\,[N]\} $. These are related to the problem setting (availability of resources), rather than the choice of D\&C pairs. The assessment of all the aforementioned strategies will be done using Monte Carlo simulations, when the list of capacities is regarded as a random variable ($ \upc $ instead of $ \capac$) with known probabilistic distribution.

\section{Strategies and Heuristic Method}\label{Sec The Heuristic Method}
%
%
Since no theoretical results can be found so far for the Divide and Conquer method, its efficiency has to be determined empirically. To that end, numerical experiments will be conducted with randomly generated data, according to classical discrete distributions. Next, several strategies will be evaluated in these settings (see \textsc{Figure} \ref{Fig Branch Strategies Tree}). It is important to stress that the type of strategies, as well as their potential values (numerical in most of the cases) presented here, were chosen in order to simulate \textbf{plausible instances} of the initial problem rather than \textbf{arbitrary instances} of \textsc{Problem} \ref{Pblm Integer Problem}. 
%
%
%
\subsection{Random Setting}\label{Sec Random Setting}
%
%
A \textsc{Random Setting Algorithm} generates lists of eligible items according to certain parameters defined by the user, namely the number of items, the distribution of its capacities (Uniform, Poisson, Binomial) and the demand-capacity fraction $ o $; which will range between $ 0.5 $ and $  0.9 $, this will guarantee the hypotheses of (iv) \textsc{Theorem} \ref{Thm DC feasibility} are satisfied. If $ \upc $ denotes the random variable having the capacity of the eligible items, the code uses the following parameters for the distributions
\begin{enumerate}[(i)]
	\item \textbf{Uniform}. Range sizes $ [40, 120]\cap \N $ i.e., 
	$ \prob (\upc = n) = \dfrac{1}{80} $, for $ n \in [40, 120]\cap \N $.
	
	\item \textbf{Poisson}. Average $ \lambda = 65 $, then 
	$ \prob( \upc = n ) = \dfrac{1}{\exp(\lambda)} \,\dfrac{\lambda^{n}}{n!} $, for $ n \in \N \cup \{0\} $.
	\item \textbf{Binomial}. Sample space $ [480] $, success probability, $ p = 0.2 $, i.e., $ \prob (\upc = n) = { 480 \choose n} p^{n} (1 - p)^{480 - n} $, for $ n \in [0, 480] \cap \N  $.
\end{enumerate}
An example of 4 realizations, each consisting in 8 eligible items, uniformly distributed with demand-capacity fraction of 0.9 is displayed in \textsc{Table} \ref{Tbl Uniform Random Setting Example} below. The \textsc{Random Setting Algorithm} \ref{Alg Random Setting Algorithm}, produces a table analogous to \textsc{Table} \ref{Tbl Uniform Random Setting Example}
. 
\begin{algorithm} 
	\caption{Random Setting Algorithm}
	\label{Alg Random Setting Algorithm}
	\begin{algorithmic}[1]
		\Procedure{Random Setting}{Items' Number: $ n $, 
			Probabilistic Distribution: $ d $, 
			Demand: $ o $,
			Number of Trials: $ t $ }
		\Function{Random Generation}{ $ d, n $ }
		\If{ $ d = $ Uniform }
		\State $ items = $ list of $ n $ random items uniformly distributed on the interval $ [40, 120] $ \Comment{Python command: numpy.random.randint(low = 40, high = 120, size = $ n $)}  
		\ElsIf{ $ d = $ Poisson }
		\State $ items = $ list of $ n $ random items, Poisson distributed with average $ 65 $
		\Comment{Python command: numpy.random.poisson(65, size = $ n $)}  
		\Else
		\State $ items = $ list of $ n $ random items, binomially distributed on the interval $ [0, 480] $ with success probability $ p = 0.2 $
		\Comment{Python command: np.random.binomial(480, 0.2, $ n $)} 
		\EndIf
		\State \Return $ items $
		\EndFunction
		\State  $ Eligible\_Items =  \emptyset $ \Comment{Initialize the $ Eligible\_Items $ table}
		\For{$ trial \leq t $}
		\State \textsc{Random Generation}($ d, n $) $ \rightarrow  Eligible\_Items $\Comment{Push list of $ n $ randomly generated items as a column of the $ Eligible\_Items $ table.}
		\EndFor 
		\State $ last\_row = o\times \sum\limits_{i = 1}^{n} row_{i} $ with $ row_{i} $ $i - $th row of  $ Eligible\_Items $
		\Comment{Computing the demand fraction}
		\State $ last\_row\rightarrow Eligible\_Items $ \Comment{Push $ last\_row $ as the last row or $ Eligible\_Items $}
		\State Export $ Eligible\_Items $ \Comment{In this work, to the file \textbf{Eligible\_Items.xls}.}
		\EndProcedure
	\end{algorithmic}
\end{algorithm}
\begin{table}[h!]
	\begin{centering}
		\rowcolors{2}{gray!25}{white}
		\begin{tabular}{cccccc}
			\rowcolor{gray!50}
			\hline
			Item & Realization 1 & Realization 2 & Realization 3 & Realization 4 & Realization 5 
			\\
			\hline
			0 & 113	& 47    & 84	& 58	& 53 
			\\
			1 & 54	 & 67	 & 119	& 49	& 104 
			\\
			2 & 95	 & 65	 & 64	 & 109	& 119 
			\\
			3 & 89	 & 95	 & 91	 & 78	 & 61 
			\\
			4 & 85	 & 72	 & 94	 & 72	 & 56 
			\\
			5 & 87	 & 60	 & 62	 & 70	 & 94 
			\\
			6 & 76	 & 110	& 71	& 73	& 118 
			\\
			7 & 105	 & 108 & 51	   & 49	   & 72	
			\\
			\hline
			$ \sum_{\, i\, = 0}^{7} 
			c_{i} 
			$ 
			&
			704 & 	624	& 636 &	558	& 677 
			\\
			$ D $ & 	633 & 561 & 572 & 502 & 609 
			\\
			\hline		
		\end{tabular}
		\caption{Example of Random Setting Data: 5 Realizations, 8 eligible items with uniformly distributed capacity and 0.9 demand-capacity fraction}
		\label{Tbl Uniform Random Setting Example}
	\end{centering}
\end{table}
\begin{remark}\label{Rem Set of Rooms Size, N}
	\begin{enumerate}[(i)]
		\item 
		Since the successive application of the D\&C approach generates binary trees, for practical reasons, the numerical experiments will have a power of two (i.e., $ N = 2^{k} $ for some $ k \in \N $) as the number of eligible items.
		
		\item When generating a D\&C tree we want to distribute the demand between left and right children according to the relation \eqref{Eqn Demand Fraction}. Then, the
		inequality \eqref{Ineq Equivalent Control on the Complement} (equivalent to the hypothesis \eqref{Ineq  Control on the Complement} of part (iv) \textsc{Theorem} \ref{Thm DC feasibility}) must be satisfied. To that end a demand-capacity fraction $ o\in \{ 0.5, 0.55, 0.6, \ldots, 0.9\} $, furnishes a reasonable domain for numerical experimentation.
	\end{enumerate}
\end{remark}
%
%
%
\subsection{Tree Generation}
%
%
There will be two ways of generating a D\&C tree. Every vertex $ V $ of the tree, is associated with a subproblem analogous to \textsc{Problem} \ref{Pblm Integer Problem}, whose input is $ (A^{V}, D^{V}) $, with $ A^{V}\subseteq [N] $ a subset of items and an assigned demand $ D^{V} $. Denote by $ V , V_{l}, V_{r} $ a vertex together with its left and right children respectively and by $  \vert A^{V} \vert , A^{V _{l}}\vert, \vert A^{V_{r}}\vert $ the corresponding cardinals. The trees are constructed using the Left Pre-Order i.e., the stack has the structure $ [\text{root}, \text{left-child}, \text{right-child}] $ (see \textit{Algorithm 3.3.1} in\cite{GrossYellen} for details). The assigned demands to the left and right children will be given by the \textsc{Expression} \eqref{Eqn Demand Fraction}. All the difference between the algorithms is the way left-child and right-child are defined.     
\begin{enumerate}[I.]
	\item \textbf{First Case: Head-Left Subtree.} Select the following parameters 
	\begin{enumerate}[(i)]
		\item Select a sorting criterion: Specific Weight $ \boldsymbol{\gamma} $, Capacity $ \capac$, Prices $ \p $ or Random.
		
		
		\item Select a fraction for the head-left subtree, i.e. $ f \in [0, 1] $. 
		
		\item Define the minimum number of items in a subproblem, i.e. the quantity items in the subproblems associated to a leaf of the D\&C tree, namely $ m = 1, 2,... $, etc.
		
	\end{enumerate}
	%
	%
	Once the list of eligible items is sorted according to criterion $ s \in \{ \p, \capac, \boldsymbol{\gamma}\} $, the list of items $ A^{V_{l}} $ assigned to the left-child $V_{l}$ is defined as the first items of the list $ A^{V} $ such that $ \vert A^{V_{l}} \vert = \lfloor f\times \vert A^{V} \vert \rfloor $ i.e., the head of the list. The list assigned to the right-child is defined as the complement of that assigned to the left-child i.e., $ A^{V_{r}} \defining A^{V} - A^{V_{l}} $. The left and right demands are computed according to \textsc{Equation} \eqref{Eqn Demand Fraction}. The tree is constructed recursively as \textsc{Algorithm} \ref{Alg Head-Left Tree Generation} shows.
	\begin{algorithm} 
		\caption{Head-Left Subtree Algorithm, returns a D\&C tree }
		\label{Alg Head-Left Tree Generation}
		\begin{algorithmic}[1]
			\Procedure{Head-Left Subtree Generator}{Items' List.
				Prices: $ \p $, Capacities: $ \capac$,\newline
				Demand: $ D $. 
				Sorting: $s \in \{ \p ,\capac, \boldsymbol{\gamma}, \text{random}\} $,
				Head-left subtree fraction: $ f \in [0, 1] $,
				Minimum list size: $ m \in [1, \# \text{Items' List}]\cap \N $
			}
			\If {$ s = \gamma $ } \Comment{Asking if is necessary to compute specific weight}
			\State \textbf{compute} list of specific weights $ \big(\gamma_{i}: i \in [N]\big) $ 
			\Comment{Introduced in \textsc{Definition} \ref{Def Specific Weight}.}
			\EndIf
			\State 
			$ V_{0} = $ \textbf{sorted} (Items' List) according to chosen criterion $ s $\newline
			\State $ V \defining V_{0} $ \Comment{Initializing the root of the  D\&C tree}
			\State $ \text{D\&C tree} = \emptyset $ \Comment{Initializing D\&C tree as empty list}
			\State\Call{Branch}{ $ V , f , m, D, \capac$, \text{D\&C tree} } \Comment{Calling the \textsc{Branch} function of Algorithm \ref{Alg Branch Function}}
			\EndProcedure
		\end{algorithmic}
	\end{algorithm}
	\begin{algorithm} 
		\caption{Function Branch (Subroutine for Algorithm \ref{Alg Head-Left Tree Generation}) }
		\label{Alg Branch Function}
		\begin{algorithmic}[1]
			\Procedure{Branch Function}{List of Items: $ V $, 
				Head-left subtree fraction: $ f\in [0,1] $, 
				Minimum size list $ m $, 
				Demand: D, 
				Capacities: $ \capac $,
				Divide \& Conquer Tree: \text{D\&C tree}.
			}
			\Function{Branch}{ $ V , f , m, D, \capac$, \text{D\&C tree} }
			\If {$ \vert V \vert > m $}
			\State $ V \rightarrow \text{D\&C tree}  $
			\Comment{Push list $ V $ as node of the D\&C tree}
			\State $ lcs = \big\lfloor f \times \vert V \vert \big\rfloor $\Comment{Defining the size of the left child}
			\State $ V_{l} = \big(R_{i}:  1\leq  i \leq lcs \big) $ \Comment{Computing the left child}
			\State $ D_{l} = \Big\lfloor
			\dfrac{ \sum \{c_{i}: 1 \, \leq \, i \, \leq \, lcs \}}{\sum\{ c_{i}: 1 \, \leq \, i \, \leq \, \vert V \vert\} } \,
			D \Big\rfloor $
			\Comment{Computing the left demand}
			\State \textsc{Branch}( $ V_{l}, f, m, D_{l}, \capac_{l} \defining (c_{i}: 1 \, \leq \, i \, \leq \, lcs ) $) 
			\Comment{Recursing for the left subtree}
			\State $ V_{r} = \big(R_{i}:  lcs <  i \leq \vert V \vert \big) $\Comment{Computing the right child}
			\State $ D_{r} \defining D - D_{l} $ \Comment{Computing the right demand}
			\State \textsc{Branch}( $ V_{r}, f, m, D_{r}, \capac_{r} \defining (c_{i}: lcs \, \leq \, i \, \leq \, \vert V \vert ) ) $) \Comment{Recursing for the right subtree}
			\State \Return D\&C tree
			\Else 
			\State $ V \rightarrow \text{D\&C tree}  $
			\Comment{Push list $ V $ as node of the D\&C tree}
			\State \Return D\&C tree
			\EndIf
			\EndFunction
			\EndProcedure
		\end{algorithmic}
	\end{algorithm}
	In the table \ref{Tbl Tree Generated by Head-Left Algorithm} below, we present a binary tree for the first column of \textsc{Table} \ref{Tbl Uniform Random Setting Example} (Realization 1), with the following parameters: sorting by specific weight ($s =  \boldsymbol{\gamma} $), $ f = 0.5 $, $ m = 2 $; \textsc{Figure} \ref{Fig Tree Generated by Head-Left Algorithm} shows its graphic representation. Finally, \textsc{Figure} \ref{Fig Tree Generated by Head-Left Algorithm Biased} depicts a tree generated for the same realization, but with parameters $s =  \boldsymbol{\gamma} $, $ f = 0.4 $ and $ m = 2 $; the corresponding table is omitted.
	\begin{table}[h!]
		\begin{centering}
			\rowcolors{2}{gray!25}{white}
			\begin{tabular}{c | ccccccc}
				\rowcolor{gray!50}
				\hline
				\diagbox{Item}{Vertex} & $ V_{0} $ & $ V_{1} $ & $ V_{2} $ & $ V_{3} $ & $ V_{4} $ & $ V_{5} $ & $ V_{6} $
				\\
				\hline
				1 &	1 &	1 & 1 & 0 & 0 & 0 & 0 \\				
				7 & 1 & 1 &	1 & 0 & 0 & 0 & 0 \\			
				2 &	1 & 1 & 0 & 1 & 0 & 0 & 0 \\			
				3 &	1 & 1 & 0 & 1 & 0 & 0 & 0 \\		
				5 &	1 &	0 & 0 & 0 &	1 &	1 & 0 \\	
				4 &	1 &	0 & 0 & 0 &	1 &	1 & 0 \\
				6 &	1 &	0 & 0 & 0 &	1 &	0 & 1 \\
				0 &	1 &	0 & 0 & 0 &	1 &	0 & 1 \\
				\hline
				$ D $ & 633 & 309 & 144 & 165 & 324 & 155 & 169\\
				\hline		
			\end{tabular}
			\caption{\textsc{Algorithm} \ref{Alg Head-Left Tree Generation} tree generated for Realization 1 of \textsc{Table} \ref{Tbl Uniform Random Setting Example}. Parameters: sorting by specific weight $ \boldsymbol{\gamma} $, left subtree fraction $ f = 0.5 $, minimum size leaf $ m = 2 $.}
			\label{Tbl Tree Generated by Head-Left Algorithm}
		\end{centering}
	\end{table}
	\begin{figure}
		\centering
		\begin{tikzpicture}
		[scale=.9,auto=left,every node/.style={}]
		\node (n0) at (6,4) {
			$ 
			\begin{pmatrix} 
			A_{0} = [1, 7, 2, 3, 5, 4, 6, 0]\\[3pt]
			D_{0} = 633
			\end{pmatrix} \mapsto V_{0}$
		};
		\node (n1) at (2,2)  {$ 
			\begin{pmatrix}
			A_{1} = [1, 7, 2, 3]\\[3pt]
			D_{1} = 309
			\end{pmatrix} \mapsto V_{1}$};
		\node (n2) at (0,0)  {$ 
			\begin{pmatrix}
			A_{2} = [1, 7]\\[3pt]
			D_{2} = 144
			\end{pmatrix} \mapsto V_{2}$};
		\node (n3) at (4,0)  {$ 
			\begin{pmatrix}
			A_{3} = [2,3] \\[3pt]
			D_{3} = 165
			\end{pmatrix} \mapsto V_{3} $};
		\node (n4) at (10,2)  {$ 
			\begin{pmatrix}
			A_{4} = [5, 4, 6, 0] \\[3pt]
			D_{4} = 324
			\end{pmatrix} \mapsto V_{4} $};
		\node (n5) at (8,0)  {$ 
			\begin{pmatrix}
			A_{5} = [5,4] \\[3pt]
			D_{5} = 155
			\end{pmatrix} \mapsto V_{5} $};
		\node (n6) at (12,0)  {$ 
			\begin{pmatrix}
			A_{6} = [6, 0] \\[3pt]
			D_{6} = 169
			\end{pmatrix}  \mapsto V_{6} $};
		
		\foreach \from/\to in {n0/n1, n1/n2,n1/n3, n0/n4, n4/n5, n4/n6}
		\draw (\from) -- (\to);
		\end{tikzpicture}
		%
		\caption{\textsc{Algorithm} \ref{Alg Head-Left Tree Generation} D\&C tree generated for Realization 1 of \textsc{Table} \ref{Tbl Uniform Random Setting Example}. The tree is consistent with \textsc{Table} \ref{Tbl Tree Generated by Head-Left Algorithm}. Parameters: sorting by specific weight ($ s = \boldsymbol{\gamma} $), left subtree fraction $ f = 0.5 $, minimum size leaf $ m = 2 $. Every vertex $ V_{i} $ has associated a subproblem analogous to \textsc{Problem} \ref{Pblm Integer Problem}, whose input data are the demand $ D_{i} $ and the sorted list of eligible items $ A_{i} $ (together with its corresponding lists of capacities and prices).}
		\label{Fig Tree Generated by Head-Left Algorithm}
	\end{figure}

	%
	%
	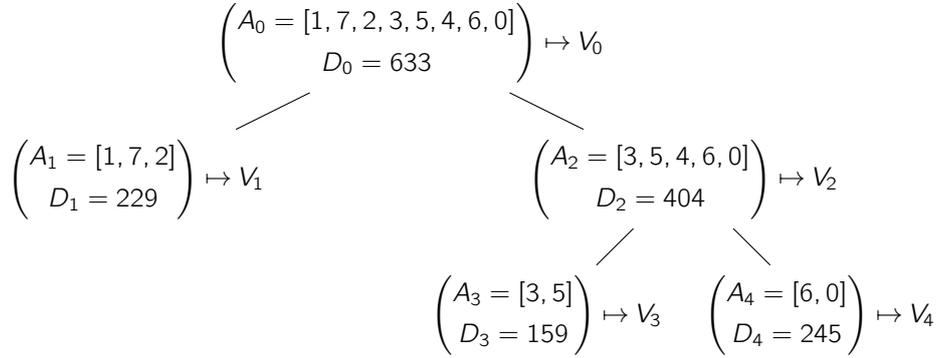
\begin{figure}
		\centering
		\begin{tikzpicture}
		[scale=.9,auto=left,every node/.style={}]
		\node (n1) at (6,4) {$ 
			\begin{pmatrix}
			A_{0} = [1, 7, 2, 3, 5, 4, 6, 0]\\[3pt]
			D_{0} = 633 
			\end{pmatrix} \mapsto V_{0}$};
		\node (n2) at (2,2)  { $ 
			\begin{pmatrix}
			A_{1} =  [1, 7, 2]\\[3pt]
			D_{1} = 229
			\end{pmatrix} \mapsto V_{1} $ };
		\node (n3) at (10,2)  {$ 
			\begin{pmatrix}
			A_{2} = [3, 5, 4, 6, 0]\\[3pt]
			D_{2} = 404
			\end{pmatrix} \mapsto V_{2} $};
		\node (n4) at (8,0)  {$ 
			\begin{pmatrix}
			A_{3} = [3, 5] \\[3pt]
			D_{3} = 159
			\end{pmatrix} \mapsto V_{3} $ };
		\node (n5) at (12,0)  {$ 
			\begin{pmatrix}
			A_{4} = [6, 0]\\[3pt]
			D_{4} = 245
			\end{pmatrix} \mapsto V_{4} $ };
		
		\foreach \from/\to in {n1/n2,n1/n3, n3/n4,n3/n5}
		\draw (\from) -- (\to);
		\end{tikzpicture}
		%
		\caption{\textsc{Algorithm} \ref{Alg Head-Left Tree Generation} D\&C tree generated for Realization 1 of \textsc{Table} \ref{Tbl Uniform Random Setting Example}. Parameters: sorting by specific weight ($ s = \boldsymbol{\gamma} $), left subtree fraction $ f = 0.4 $, minimum size leaf $ m = 2 $. Every vertex $ V_{i} $ has associated a subproblem analogous to \textsc{Problem} \ref{Pblm Integer Problem}, whose input data are the demand $ D_{i} $ and the sorted list of eligible items $ A_{i} $ (together with its corresponding lists of capacities and prices).}
		\label{Fig Tree Generated by Head-Left Algorithm Biased}
	\end{figure}
	%
	%
	\item \textbf{Second Case: Balanced Left-Right Subtrees.} Select the same parameters as in the previous case except for the fraction head $ f\in [0,1] $ since this will be $ 0.5 $ by default. Once the list of eligible items is sorted according to criterion $ s \in \{ \p, \capac, \boldsymbol{\gamma}\} $, the list of items $ A^{V_{l}} $ assigned to the left-child $ V_{l} $ is defined as the items in even positions on the sorted list $ A^{V} $. The items $ A^{V_{r}} $ assigned to the right-child, is defined as the complement of those assigned to the left-child i.e., $ A^{V_{r}} \defining A^{V} - A^{V_{l}} $ i.e., the left and right lists of items are as balanced as possible, according to $ s $. The left and right demands are computed according to \textsc{Equation} \eqref{Eqn Demand Fraction}.
	%
	%
	%
	%
	%
	Again, the tree is constructed recursively as the \textsc{Algorithm} \ref{Alg Balanced-Left Tree Generation} shows. In \textsc{Table} \ref{Tbl Balanced Tree for Particular Example} below, we present a binary tree for the first column (Realization 1) of \textsc{Table} \ref{Tbl Uniform Random Setting Example}, with the following parameters: sorting by specific weight ($s =  \boldsymbol{\gamma} $), $ m = 2 $; its graphic representation is displayed in \textsc{Figure} \ref{Fig Tree Generated by Balanced Left-Right Algorithm}.
	\begin{algorithm} 
		\caption{Balanced Left-Right Subtrees Algorithm, returns a D\&C tree}
		\label{Alg Balanced-Left Tree Generation}
		\begin{algorithmic}[1]
			\Procedure{Balanced Left-Right Subtrees Generator}{Items' List.
				Prices: $ \p $, Capacities: $ \capac$,\newline
				Demand: $ D $. 
				Sorting: $s \in \{ \p ,\capac, \boldsymbol{\gamma}, \text{random}\} $,
				Minimum list size: $ m \in [1, \# \text{Items' List}]\cap \N $
			}
			\If {$ s = \gamma $ } \Comment{Initializing the root of the  D\&C tree}
			\State \textbf{compute} list of specific weights $ \big(\gamma_{i}: i \in [N]\big) $ 
			\Comment{Introduced in \textsc{Definition} \ref{Def Specific Weight}.}
			\EndIf
			\State 
			$ V_{0} = $ \textbf{sorted} (Items' List) according to chosen criterion $ s $\newline
			\State $ V \defining V_{0} $ \Comment{Initializing the root of the  D\&C tree}
			\State \text{D\&C tree} $  DCT = \emptyset $ \Comment{Initializing D\&C tree as empty list}
			\Function{Branch}{ $ V , m, D, \capac$ }
			\If {$ \vert V \vert > m $}
			\State $ V \rightarrow DCT  $
			\Comment{Push list $ V $ as node of the D\&C tree}
			\State $ V_{l} = \big(R_{i}:  1\leq  i \leq \vert V \vert, \, i \text{ even } \big) $ 
			\Comment{Computing the left child}
			\State $ D_{l} = \Big\lfloor
			\dfrac{ \sum \{c_{i}: 1\leq  i \leq \vert V \vert, \, i \text{ even } \}}{\sum\{ c_{i}: 1\leq  i \leq \vert V \vert\} } \,
			D \Big\rfloor $
			\Comment{Computing the left demand}
			\State \textsc{Branch}( $ V_{l}, m, D_{l}, \capac_{l} \defining (c_{i}: 1\leq  i \leq \vert V \vert, \, i \text{ even }  ) $) 
			\Comment{Recursing for the left subtree}
			\State $ V_{r} = \big(R_{i}:  1 \leq  i \leq \vert V \vert , \, i \text{ odd } \big) $
			\Comment{Computing the right child}
			\State $ D_{r} \defining D - D_{l} $ \Comment{Computing the right demand}
			\State \textsc{Branch}( $ V_{r}, m, D_{r}, \capac_{r} \defining (1\leq  i \leq \vert V \vert, \, i \text{ odd }  ) ) $) \Comment{Recursing for the right subtree}
			\State \Return $ DCT $\Comment{return the D\&C tree}
			\Else 
			\State $ V \rightarrow DCT  $
			\Comment{Push list $ V $ as node of the D\&C tree}
			\State \Return $ DCT $\Comment{return the D\&C tree}
			\EndIf
			\EndFunction
			\EndProcedure
		\end{algorithmic}
	\end{algorithm}
	\begin{table}[h!]
		\begin{centering}
			\rowcolors{2}{gray!25}{white}
			\begin{tabular}{c | ccccccc}
				\rowcolor{gray!50}
				\hline
				\diagbox{Item}{Vertex} & $ V_{0} $ & $ V_{1} $ & $ V_{2} $ & $ V_{3} $ & $ V_{4} $ & $ V_{5} $ & $ V_{6} $
				\\
				\hline
				1 &	1 & 1 & 1 & 0 & 0 & 0 & 0\\			
				7 &	1 & 0 & 0 & 0 & 1 &	1 & 0 \\
				2 &	1 & 1 & 0 & 1 & 0 & 0 & 0 \\		
				3 &	1 & 0 & 0 & 0 &	1 & 0 & 1 \\
				5 &	1 &	1 & 1 &	0 & 0 & 0 & 0 \\		
				4 &	1 & 0 & 0 & 0 & 1 & 1 & 0 \\
				6 &	1 &	1 & 0 &	1 & 0 & 0 & 0 \\			
				0 &	1 & 0 & 0 & 0 &	1 & 0 &	1 \\
				$ D $ &	633 & 281 &	127 & 154 &	352 & 171 &	181 \\
				\hline		
			\end{tabular}
			\caption{\textsc{Algorithm} \ref{Alg Balanced-Left Tree Generation} tree generated for Realization 1 of \textsc{Table} \ref{Tbl Uniform Random Setting Example}. Parameters: sorting by specific weight $ \boldsymbol{\gamma} $, minimum size leaf $ m = 2 $.}
			\label{Tbl Balanced Tree for Particular Example}
		\end{centering}
	\end{table}
\end{enumerate}
\begin{figure}
	\centering
	\begin{tikzpicture}
	[scale=.9,auto=left,every node/.style={}]
	\node (n0) at (6,4) {
		$ 
		\begin{pmatrix}
		A_{0} = [1, 7, 2, 3, 5, 4, 6, 0] \\[3pt]
		D_{0} = 633
		\end{pmatrix} \mapsto V_{0} $
	};
	\node (n1) at (2,2)  {$ 
		\begin{pmatrix}
		A_{1} = [1, 2, 5, 6]\\[3pt]
		D_{1} = 281
		\end{pmatrix} \mapsto V_{1} $};
	\node (n2) at (0,0)  {$ 
		\begin{pmatrix}
		A_{2} = [1, 5]\\[3pt]
		D_{2} = 127
		\end{pmatrix} \mapsto V_{2} $};
	\node (n3) at (4,0)  {$ 
		\begin{pmatrix}
		A_{3} =  [2,6]\\[3pt]
		D_{3} = 154
		\end{pmatrix} \mapsto V_{3} $};
	\node (n4) at (10,2)  {$ 
		\begin{pmatrix}
		A_{4} = [7, 3, 4, 0]\\[3pt]
		D_{4} = 352
		\end{pmatrix} \mapsto V_{4} $};
	\node (n5) at (8,0)  {$ 
		\begin{pmatrix}
		A_{5} =  [7,4]\\[3pt]
		D_{5} = 171
		\end{pmatrix} \mapsto V_{5} $};
	\node (n6) at (12,0)  {$ 
		\begin{pmatrix}
		A_{6} = [3, 0]\\[3pt]
		D_{6} = 181
		\end{pmatrix} \mapsto V_{6} $};
	
	\foreach \from/\to in {n0/n1, n1/n2,n1/n3, n0/n4, n4/n5, n4/n6}
	\draw (\from) -- (\to);
	\end{tikzpicture}
	%
	\caption{\textsc{Algorithm} \ref{Alg Balanced-Left Tree Generation} D\&C tree generated for Realization 1 of \textsc{Table} \ref{Tbl Uniform Random Setting Example}. The tree is consistent with \textsc{Table} \ref{Tbl Balanced Tree for Particular Example}. Parameters: sorting by specific weight ($ s = \boldsymbol{\gamma} $), minimum size leaf $ m = 2 $. Every vertex $ V_{i} $ has associated a subproblem analogous to \textsc{Problem} \ref{Pblm Integer Problem}, whose input data are the demand $ D_{i} $ and the sorted list of eligible items $ A_{i} $ (together with its corresponding lists of capacities and prices).}
	\label{Fig Tree Generated by Balanced Left-Right Algorithm}
\end{figure}
%
%
%
%
%
\subsection{Efficiency Quantification}
%
%
In this section we describe the general algorithm to compute the efficiency of the D\&C tree approach. The efficiencies will be measured according to \textsc{Definition} \ref{Def DC Tree efficiency}, moreover the computations will be done based on three values: 
\begin{enumerate}[1.]
	\item 
	Exact solution of \textsc{Problem} \ref{Pblm Integer Problem}, computed using the algorithm COMBO presented in \cite{martello1999dynamic}, from now on 
	%
	%
	denoted by $ DPS $ (the algorithm heavily relies on dynamic programming).
	
	\item Upper bound furnished by the \textsc{Greedy Algorithm} \ref{Alg Greedy Algorithm}, denoted by $ GAS $ in the sequel.
	
	\item Lower bound, given by the solution of \textsc{Problem} \ref{Pblm Natural LOP Problem}, i.e., the natural linear relaxation  of the problem \ref{Pblm Integer Problem}, from now on denoted by $ LRS $.
\end{enumerate}
The effectiveness of upper and lower bounds mentioned above is measured in the standard way i.e.,
\begin{align}\label{Eqn upper and lower bounds efficiency}
& GAE \defining 100\times\frac{GAS - DPS}{DPS}, &
& LRE \defining 100\times\frac{DPS - LRS}{DPS}.
\end{align}
Here, $ GAE, LRE $ respectively indicate, Greedy Algorithm and Linear Relaxation Efficiency. 
The general structure is as follows
\begin{enumerate}[(i)]
	\item Execute the \textsc{Random Setting Algorithm} described in \textsc{Section} \ref{Sec Random Setting}, according to its parameters of choice and store its results in the file \textbf{Eligible\_Items.xls}.
	
	
	\item Loop through the columns of file \textbf{Eligible\_Items.xls}, each of them is a random realization (see \textsc{Table} \ref{Tbl Uniform Random Setting Example}). 
	
	\item For each column/realization, 
	\begin{enumerate}[(a)]
		\item Retrieve the basic information of \textsc{Problem} \ref{Pblm Integer Problem} i.e., Items' List, Prices: $ \p $, Capacities: $ \capac$, Demand: $ D $.
		
		\item Build the D\&C tree, Head-Left (\textsc{Algorithm} \ref{Alg Head-Left Tree Generation}) or balanced (\textsc{Algorithm} \ref{Alg Balanced-Left Tree Generation}) according to user's choice. 
		
		\item Loop through the D\&C tree nodes, compute the Greedy Algorithm \ref{Alg Greedy Algorithm}, Exact and Linear Relaxation solutions and store them in the D\&C tree structure. 
		
		\item Loop through the D\&C tree heights, compute the global and stepwise efficiencies according to \textsc{Definition} \ref{Def DC Tree efficiency} (iv) and store them in stack structures within a realizations' global table (see, \textsc{Table} \ref{Tbl DPS GbE and SwE Efficiencies 5 Realizations of Table 4}). Compute the Greedy Algorithm and Linear Relaxation Efficiencies as defined in \textsc{Equation} \eqref{Eqn upper and lower bounds efficiency} and store them in stack structures within a realizations' global table (see \textsc{Table} \ref{Tbl DPS GAE and LRE Efficiencies 5 Realizations of Table 4}). 
	\end{enumerate}
	\item In the realizations' global table, compute the average of the global and stepwise efficiencies.
\end{enumerate}
The steps (ii) and (iii) of the previous description are detailed in the pseudocode \ref{Alg D and C Efficiency Quantification}, an example of its output is presented in the table \ref{Tbl Realization 1 Efficiencies Table} below, where the efficiencies of the method are reported for the Realization 1 of \textsc{Table} \ref{Tbl Uniform Random Setting Example}, using the D\&C tree structure, depicted in \textsc{Figure} \ref{Fig Tree Generated by Head-Left Algorithm} and detailed in \textsc{Table} \ref{Tbl Tree Generated by Head-Left Algorithm}. 
\begin{table}[h!]
	\begin{centering}
		\rowcolors{2}{gray!25}{white}
		\begin{tabular}{cccccccccc}
			\rowcolor{gray!50}
			\hline
			Height 
			& $ LRS $ & $ DPS $ & $ GAS $ 
			& $ GbE_{LRS} $ & $ GbE_{DPS} $ & $ GbE_{GAS} $
			& $ SwE_{LRS} $ & $ SwE_{DPS} $ & $ SwE_{GAS} $
			\\
			\hline
			0 &	14.12 &	15 & 16 & 0.00 & 0.00 & 0.00	& & &	\\
			1 &	14.25 & 16 & 16	& 0.98 & 6.67 &	0.00 &	0.98 &	6.67 &	0.00 \\
			2 &	14.36 &	16 & 16	& 1.71 & 6.67 & 0.00 &	0.72 &	0.00 & 0.00	
			\\
			\hline		
		\end{tabular}
		\caption{\textsc{Algorithm} \ref{Alg D and C Efficiency Quantification} height efficiencies for Realization 1 of \textsc{Table} \ref{Tbl Uniform Random Setting Example}. Parameters: sorting by specific weight $ \boldsymbol{\gamma} $, left subtree fraction $ f = 0.5 $, minimum size leaf $ m = 2 $, see also \textsc{Figure} \ref{Fig Tree Generated by Head-Left Algorithm} and \textsc{Table} \ref{Tbl Tree Generated by Head-Left Algorithm} for details on the tree structure. The Linear Relaxation, Exact and Greedy Algorithm solutions are represented with the initials $ LRS, DPS $ and $ GAS $ respectively. The Global and Stepwise Efficiencies are represented with the initials $ GbE, SwE $ respectively and the subindex affecting them, indicates for which of the solutions $ LRS, DPS, GAS $ the column values apply.}
		\label{Tbl Realization 1 Efficiencies Table}
	\end{centering}
\end{table}

\noindent In addition, for the five realizations of \textsc{Table} \ref{Tbl Uniform Random Setting Example}, the  table \ref{Tbl DPS GbE and SwE Efficiencies 5 Realizations of Table 4} presents the result of computing the global and stepwise efficiencies ($ GbE $ and $ SwE $) of the Exact Solutions ($ DPS $), while \textsc{Table} \ref{Tbl DPS GAE and LRE Efficiencies 5 Realizations of Table 4} displays the corresponding values of the Greedy Algorithm and the Linear Relaxation Efficiencies ($ GAE $ and $ LRE $).
\begin{table}[h!]
	\begin{centering}
		\rowcolors{2}{gray!25}{white}		
		\begin{tabular}{cccccccccccc}
			\rowcolor{gray!50}
			\hline
			Height 
			& $ GbE_{1} $ & $ GbE_{2}$ & $ GbE_{3} $ & $ GbE_{4} $ & $ GbE_{5} $ 
			& $ SwE_{1} $ & $ SwE_{2}$ & $ SwE_{3} $ & $ SwE_{4} $ & $ SwE_{5} $ 
			\\
			\hline
			0 & 0.00 & 	0.00 &	0.00 &	0.00 &	0.00 & & & & & \\				
			1 &	6.67 &	14.29 &	14.29 &	7.14 &	13.33 &	6.67 & 14.29 & 14.29 & 7.14	& 13.33 \\
			2 & 6.67 &	14.29 &	14.29 & 7.14 &	13.33 &	0.00 &	0.00  &	0.00   & 0.00 & 0.00
			\\
			\hline		
		\end{tabular}
		\caption{\textsc{Algorithm} \ref{Alg D and C Efficiency Quantification} example of Global Efficiency ($ GbE $) and Stepwise Efficiency ($ SwE $) results for the case Exact Solution ($ DPS $) through the 5 realizations of \textsc{Table} \ref{Tbl Uniform Random Setting Example}. Parameters: sorting by specific weight $ \boldsymbol{\gamma} $, left subtree fraction $ f = 0.5 $, minimum size leaf $ m = 2 $. The subindex affecting $ GbE $ and $ SwE $ indicates the corresponding number of realization for which the column applies.}
		\label{Tbl DPS GbE and SwE Efficiencies 5 Realizations of Table 4}
	\end{centering}
\end{table}
\begin{table}[h!]
	\begin{centering}
		\rowcolors{2}{gray!25}{white}		
		\begin{tabular}{cccccccccccc}
			\rowcolor{gray!50}
			\hline
			Height 
			& $ GAE_{1} $ & $ GAE_{2}$ & $ GAE_{3} $ & $ GAE_{4} $ & $ GAE_{5} $ 
			& $ LRE_{1} $ & $ LRE_{2}$ & $ LRE_{3} $ & $ LRE_{4} $ & $ LRE_{5} $ 
			\\
			\hline
			0 & 6.67 &	0.00 &	0.00 &	7.14 &	0.00 &	5.90  &	0.66 &	0.45 &	6.65 &	2.62 \\
			1 & 0.00 &	0.00 & 	0.00 &	0.00 &	0.00 &	10.91 &	11.76 &	11.21 &	11.74 &	12.00 \\
			2 & 0.00 &	0.00 &	0.00 &	0.00 &	0.00 &	10.27 &	10.68 &	10.51 &	10.52 &	10.58
			\\
			\hline		
		\end{tabular}
		\caption{\textsc{Algorithm} \ref{Alg D and C Efficiency Quantification} example of Greedy Algorithm Efficiency ($ GAE $) and Linear Relaxation Efficiency ($ LRE $) (see \textsc{Equation} \eqref{Eqn upper and lower bounds efficiency} for its definition), through the 5 realizations of \textsc{Table} \ref{Tbl Uniform Random Setting Example}. Parameters: sorting by specific weight $ \boldsymbol{\gamma} $, left subtree fraction $ f = 0.5 $, minimum size leaf $ m = 2 $. The subindex affecting $ GAE $ and $ LRE $ indicates the corresponding number of realization for which the column applies.}
		\label{Tbl DPS GAE and LRE Efficiencies 5 Realizations of Table 4}
	\end{centering}
\end{table}
\begin{algorithm} 
	\caption{D\&C Efficiency Quantification, returns a list of global and stepwise efficiencies}
	\label{Alg D and C Efficiency Quantification}
	\begin{algorithmic}[1]
		\Procedure{D\&C Efficiency Quantification}{File \textbf{Eligible\_Items.xls} contains:\newline
			Items' List, Prices: $ \p $, Capacities: $ \capac$, Demand: $ D $. \newline
			\textbf{User Decisions}: Sorting: $s \in \{ \p ,\capac, \boldsymbol{\gamma}, \text{random}\} $,
			Head-left subtree fraction: $ f \in [0, 1] $,\newline
			Minimum list size: $ m \in [1, \# \text{Items' List}]\cap \N $,
			Price-Capacity rate: $ r \in [1, \max\limits_{i} c_{i} ]$,\newline
			Type of Tree: $ t \in \{ \text{Head-Left, Balanced} \} .$ 
		}
		\For{column of \textbf{Eligible\_Items.xls} }
		\Comment{Each column is a random realization, e.g. \textsc{Table} \ref{Tbl Uniform Random Setting Example}}
		
		\State \textbf{retrieve} from \textbf{Eligible\_Items.xls} the information: Items' List, Prices: $ \p $, Capacities: $ \capac$, Demand: $ D $, corresponding to column/realization.
		
		\If {$ t = \text{Head-Left} $ } 
		
		\State D\&C tree: $ DCT \defining $ \textbf{call} \textsc{Algorithm} \ref{Alg Head-Left Tree Generation} (Items' List, $ \p $, $ \capac$, $ D $, $ s $, $ f $, $ m $ )\newline
		\Comment{Producing the Head-Left D\&C tree}
		
		\Else
		
		\State D\&C tree: $ DCT \defining $ \textbf{call} \textsc{Algorithm} \ref{Alg Balanced-Left Tree Generation} (Items' List, $ \p $, $ \capac$, $ D $, $ s $, $ m $ )\newline
		\Comment{Producing the Balanced D\&C tree}
		
		\EndIf
		
		\State Solutions Tree: $ ST = \emptyset $
		\Comment{Initializing Solutions Tree as empty list}
		
		\For {$ V \in \text{vertices of } DCT $} 
		\Comment{Recall that $ DCT $ has table format as \textsc{Table} \ref{Tbl Balanced Tree for Particular Example} }
		
		\State Linear Relaxation Solution: $ LRS_{V}\leftarrow $ \textbf{call simplex algorithm solver} (Data $ \{\p, \capac, D\} $, corresponding to vertex $ V $)
		
		\State Exact Solution: $ DPS_{V} \leftarrow $  \textbf{call MT1 solver} Data $ \{\p, \capac, D\} $, corresponding to vertex $ V $) 
		
		\State Greedy Algorithm Solution: $ GAS_{V} \leftarrow $ \textbf{call} \textsc{Algorithm} \ref{Alg Greedy Algorithm} (Data $ \{\p, \capac, D\} $, corresponding to vertex $ V $) 
		
		\State $ [LRS_{V}, DPS_{V}, GAS_{V}] \rightarrow ST $ 
		\Comment{Push the triple $ [LRS_{V}, DPS_{V}, GAS_{V}] $ as vertex of the solutions tree $ ST $}
		
		
		\EndFor
		
		\State $ z_{*}\defining \emptyset $
		\Comment{Initializing solution values stack as empty list} 
		
		\State 
		$ GbE \defining [0] $
		\Comment{Initializing global efficiency stack; 0 is the first value} 
		
		\State 
		$ SwE \defining \emptyset $
		\Comment{Initializing stepwise eficiency stack as empty list} 
		
		\State 
		$ GAE \defining \emptyset $
		\Comment{Initializing greedy algorithm efficiency stack as empty list} 
		
		\State 
		$ LRE \defining \emptyset $
		\Comment{Initializing linear relaxation eficiency stack as empty list} 
		
		\State $ H \defining $ height of $ DCT $. 
		
		\For {$ h \in [H] $ }
		
		\State $ DCT_{h} = \text{subgraph of } DCT \text{ induced on the set }\{V \in DCT: \text{ height} (V) \leq h\} $ 
		\Comment{Tree pruned at height $ h $}
		
		\State $ L(DCT_{h}) = \{V \in DCT_{h}: \deg(V) = 1 \}$
		\Comment{Selecting the leaves of the pruned tree $ DCT_{h} $}
		
		\State $ z_{*}^{h}\leftarrow \sum\{ [LRS_{V}, DPS_{V}, GAS_{V}] : V\,\in \,L(DCT_{h}) \}=
		\sum\{ ST(V) : V\,\in \,L(DCT_{h}) \} $
		\Comment{Push the total solutions (Linear Relaxation, Exact, Greedy) at height $ h $ of the three $ DCT $, to the stack}
		
		\If{ $ h > 0 $ }
		
		\State $ GbE(h) \leftarrow 100\times \dfrac{z_{*}^{h} - z_{*}^{0}}{z_{*}^{0}} $
		\Comment{Push global efficiency at height $ h $ to the stack}
		
		\State $ SwE(h - 1)\leftarrow 100\times \dfrac{z_{*}^{h-1}- z_{*}^{h}}{z_{*}^{h - 1}} $
		\Comment{Push stepwise efficiencyror at height $ h $ to the stack}
		
		\EndIf
		
		\State $ GAE(h) \leftarrow 100\times \dfrac{ z_{*}^{h}[ GAS ] - z_{*}^{h}[ DPS ] }{ z_{*}^{h}[ DPS ] } $
		\Comment{Push greedy algorithm efficiencies into the stack, see \textsc{Equation} \eqref{Eqn upper and lower bounds efficiency} }
		
		\State $ LRE(h) \leftarrow 100\times \dfrac{ z_{*}^{h}[ DPS ] - z_{*}^{h}[ LRS ] }{ z_{*}^{h}[ DPS ] } $
		\Comment{Push linear relaxation efficiencies into the stack, see \textsc{Equation} \eqref{Eqn upper and lower bounds efficiency} }
		
		\EndFor
		
		\State \Return $ (GbE, SwE) $ 
		\Comment{Efficiencies corresponding to column/realization}
		
		\EndFor
		
		\EndProcedure
	\end{algorithmic}
\end{algorithm}
So far, we have been using Realization 1 in \textsc{Table} \ref{Tbl Uniform Random Setting Example} to illustrate the method, however, we close this section presenting an example significantly larger in order to illustrate the method for a richer D\&C tree and bigger range of heights. 
\begin{example}[The D\&C tree of a large random realization]
	\label{Exm DC tree random realization}
	In \textsc{Table} \ref{Fig Table Solutions through tree} we present the $ LRS, DPS, GAS $ solutions for a D\&C tree corresponding to a random realization of 128 eligible items, uniformly distributed capacities, with demand-capacity fraction of 0.9. The respective D\&C tree is constructed using the head-left algorithm \ref{Alg Head-Left Tree Generation}, sorted by specific weight $ \boldsymbol{\gamma} $, left subtree fraction $ f = 0.5 $ and minimum size $ m = 1 $, i.e., its height is 7. To avoid redundancy, we omit tables displaying the corresponding values of $ GAE $, $ LRE $ as well as $ GbE, SwE $ for $ LRS, DPS, GAS $, analogous to those reported in \textsc{Tables} \ref{Tbl Realization 1 Efficiencies Table} and \ref{Tbl DPS GAE and LRE Efficiencies 5 Realizations of Table 4}, since they can be completely derived from \textsc{Table} \ref{Fig Table Solutions through tree}; however, we display the graphics corresponding to all such tables. \\
	
	\noindent In \textsc{Figure} \ref{Fig Efficiencies through tree} we depict the behavior through the heights of a D\&C tree, for the solutions $ LRS, DPS, GAS $, the efficiencies $ GAE, LRE $, as well as the global and stepwise efficiencies $ \big\{GbE_{LRS}, GbE_{DPS}, GbE_{GAS}\big\} $, $ \big\{SwE_{LRS}, SwE_{DPS}, SwE_{GAS}\big\} $. 
	As it can be seen in figures (a), (b), $ GAS $ is significantly more accurate than $ LRS $ to the point that one curve stays below the other through all the height of the D\&C tree. In the case of global efficiencies we also observe that the behavior of $ GbE_{GAS} $ and $ GbE_{DPS} $ are similar, though none is above the other through all the D\&C tree heights and $ GbE_{DPS} $ stays below both of them. A similar behavior is observed for the case of stepwise efficiencies ($ SwE $), although the curves $ SwE_{DPS} $ and $ SwE_{LRS} $ intersect in this case for $ h = 2 $. Observe that if $ h \geq 4 $, the results for $ DPS, GAS, GAE, GBE_{DPS}, GbE_{GAS} $ become stable i.e., the D\&C method no longer deteriorates the exact solution; since $ N = 128 $, $ h \geq 4 $ corresponds to lists of $ 8 $ items or smaller.
	\begin{table}[h!]
		\begin{centering}
			\rowcolors{2}{gray!25}{white}		
			\begin{tabular}{cccc}
				\rowcolor{gray!50}
				\hline
				Height 
				& $ LRS $ & $ DPS $ & $ GAS $
				\\
				\hline
				0 & 233.43	& 234 &	236 \\
				1 & 236.41	& 239 &	239 \\
				2 & 238.02	& 240 & 242 \\
				3 & 238.79	& 248 &	249 \\
				4 & 239.12	& 262 &	266 \\
				5 & 239.25	& 266 &	266 \\
				6 & 239.33	& 266 &	266 \\
				7 & 239.38	& 266 &	266
				\\
				\hline		
			\end{tabular}
			\caption{\textsc{Example} \ref{Exm DC tree random realization}. Solutions $ LRS, DPS $ and $ GAS $ table for a random realization of 128 eligible items uniformly distributed and demand-capacity fraction of 0.9. The D\&C tree has height 7, generated by the head-left algorithm \ref{Alg Head-Left Tree Generation}, sorted by specific weight $ \boldsymbol{\gamma} $, left subtree fraction $ f = 0.5 $ and minimum size $ m = 1 $.}
			\label{Fig Table Solutions through tree} 
		\end{centering}
	\end{table}
	\begin{figure}[h!]
		\centering
		\begin{subfigure}[$ LRS, DPS $ and $ GAS $ solutions.]
			{\includegraphics[scale = 0.380]{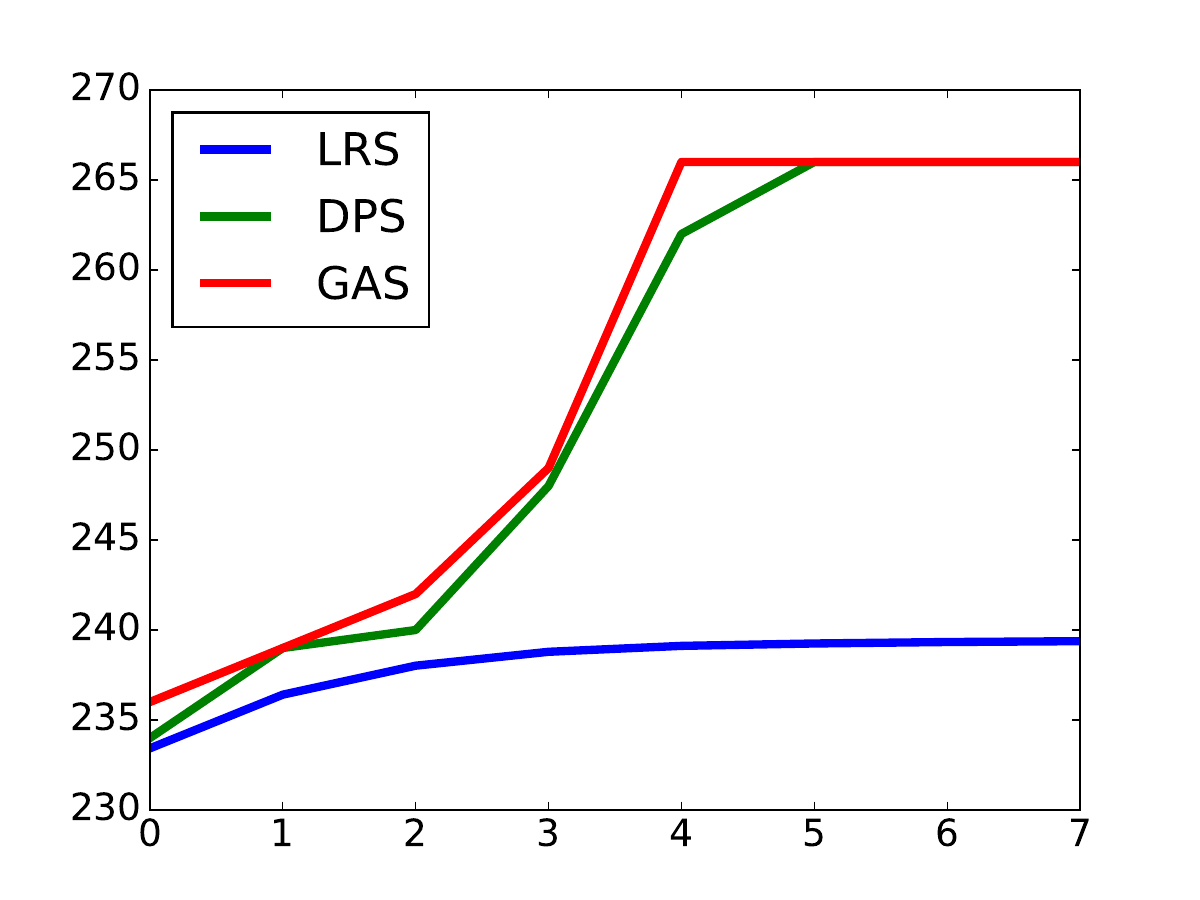} } 
		\end{subfigure}
		~ 
		\begin{subfigure}[$ GAE $ and $ LRE $ efficiencies.]
			{\includegraphics[scale = 0.380]{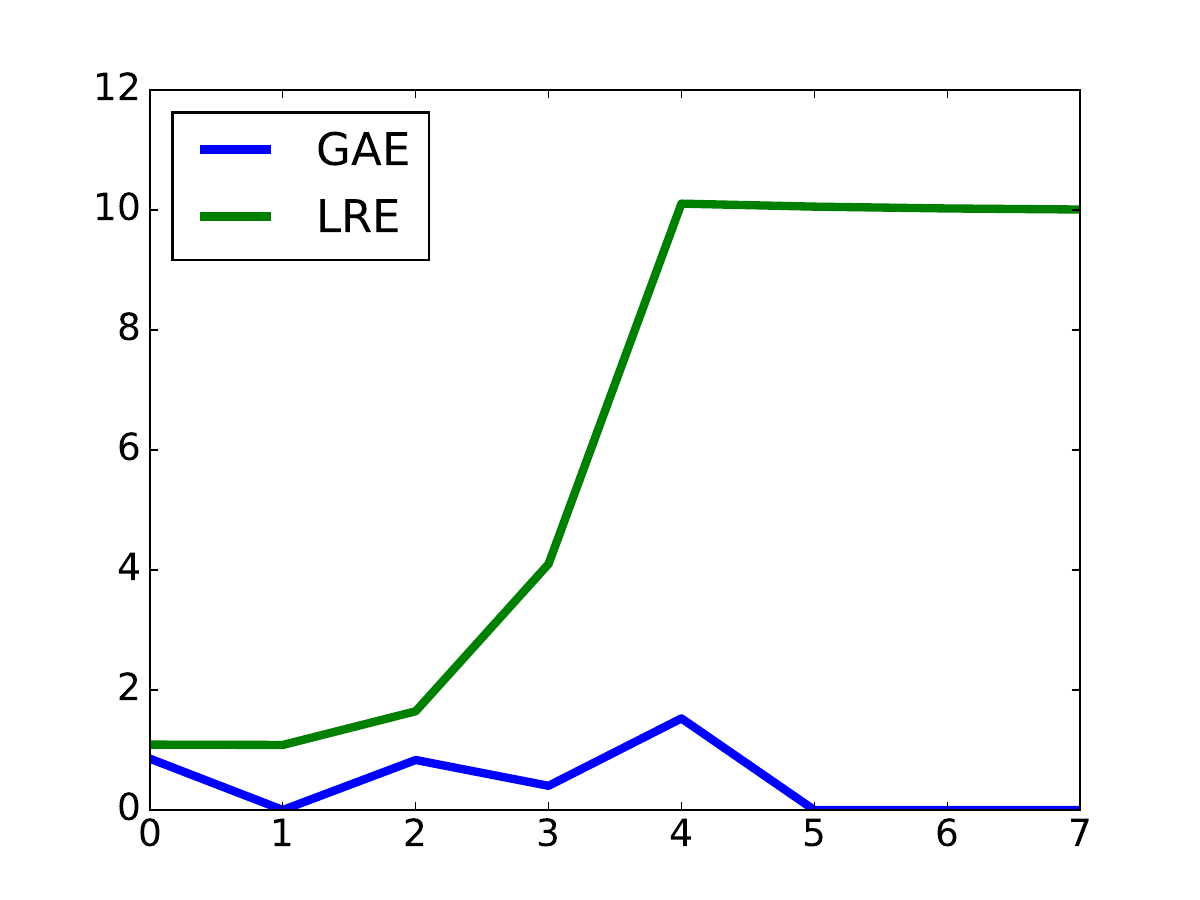} } 
		\end{subfigure}
		\begin{subfigure}[$ GbE_{LRS}, GbE_{DPS} $ and $ GbE_{GAS} $ efficiencies.]
			{\includegraphics[scale = 0.380]{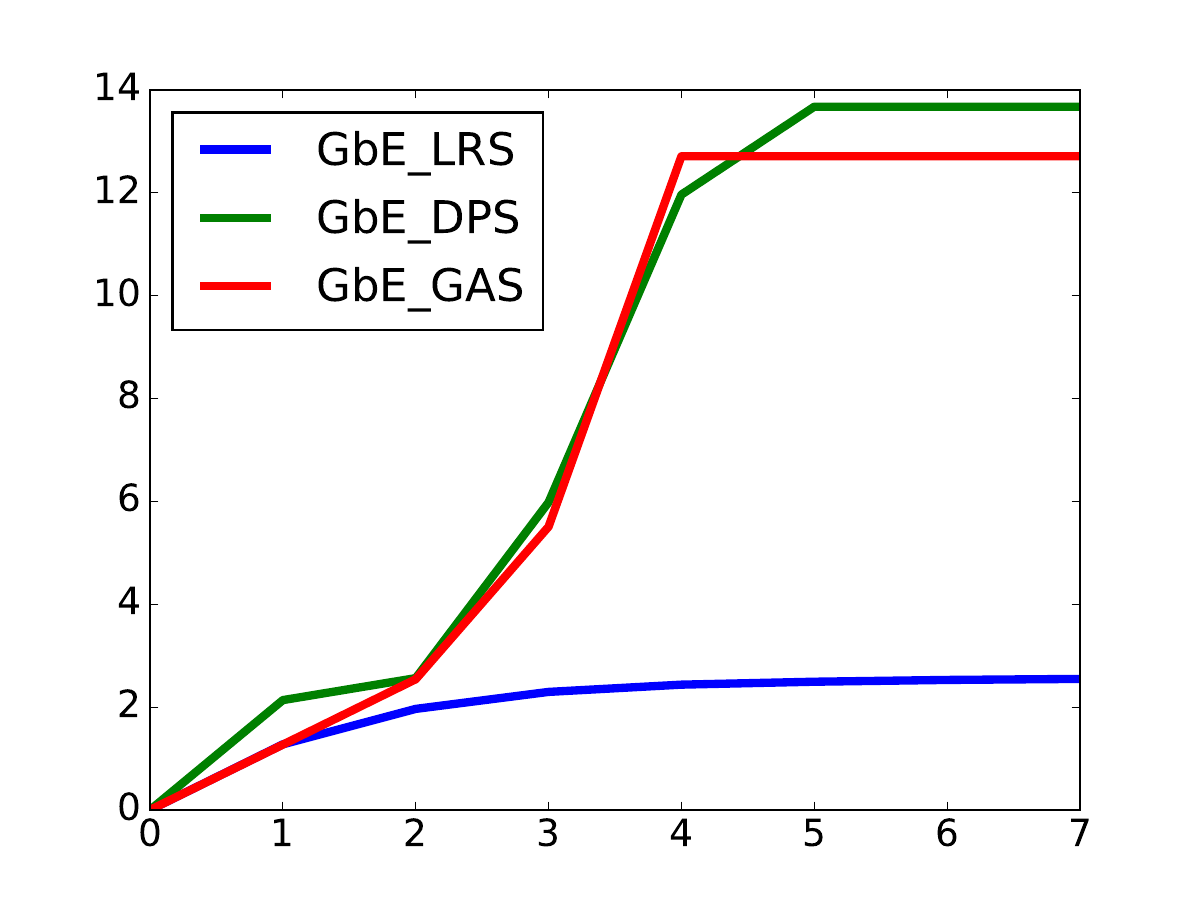} } 
		\end{subfigure}
		~ 
		\begin{subfigure}[$ SwE_{LRS}, SwE_{DPS} $ and $ SwE_{GAS} $ efficiencies.]
			{\includegraphics[scale = 0.380]{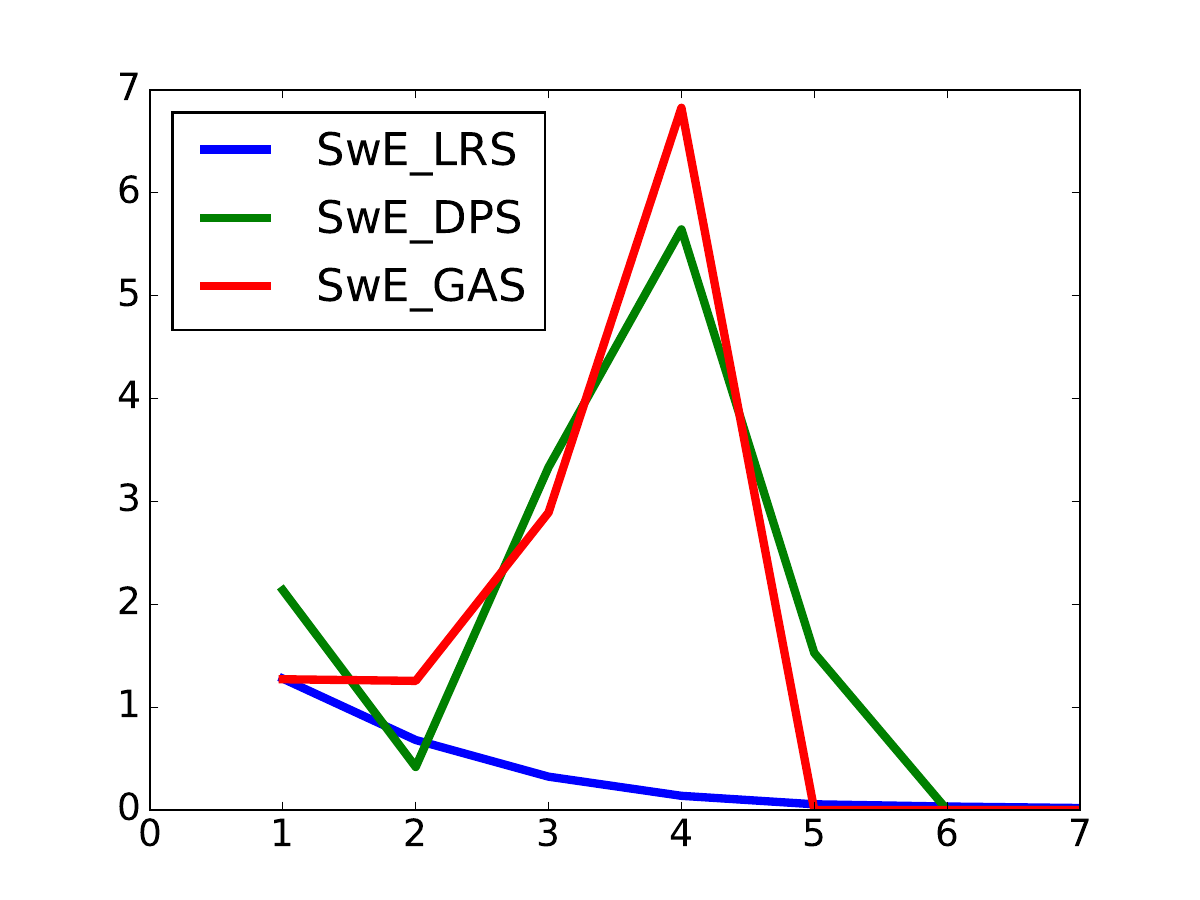} } 
		\end{subfigure}
		\caption{\textsc{Example} \ref{Exm DC tree random realization}. Random realization of 128 eligible items, with uniformly distributed capacities and demand-capacity fraction of 0.9. The D\&C tree has height 7, it is generated by the head-left algorithm \ref{Alg Head-Left Tree Generation}, sorted by specific weight $ \boldsymbol{\gamma} $, left subtree fraction $ f = 0.5 $ and minimum size $ m = 1 $. In Figure (a) the $ y $-axis is expressed in absolute values while in figures (b), (c) and (d) the $ y $-axis is a percentage.}
		\label{Fig Efficiencies through tree} 
	\end{figure}
	\noindent Finally, in \textsc{Figure} \ref{Fig Shorten Five Random Realizations} we present the efficiencies $ GbE_{DPS}, SwE_{DPS}, GAE$ and $ LRE $ for five random realizations. We choose depicting this efficiencies because the Exact Solution ($ DPS $) is the most important parameter, as it measures the quality of the exact solution and the $ GAE $, $ LRE $ efficiencies store the quality of the usual bounds (Greedy Algorithm and Linear Relaxation). The realizations are generated with the same parameters of the previous one (therefore comparable to it) and follow similar behavior amongst them as expected. In particular, notice that for $ h \geq 4 $ (subproblems of size $ 8 $ or smaller) the solutions stabilize. 
	\begin{figure}[h!]
		\centering
		\begin{subfigure}[$ LRE $ efficiencies.]
			{\includegraphics[scale = 0.380]{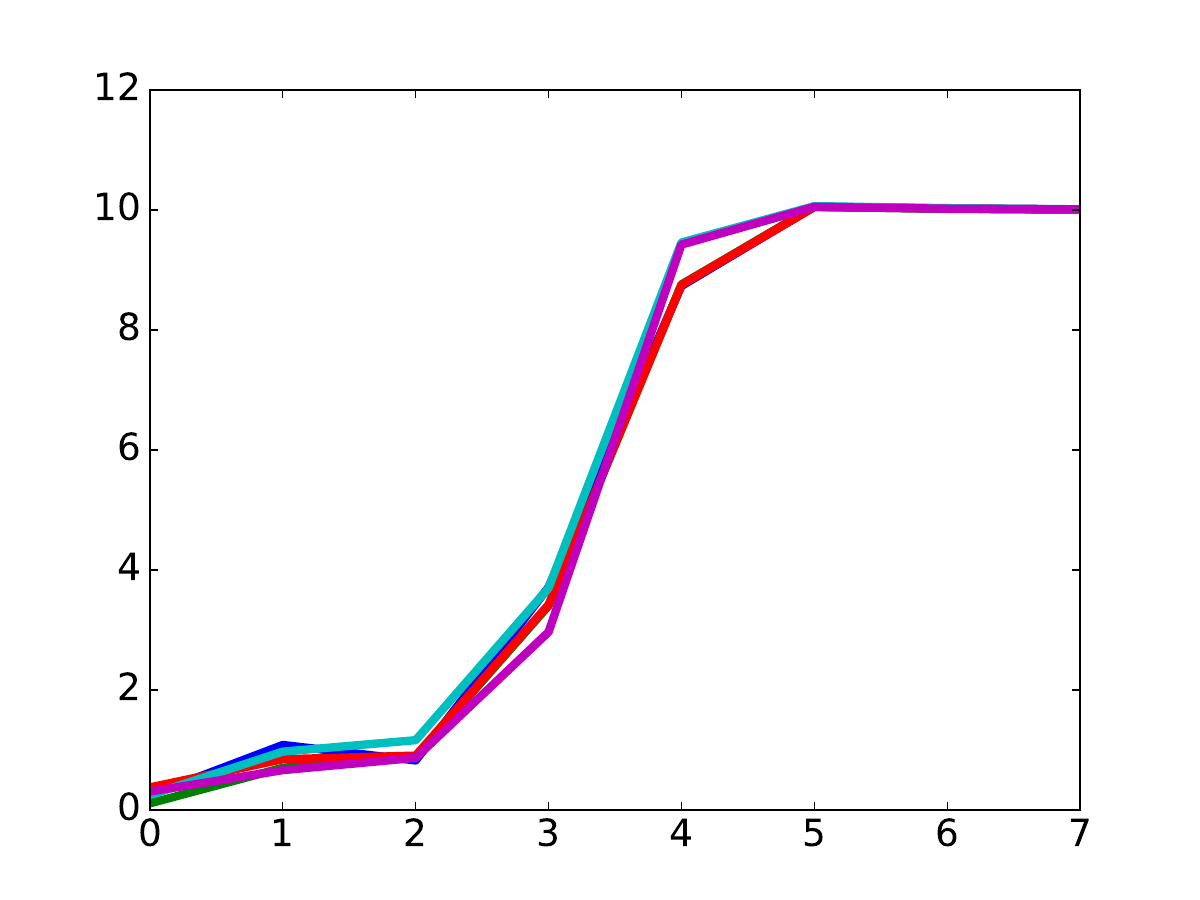} } 
		\end{subfigure}
		~ 
		\begin{subfigure}[$ GAE $ efficiencies. (Different line styles are used to help the reader identify the path that each random realization follows.)]
			{\includegraphics[scale = 0.38]{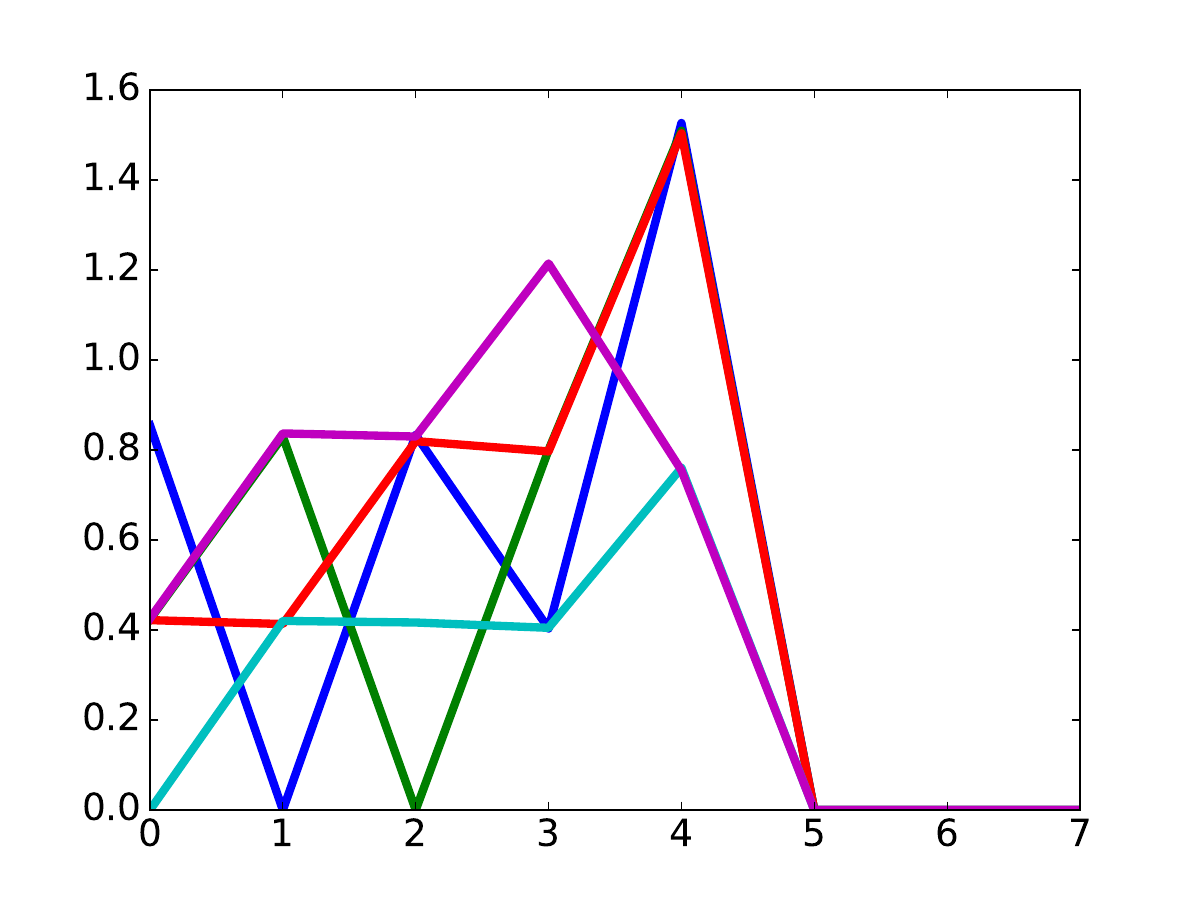} } 
		\end{subfigure}
		\begin{subfigure}[$ GbE_{DPS} $ efficiencies.]
			{\includegraphics[scale = 0.380]{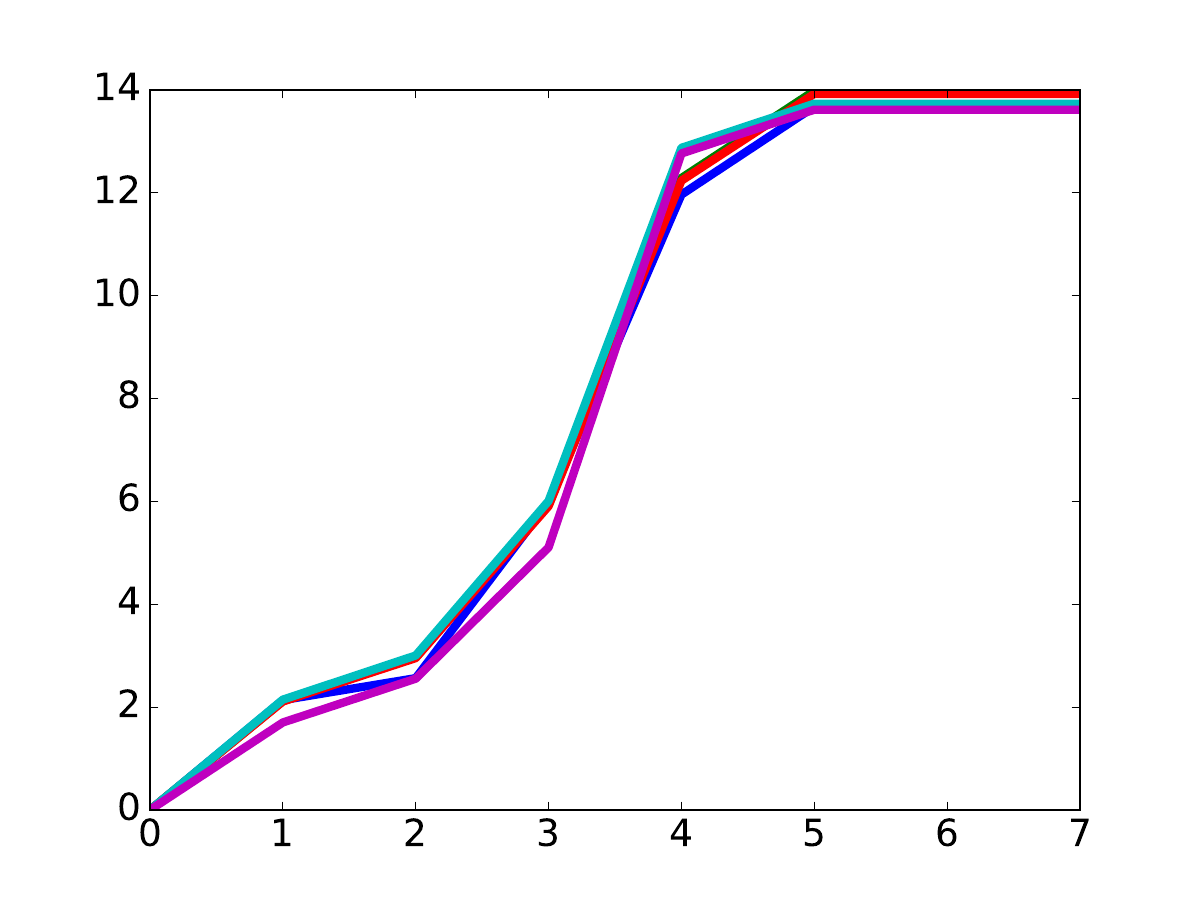} } 
		\end{subfigure}
		~ 
		\begin{subfigure}[$ SwE_{DPS} $ efficiencies.]
			{\includegraphics[scale = 0.380]{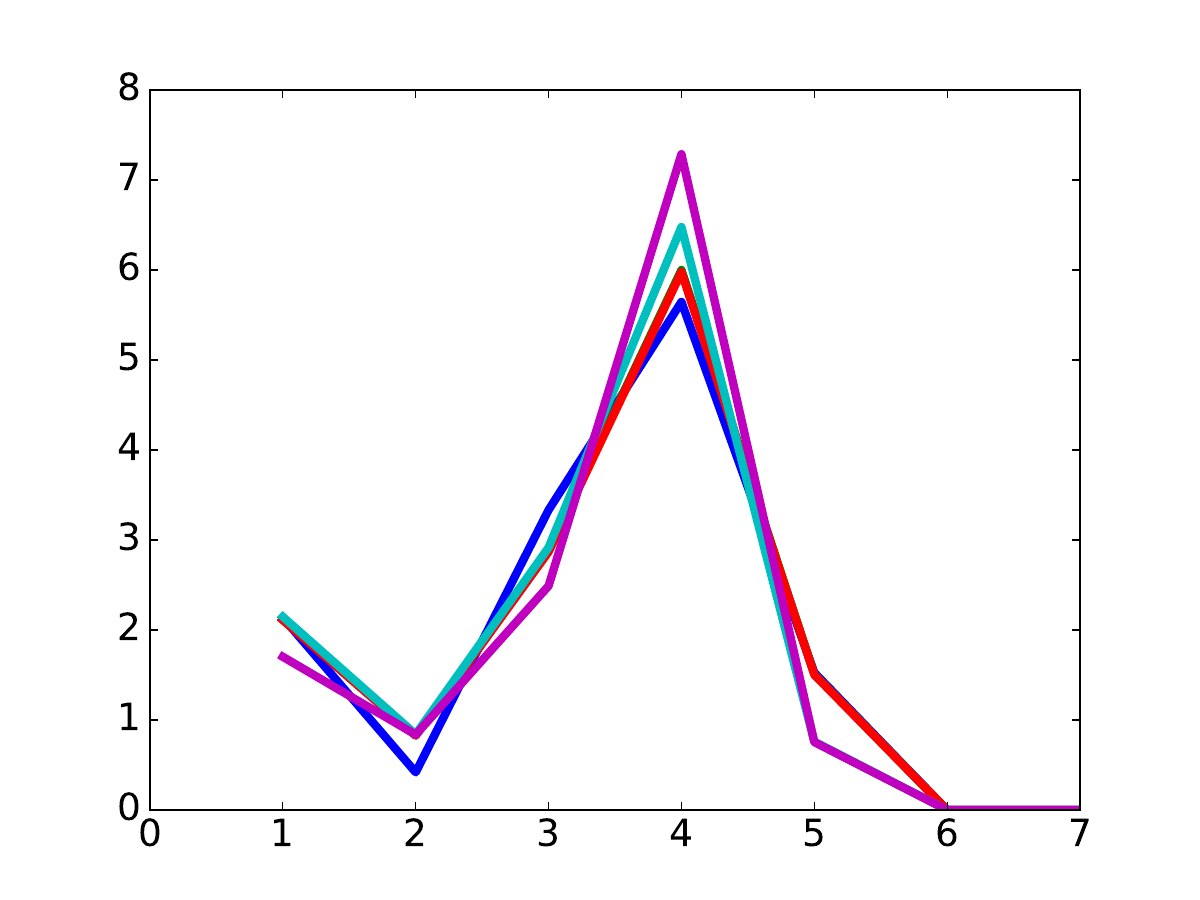} } 
		\end{subfigure}
		\caption{\textsc{Example} \ref{Exm DC tree random realization}. Five random random realization of 128 eligible items, with uniformly distributed capacities and demand-capacity fraction of 0.9. The D\&C tree has height 7, it is generated by the head-left algorithm \ref{Alg Head-Left Tree Generation}, sorted by specific weight $ \boldsymbol{\gamma} $, left subtree fraction $ f = 0.5 $ and minimum size $ m = 1 $.}
		\label{Fig Shorten Five Random Realizations} 
	\end{figure}
\end{example}
\begin{remark}\label{Rem Global behavior comments}
	Examples of 128 eligible items, with a large number of realizations and different distributions (uniform, binomial, Poisson) present similar behavior to the one presented in \textsc{Example} \ref{Exm DC tree random realization}. For the three distributions, most of the results stabilize for $ h \geq 4 $ (subproblems of $ 8 $ items).
\end{remark}
%
%
%
\section{Numerical Experiments}
%
%
%
%
In this section, we present the results from the numerical experiments. All the codes needed for the present work were implemented in Python 3.4 and the databases were handled with Pandas (Python Data Analysis Library). The full scale experiments were run in the server Gauss at Universidad Nacional de Colombia, Sede Medell\'in, Facultad de Ciencias. The Script can be downloaded from the address \url{https://sites.google.com/a/unal.edu.co/fernando-a-morales-j/home/research/software}
%
%
%
%
\subsection{The Experiments Design}
%
%
The numerical experiments are aimed to asses the effectiveness of the heuristic D\&C method presented in \textsc{Section} \ref{Sec The Heuristic Method}. Its whole construction was done in a way such that its effectiveness could be analyzed under the probabilistic view of the Law of Large Numbers (which we write below for the sake of completeness, its proof and details can be found in \cite{BillingsleyProb}). 
\begin{theorem}[Law of Large Numbers]\label{Th the Law of Large Numbers}
	Let $ \big(\Z^{(n)}:n\in \N\big) $ be a sequence of independent, identically distributed random variables with expectation $ \Exp\big(\Z^{(1)}\big) $, then
	\begin{equation}\label{Eq the Law of Large Numbers}
	\prob\bigg[ \Big\vert \frac{\Z^{(1)} + \Z^{(2)} + \ldots + \Z^{(n)}}{n} - 
	\Exp\big(\Z^{(1)}\big) \Big\vert > 0 \bigg]
	\xrightarrow[n\, \rightarrow \,\infty]{} 0 ,
	\end{equation}
	i.e. , the sequence $ \big(Z^{(n)}:n\in \N\big) $  converges to $ \mu $ in the Ces\`aro sense.
\end{theorem}
%
%
%
The D\&C method introduces several free/decision parameters to analyze the behavior of \textsc{Problem} \ref{Pblm Integer Problem} under different scenarios. We have the following list of domains for each of these parameters
\begin{enumerate}[a.]
	\item Number of items: $ N \in \N $.
	
	\item Distribution of items' capacities: $ \dist \in \{ \text{Ud, Pd, Bd} \} $ (Ud: uniform, Pd: Poisson, Bd: Binomial). 
	
	\item Demand-Capacity fraction: $ o \in \{0,50, 0.55, \ldots,  0,90\} $ (to satisfy hypotheses of (iv) \textsc{Theorem} \ref{Thm Feasibility of DC}).
	
	\item Price-Capacity rate: $ r \in \{34, 44, 54, 64, 74 \} $ (to avoid hypotheses of \textsc{Theorem} \ref{Thm Quality of the Greedy Algorithm} been satisfied).
	
	\item D\&C tree algorithm $ \talg \in \{\text{hlT, blT} \} $ (hlT head-left Tree \textsc{Algorithm} \ref{Alg Head-Left Tree Generation}, blT balanced-left Tree \textsc{Algorithm} \ref{Alg Balanced-Left Tree Generation}).
	
	\item Eligible Items list sorting method: $ s \in \{ \p, \capac, \boldsymbol{\gamma} , \text{random} \} $.
	
	\item Fraction of the left list: $ f\in \{ 0.35, 0.40, \ldots, 0.65 \} $.
	
	\item Minimum list size: $ m\in \N $.
\end{enumerate}
\begin{remark}[Parameters Domains]
	It is clear that $ o $ and $ f $ could very well adopt any value inside the interval $ [0.1] $, while $ r $ could be any arbitrary number in $ \N $. However, adopting such ranges is impractical for two reasons. First, their infinite nature prevents an exhaustive exploration as we intend to do. Second, most of the values in such a large range are unrealistic. For instance: $ o = 0.1 $ means that the capacity of available items is 10 times the demand (scenario that will hardly occur in real-world problems), $ f = 0 $ means no D\&C pair was introduced and $ r\geq \max_{i\,\in\,[N]} c_{i} $ means that all the items have the same price regardless of their capacity. 
\end{remark}
In order to model, an integer problem of type \ref{Pblm Integer Problem} and its D\&C solution as random variables, we need to introduce the following definition
\begin{definition}\label{Def Probabilistic View}
	Consider the following probabilistic space and random variables.
	\begin{enumerate}[(i)]
		\item Denote by $ \Omega $ the set of all possible integer problems of the type \ref{Pblm Integer Problem}.
		
		\item Define the \textbf{random problem generator} variable as
		\begin{equation}\label{Eqn Problem Random Variable}
		\begin{split}
		\X : \N \times 
		\{\text{Ud, Pd, Bd} \}\times
		\{0,50, 0.55, \ldots,  0,90\} \rightarrow & \Omega\\
		(N, \dist, o) \mapsto & \X(N, \dist, o) .
		\end{split}
		\end{equation}
		Here, $\X(N, \dist, o) $ is an integer problem of type \ref{Pblm Integer Problem}.
		
		\item Define the \textbf{D\&C solution variable} by
		\begin{equation}\label{Eqn Solution Problem Random Variable}
		\begin{split}
		\ups  : \Omega \times \{34, 44,\ldots , 74 \}  \times
		\{\text{hlT, blT} \} \times 
		\{ 0.35, 0.40, \ldots, 0.65 \} \times \\
		\{\p, \capac, \boldsymbol{\gamma}, \text{random} \}\times
		\N  \rightarrow &  \bigcup\limits_{h\,\in \, \N} \N^{h}\\
		(\X, r, \talg, s, f, m )\mapsto & \ups(\X, r, \talg, s, f, m). 
		\end{split}
		\end{equation}
		In the expression above, it is understood that $ \X = \X(N, dist, o) $ is the random problem generator variable and $\ups(\X, r, \talg, s, f, m) $ indicates the solution for the chosen integer problem $ \X \in \Omega $, under the D\&C tree solution parameters $ r, \talg, s, f, m $. This is, a stack/vector of solutions in $ \N^{\uph} $ where $ \uph $ is the height of the constructed D\&C tree. In particular, notice that $ \uph $ is also a random variable.
	\end{enumerate}
\end{definition}
Notice that if the parameters $ N, s, m $ are fixed, then, $ \uph $ is constant and the D\&C solutions random variable $\ups\big(\X(N, dist, o), r, \talg, s, f, m\big)\in \N^{\uph} $. However, a Monte Carlo simulation analysis can not be applied under these conditions, because the realizations of the random variable $ \ups $ would not meet the hypotheses of the Law of Large Numbers \ref{Th the Law of Large Numbers}, more specifically, the identically distributed condition. On the other hand, the analysis is pertinent for several realizations of the random variables $ \X $ and $ \ups $, with a fixed list of free/decision parameters, namely $ P = (N, dist, o, r, \talg, s, f, m) $. Under these conditions the Law of Large Numbers  can be applied on $ \ups $ to estimate the expected effectiveness of the method, conditioned to the chosen set of parameters $ P $. \\

\noindent In order to compare the different scenarios without introducing too many possibilities a standard setting has to be defined, which we introduce below, together with the justification behind its choice.
\begin{definition}\label{Def Standard Setting}
	In the following we refer to the \textbf{standard setting} of a numerical experiment 

	\begin{equation*}
	P = \begin{cases}
	(N, dist, o, r, \talg, s, f, m),  & \text{for }  \talg = \text{hlT}, \\
	(N, dist, o, r, \talg, s, m), & \text{for } \talg = \text{blT} ,
	\end{cases}
	\end{equation*}%
	if its parameters satisfy the following values:
	\begin{enumerate}[(i)]
		\item Head Fraction, $ f = 0.5 $ (applies for the head-left method only). To make it comparable with the balanced method.
		
		\item Demand-Capacity Fraction, $ o = 0.9 $. 
		
		\item Price-Capacity rate, $ r = 54 $. From experience, this is a reasonable value, as it permits explore problems of computable size without landing into trivial scenarios.
		
		\item Eligible Items list sorting method, $ s = \boldsymbol{\gamma} $ i.e., specific weight. Because this greedy function is closely related to the solutions furnished by the linear relaxation (LRS), presented in \textsc{Theorem} \ref{Thm Simplex and Greedy}, as well as the Greedy Algorithm \ref{Alg Greedy Algorithm}.
		
		\item Minimum list size, $ m = 4 $. From multiple random realizations, it has been observed that the D\&C method does not yield significantly different results for list sizes smaller than $ m = 8 $; see \textsc{Remark} \ref{Rem Global behavior comments}. Consequently, we adopt the size $ m = 4 $ in order to capture one step (and only one) of this ``steady behavior".
		
		\item Number of eligible items, $ N = 512 $. This size was chosen because for $ m = 4 $ it will produce in most of the studied cases a D\&C tree of height $ 7 $. The only exceptions will occur for head-left generated trees with head fraction $ f \neq 0.5 $.
		
	\end{enumerate}
	In addition the next conventions are adopted
	\begin{enumerate}[a.]
		\item An experiment is defined by a list of parameters, namely $ P $; from now on we do not make a distinction between the experiment and its list of parameters. Moreover, $ P $ has 8 parameters if $ \talg = \text{hlT} $ and 7 if $ \talg = \text{blT} $. To ease notation, from now on we denote $ P = (512, \dist, o, r, \talg, s, f, 4) $ for any experiment in general, in the understanding that if $ \talg = \text{blT} $ the head fraction $ f $ is not present in the list $ P $.
		
		\item Each case will be analyzed using $ 50 $ randomly generated realizations of $ 512 $ items with Uniform, Poisson and Binomial distributions respectively i.e., $ P = (512, dist, o, r, \newline \talg, s, f, 4) $, see \textsc{Figure} \ref{Fig Branch Strategies Tree}. 
		
		\item Given a standard setting $ P = (512, dist, o, r, \talg, s, f, 4) $ and a variable $ v \in \{ o, r, s, f \} $, we denote by $ P(v) $ the list of experiments where the variable $ v $ runs through its whole domain, see \textsc{Table} \ref{Tbl 50 realizations example} and \textsc{Figure} \ref{Fig Averaged Efficiencies Example} .
		
		\item The analysis of the efficiencies $ 
		GAE, LRE , GbE_{LRS}, GbE_{DPS},  GbE_{GAS}, 
		SwE_{LRS} $, $ SwE_{DPS} $ and $ SwE_{GAS}
		$ 
		will be done using their average values, corresponding to the 50 random realizations mentioned above. 
		In the following, we denote by $ \mathcal{E} $ the list ot these efficiencies;
		due to the Law of Large Numbers \ref{Th the Law of Large Numbers} we know this is an approximation of their expected values. An example is presented in \textsc{Table} \ref{Tbl 50 realizations example} and \textsc{Figure} \ref{Fig Averaged Efficiencies Example} below.
	\end{enumerate}
\end{definition}
\begin{figure}
	\centering
	\begin{subfigure}[]
		\centering
		\begin{tikzpicture}
		[scale=.9,auto=left,every node/.style={}]
		\node (n0) at (6,6) {
			$ \begin{pmatrix}
			\textsc{Standard Setting}\\[3pt]
			N = 512,\; m = 4 \\[3pt]
			\text{dist} \in \{ \text{Ub, Pd, Bd} \} \\[3pt]
			\textit{T-alg} \in \{ \text{hlT, blT} \}\\[3pt]
			o = 0.9, \; r = 54, \, s = \gamma, \; f = 0.5
			\end{pmatrix}  $
		};
		\node (n1) at (2,2)  {$ \begin{pmatrix}
			\textsc{Standard Setting}\\[3pt]
			N = 512,\; m = 4 \\[3pt]
			\text{dist} \in \{ \text{Ub, Pd, Bd} \}\\[3pt]
			\textit{T-alg} =  \text{hlT} \\[3pt]
			o = 0.9, \; r = 54, \, s = \gamma, \; f = 0.5
			\end{pmatrix} $};
		\node (n2) at (0,-0.5)  {};
		\node (n3) at (4,-0.5)  {};
		\node (n7) at (2,-0.5) {};
		\node (n4) at (10,2)  {$ \begin{pmatrix}
			\textsc{Standard Setting}\\[3pt]
			N = 512,\; m = 4 \\[3pt]
			\text{dist} \in \{ \text{Ub, Pd, Bd} \}\\[3pt]
			\textit{T-alg} =  \text{blT} \\[3pt]
			o = 0.9, \; r = 54, \, s = \gamma, \; f = 0.5
			\end{pmatrix} $};
		\node (n5) at (8,-0.5)  {};
		\node (n6) at (12,-0.5)  {};
		\node (n8) at (10, -0.5) {};
		
		\foreach \from/\to in {n0/n1, n1/n2,n1/n3, n1/n7, n0/n4, n4/n5, n4/n6, n4/n8}
		\draw (\from) -- (\to);
		\end{tikzpicture}
		%
	\end{subfigure}
	%
	\begin{subfigure}[]
		\centering 
		\begin{tikzpicture}
		[scale=.9,auto=left,every node/.style={}]
		\node (n0) at (5,6) {
			$ \begin{pmatrix}
			\textsc{Balanced-Left Method}\\[3pt]
			N = 512,\; m = 4 \\[3pt]
			\text{dist} \in \{ \text{Ub, Pd, Bd} \} \\[3pt]
			\textit{T-alg} \in \{ \text{hlT, blT} \}\\[3pt]
			o = 0.9, \; r = 54.
			\end{pmatrix}  $
		};
		\node (n1) at (0,2)  {$ \begin{pmatrix}
			\textsc{Balanced-Left Method}\\[3pt]
			N = 512,\; m = 4 \\[3pt]
			\text{dist} \in \{ \text{Ub, Pd, Bd} \}\\[3pt]
			\textit{T-alg} =  \text{blT} \\[3pt]
			o = 0.9, \; r = 54.
			\end{pmatrix} $};
		\node (n2) at (5,2)  {$ \begin{pmatrix}
			\textsc{Balanced-Left Method}\\[3pt]
			N = 512,\; m = 4 \\[3pt]
			\text{dist} \in \{ \text{Ub, Pd, Bd} \}\\[3pt]
			\textit{T-alg} =  \text{blT} \\[3pt]
			o = 0.9, \; r = 54.
			\end{pmatrix} $};
		\node (n3) at (10,2)  {$ \begin{pmatrix}
			\textsc{Standard Setting}\\[3pt]
			N = 512,\; m = 4 \\[3pt]
			\text{dist} \in \{ \text{Ub, Pd, Bd} \}\\[3pt]
			\textit{T-alg} =  \text{blT} \\[3pt]
			o = 0.9, \; r = 54, \, s = \gamma, \; f = 0.5
			\end{pmatrix} $  };
		\node (n4) at (-2, -0.5)  {};
		\node (n5) at (0, -0.5)  {};
		\node (n6) at (2, -0.5)  {};
		\node (n7) at (3,-0.5) {};
		\node (n8) at (5, -0.5) {};
		\node (n9) at (7,-0.5) {};
		\node (n10) at (8,-0.5) {};
		\node (n11) at (10, -0.5) {};
		\node (n12) at (12,-0.5) {};
		
		\foreach \from/\to in {n0/n1, n0/n2, n0/n3, n1/n4, n1/n5, n1/n6, 
			n2/n7, n2/n8, n2/n9, n3/n10, n3/n11, n3/n12}
		\draw (\from) -- (\to);
		\end{tikzpicture}
		%
	\end{subfigure}
	%
	%
	\caption{Schematics of the set of numerical experiments in search of optimal strategies. The first level, depicted in Figure (a), branches on the tree generation method: $ \text{lhT} $ and $ \text{blT} $. The second level branches on the remaining strategies: $ o $, $ r $, $ s $ for both $ \{ \text{lhT}, \text{blT}\} $ and $ f $ for the $ \text{lhT} $ method. Figure (b) displays the branching process for the $ \text{blT} $ method; a similar diagram corresponds for the $ \text{lhT} $ method.}
	\label{Fig Branch Strategies Tree} 
\end{figure}
%
%
\subsection{Critical Height and $ \text{hlT} $ vs. $ \text{blT} $ 
	strategies Comparison}\label{Sec Critical Height}
%
%
As a first step we find a critical height. From the numerical experiments, it is observed that the method heavily deteriorates beyond certain height i.e., after certain number of D\&C iterations, as it can be seen in the 
figures \ref{Fig Efficiencies through tree} and \ref{Fig Shorten Five Random Realizations} from \textsc{Example} \ref{Exm DC tree random realization}, where it can be observed that beyond $ h > 3 $ the slope becomes very steep, therefore a critical height needs to be adopted.
\begin{definition}\label{Def Critical Height}
	Given an experiment of 50 realizations with a fixed set of parameters $ P = (512, dist, o, r, \talg, s, f, 4) $ and let $ v\in \{ o, r, s, f \} $ be a variable running through its full domain. Define
	\begin{enumerate}[(i)]
		\item For a fixed efficiency $ \eff \in \mathcal{E} $, denote respectively $ \bar{\ups}(\eff, P) $, $ \bar{\ups}(\eff, P, v) $, the average value of 50 random realizations executed with parameters $ P $ and the list of such values when the variable $ v $ runs through its whole domain, see \textsc{Table} \ref{Tbl 50 realizations example} and \textsc{Figure} \ref{Fig Averaged Efficiencies Example} below.
		
		\item For a fixed efficiency $ \eff \in \mathcal{E} $, denote by 
		\begin{align*}
		& \bar{\ups}' (\eff, P, v)(h) \defining \bar{\ups}(\eff, P, v)(h) - \bar{\ups}(\eff, P, v)(h -1),  &
		& \text{for }  h = 1, 2, \ldots, H ,
		\end{align*}
		with $ H$ the height of the D\&C tree. Denote by $ \bar{\ups}'_{v} (\eff, P)(h) = \max\{ \bar{\ups}' (\eff, P, v)(h): v\,\in \, \text{full domain}\}  $.
		
		\item For each of the efficiencies $ \eff \in \mathcal{E} $, its critical height relative to the variable $ v $, denote by $ h_{v}(\eff, P) $ the last height $ h $ satisfying $ \bar{\ups}'_{v} (\eff, P)(h) \leq 2 \bar{\ups}'_{v} (\eff, P)(h - 1)  $, see \textsc{Table} \ref{Tbl 50 realizations example} and \textsc{Figure} \ref{Fig Averaged Efficiencies Example}. 
		
		\item The critical height of the experiment relative to the variable $ v $, denoted by $ h_{v}(P) $ is given by the mode of the list $ \{ h_{v}(\eff, P): \eff \in \mathcal{E} \} $.
		
		
		\item In order to compare the experiments $ P_{\text{hlT}} = (512, dist, o, r, \text{hlT}, s, 0.5, 4) $ and $ P_{\text{blT}} = (512, dist, o, r, \text{lbT}, s, 4) $ (head-left vs balanced), relative to the variable $ v $, we proceed as follows: set the height $ \tilde{h} \defining \min\{ h_{v}(P_\text{{hlT}}), h_{v}(P_\text{{blT}}) \} 
		$ (see \textsc{Table} \ref{Tbl Critical Heights Table Optimal Head Fraction}) and compute the $ \ell^{1}$-norm for the arrays $ \{\bar{\ups}(\eff, P_{\text{hlT}}, v)(h): \eff \in \mathcal{E},\, h = 1, 2, \ldots, \tilde{h}\}  $, $ \{\bar{\ups}(\eff, P_{\text{blT}}, v)(h): \eff \in \mathcal{E}, \, h = 1, 2, \ldots, \tilde{h}\}  $, when regarded as lists (not as matrices, as \textsc{Table} \ref{Tbl 50 realizations example} would suggest). The lowest of these norms yields the best strategy among $ \text{hlT} $ and $ \text{blT} $.
		
	\end{enumerate}
\end{definition}
\begin{example}\label{Exm 50 realizations}
	In the table \ref{Tbl 50 realizations example} below we display $ \{\bar{\ups}(GbE_{DPS}, P, r)(h): h = 0, 1, \ldots, 7\} $ i.e., the averaged values corresponding to 50 realizations for the efficiency $ \eff = GbE_{DPS} $ running through the full domain of the price-capacity rate i.e., $ v = r $. The list of parameters is given by $ P = (512, \text{Ud}, 0.9, r, \text{hlT}, \boldsymbol{\gamma}, 0.5, 4) $ with $ r \in \{ 34, 44, 54, 64, 74 \} $. The tables corresponding to the intermediate slope variables $ \{\bar{\ups}' (P, v)(h): h  = 0, 1, \ldots 7\}$, $ \{\bar{\ups}'_{v} (P)(h): h  = 0, 1, \ldots 7\}$ are omitted since they can be completely deduced from \textsc{Table} \ref{Tbl 50 realizations example}. In this particular example $ h_{v}(\eff)  = h_{r}(GbE_{DPS}) = 5 $. 
	\begin{table}[h!]
		\begin{centering}
			\rowcolors{2}{gray!25}{white}
			\begin{tabular}{c c c c c c }
				\hline
				\rowcolor{gray!50}
				Height &  $ r = 34 $ & $ r = 44 $ &$ r = 54 $ & $ r = 64 $ & $ r = 74 $ 
				\\
				\hline
				0 &	0.00 &	0.00 &	0.00 &	0.00 &	0.00 
				\\
				1 & 1.26 &	1.93 &	1.95 &	2.29 &	1.82 
				\\
				2 & 2.11 &	3.19 &	3.40 &	3.45 &	2.74 
				\\
				3 &	2.66 &	3.91 &	4.27 &	4.34 &	3.53 
				\\
				4 & 3.15 &	4.52 &	4.85 &	4.92 &	4.17 
				\\
				\ding{217} $ h_{r}(P) = 5 $ &	3.93 &	5.68 &	6.63 &	7.08 &	6.24 
				\\
				6 &	7.72 &	9.56 &	10.11 &	9.99 &	8.74 
				\\
				7 & 14.11 &	15.50 &	15.70 &	15.66 &	14.83 
				\\
				\hline
			\end{tabular}
			\caption{Average values of 50 random realizations for the efficiency variable $ \eff = GbE_{DPS} $ relative to the variable $ v = r $. The experiments parameters $ P = (N, dist, o, r, \talg, s, f, m) $ have the following values: $ N = 512 $, $ \dist = \text{Ud} $, $ o = 0.9 $, $ r \in \{34, 44, 54, 64, 74\} $, $ \talg = \text{hlT} $, $ s = \boldsymbol{\gamma} $, $f = 0.5 $, $ m = 4 $.}
			\label{Tbl 50 realizations example}
		\end{centering}
	\end{table}
	Finally, the corresponding solution is presented in \textsc{Figure} \ref{Fig Averaged Efficiencies Example} (a), together with its analogous for the efficiencies $ SwE_{DPS}, GAE, LRE $ ((b), (c) and (d) respectively). We chose to present these efficiencies because the Exact Solution behavior $ DPS $, is the central parameter to asses the quality of the method for measuring the quality of the solution, while the efficiencies $ GAE, LRE $ measure the expected quality of the usual bounds (Greedy Algorithm and Linear Relaxation) through the D\&C tree.
	\begin{figure}[h!]
		\centering
		\begin{subfigure}[$ \bar{\ups}(GbE_{DPS}, P, r) $, $ r \in\{34, 44, 54, 64, 74\} $.]
			{\includegraphics[scale = 0.380]{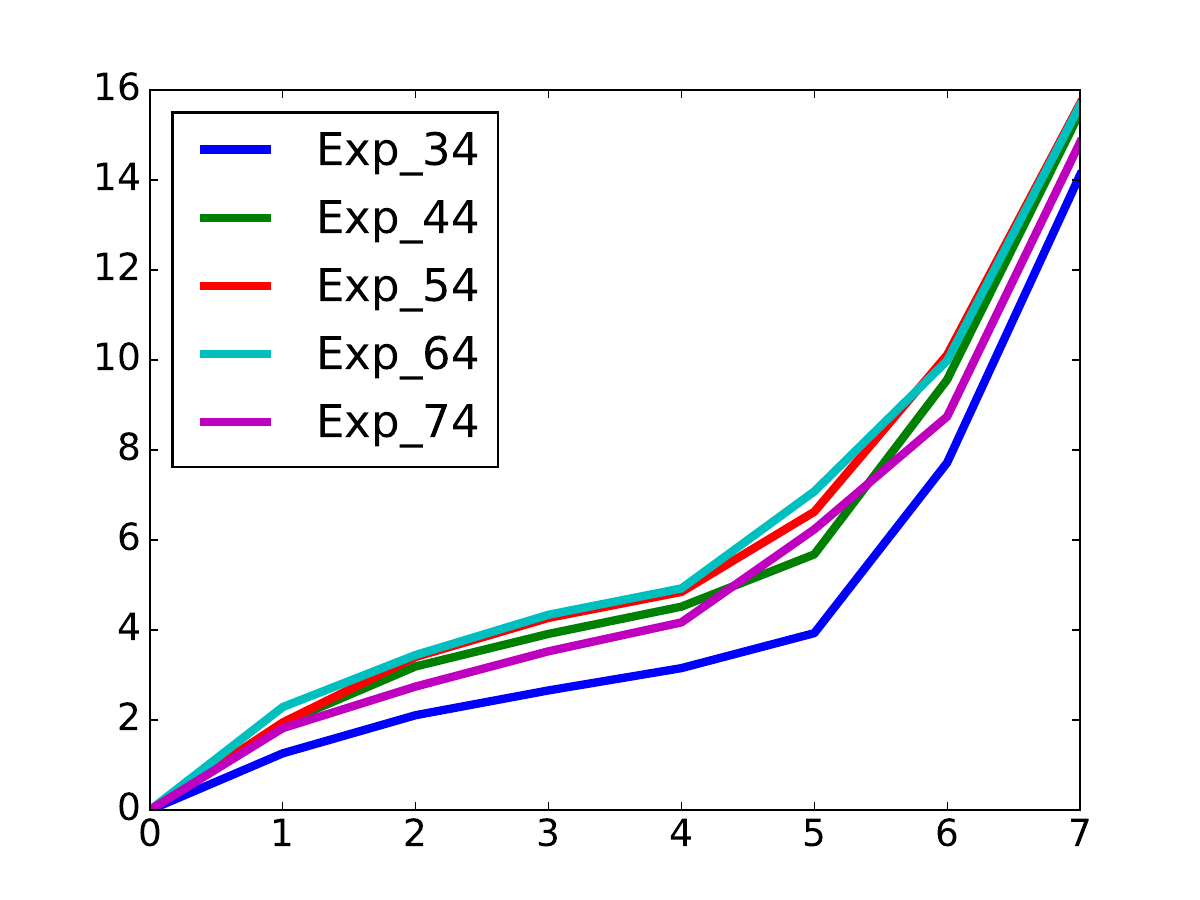} } 
		\end{subfigure}
		~ 
		\begin{subfigure}[$ \bar{\ups}(SwE_{DPS}, P, r) $, $ r \in\{34, 44, 54, 64, 74\} $.]
			{\includegraphics[scale = 0.380]{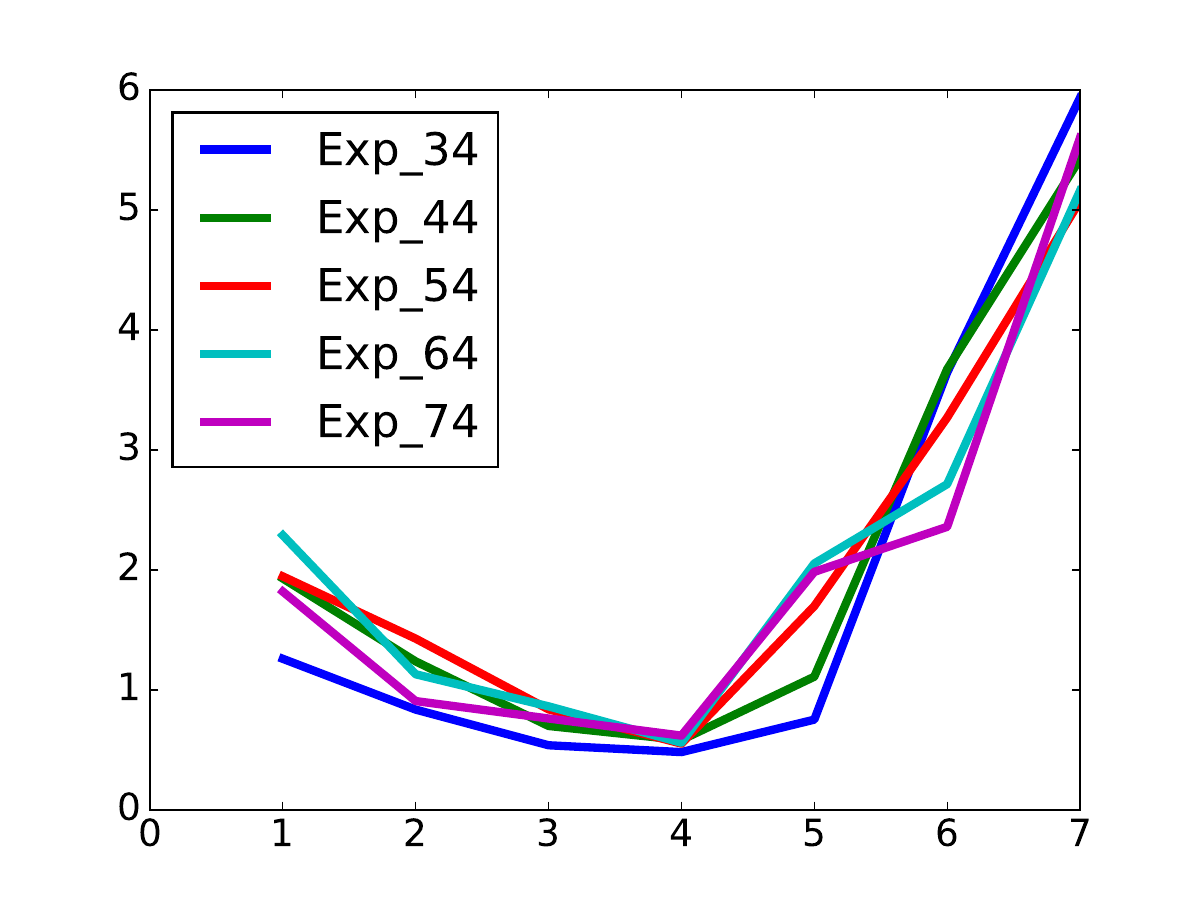} } 
		\end{subfigure}
		\begin{subfigure}[$ \bar{\ups}(GAE, P, r) $, $ r \in\{34, 44, 54, 64, 74\} $.]
			{\includegraphics[scale = 0.380]{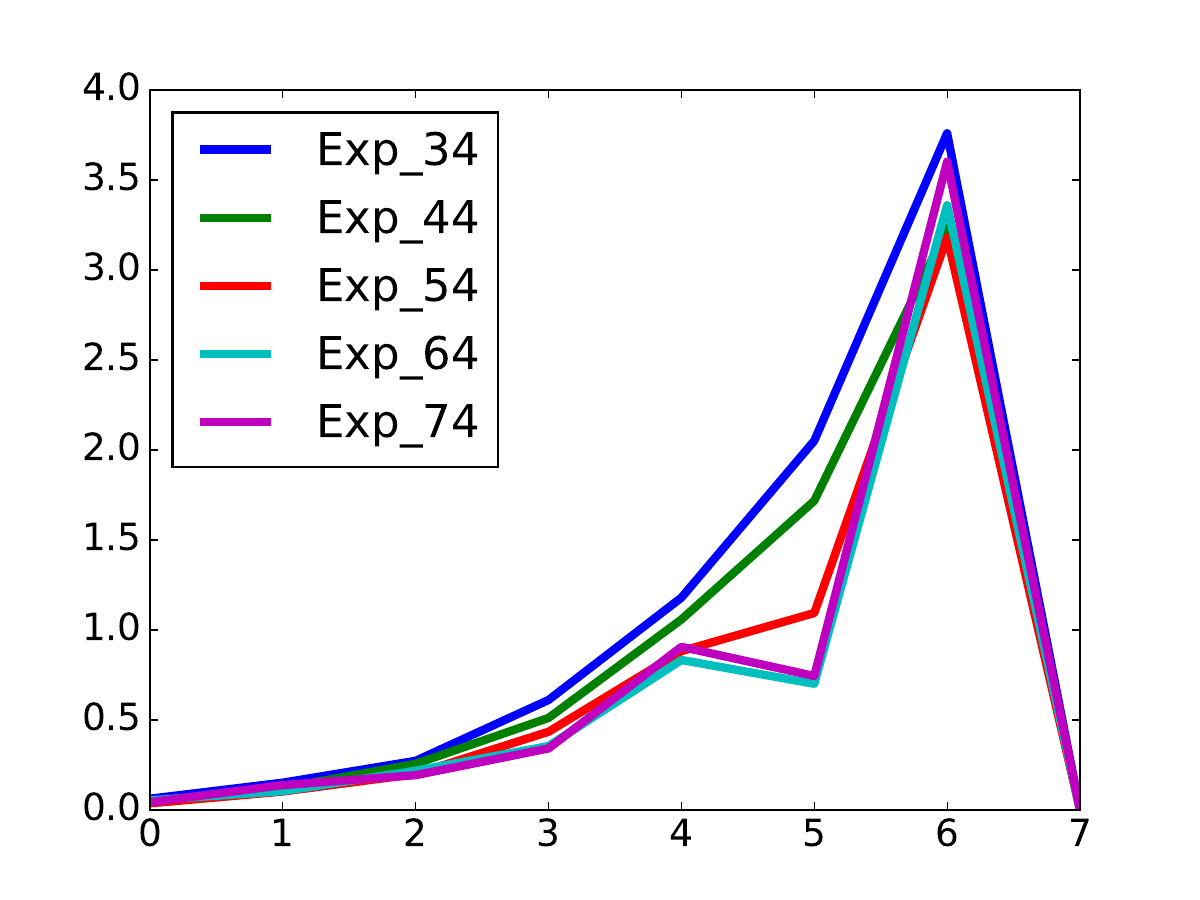} } 
		\end{subfigure}
		~ 
		\begin{subfigure}[$ \bar{\ups}(LRE, P, r) $, $ r \in\{34, 44, 54, 64, 74\} $.]
			{\includegraphics[scale = 0.380]{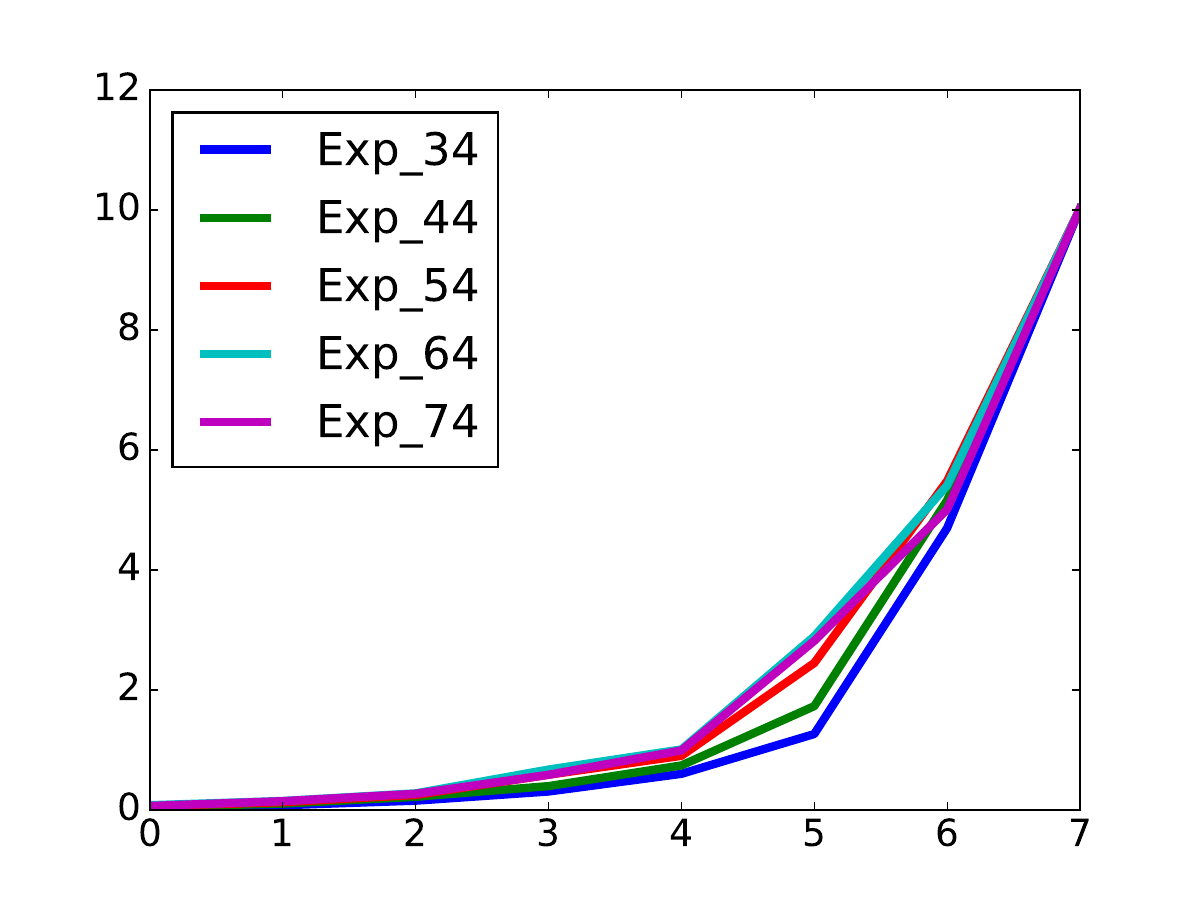} } 
		\end{subfigure}
		\caption{\textsc{Example} \ref{Exm 50 realizations}. Averaged values for 50 random realizations. Four particular efficiencies are depicted $ GbE_{DPS}, SwE_{DPS} , GAE, LRE $. The notation $ \text{Exp}\_{r} $ with $ r\in \{34,,44, 54, 64, 74 \} $ in the graphics' legends, stands for the expected value for the corresponding $ \bar{\ups}(\eff, P, r) $, $ \eff \in \{ GbE_{DPS}, SwE_{DPS} , GAE, LRE\} $.  }
		\label{Fig Averaged Efficiencies Example} 
	\end{figure}
\end{example}

\noindent Given that the aim of this section is to compare the generation methods $ \text{hlT} $ vs. $ \text{blT} $, we first find the optimal head fraction value $ f $ for $ \text{hlT} $, in order to attain the best possible efficiencies for the $ \text{hlT} $ method. The results are summarized in \textsc{Table} \ref{Tbl Head Fraction Method} below; the pointing arrows indicate the optimal head fraction values
\begin{table}[h!]
	\begin{centering}
		\rowcolors{2}{gray!25}{white}
		\begin{tabular}{c c c c c c c }
			\hline
			\rowcolor{gray!50}
			Head Fraction &  \multicolumn{2}{ c }{Uniform}  & \multicolumn{2}{ c }{Poisson} 
			& \multicolumn{2}{ c }{Binomial} \\
			\rowcolor{gray!50}
			$ f $ 
			&  \multicolumn{2}{ c }{Height $ h_{o}(P_{\text{blT}}) = 4 $ }  
			& \multicolumn{2}{ c }{Height $ h_{o}(P_{\text{blT}}) = 3 $} 
			& \multicolumn{2}{ c }{Height $ h_{o}(P_{\text{blT}}) = 3 $} \\
			\rowcolor{gray!50}
			& $ GbE_{DPS}$ &  $ \eff \in \mathcal{E} $ 
			& $ GbE_{DPS}$ &  $ \eff \in \mathcal{E} $  
			& $ GbE_{DPS}$ &  $ \eff \in \mathcal{E} $  \\
			\hline
			0.35 
			& \ding{217} 8.27  & \ding{217} 26.45 & 2.83 & 8.87 & \ding{217} 3.36 & \ding{217} 10.54\\
			0.40 
			& 8.88 & 28.20 & 3.06 & 9.60 & 3.54 & 11.04\\
			0.45 
			& 9.30 & 29.43 & 2.87 & 8.97 & 3.69 & 11.56\\
			0.50 
			& 9.62 & 30.51  & 2.94 & 9.27 & 3.67 & 11.49\\
			0.55 
			& 9.72 & 30.57 & 2.95 & 9.30 & 4.02 & 12.51\\
			0.60 
			& 9.64 & 30.48 & 3.12 & 9.78 & 4.23 & 13.12\\
			0.65 
			& 9.73 & 30.63 & \ding{217} 2.82 & \ding{217} 8.81 & 4.16 & 13.06\\
			\hline
		\end{tabular}
		\caption{Head Fraction Comparison. Table registering values of $ \ell^{1} $-norms of arrays $  \{\bar{\ups}(\eff, P_{\text{blT}}, f)(h): \eff \in \mathcal{E}, \, h = 1, 2, \ldots, h_{f}(P_{\text{blT})} \} $ and  $ \{\bar{\ups}(GbE_{DPS}, P_{\text{blT}}, f)(h): h = 1, 2, \ldots, h_{f}(P_{\text{blT}} )\} $. The values are displayed for $ f \in \text{full domain} $ and $\dist \in \{\text{Ub, Pd, Bd} \} $. The remaining parameters are $ o = 0.9 $, $ r = 54 $, $ s = \boldsymbol{\gamma} $ and Minimum List Size $ m = 4 $ i.e., $ P_{\text{hlT}} = (512, dist, o, 54, \text{hlT}, \boldsymbol{\gamma}, 0.5, 4) $. The pointing arrows indicate the optimal strategy within its column  or family of comparable experiments.}
		\label{Tbl Head Fraction Method}
	\end{centering}
\end{table}
%
\begin{remark}\label{Rem hlt interpretation}
	As it can be observed in \textsc{Table} \ref{Tbl Head Fraction Method}, the optimal values are attained at the extremes of the head fraction experimental range and this happens in the three distributions in analysis. This is hardly surprising, as a bigger head (for the Poisson distribution) or a bigger tail (for the Uniform and Binomial distributions) have a better chance to capture a big chunk of a real optimal solution. Furthermore, when the range of head fraction is extended, namely if we take $ [a, b] \subseteq [0,1] $ such that $ [0.35, 0.65] \subsetneq [a, b] $, the optima occur at the extremes $ a $ and $ b $. In particular for $ f \in \{0,1\} $ we are back in the original problem and attaining the original optima with the original computational complexity, which defeats the purpose of the D\&C method itself.
\end{remark}
\noindent Adopting the optimal heights for the $ \text{lhT} $ method and recalling the definitions above, the list of critical heights is summarized in the table \ref{Tbl Critical Heights Table Optimal Head Fraction} below. The pointing arrows indicate the comparison height between $ \text{hlT} $ and $ \text{blT} $ tree generation methods.\\
%
%
%
\begin{table}[h!]
	\begin{centering}
		\rowcolors{2}{gray!25}{white}
		\begin{tabular}{c c c c c c c }
			\hline
			\rowcolor{gray!50}
			Var &  \multicolumn{2}{ c }{Uniform, $ f = 0.35 $}  
			& \multicolumn{2}{ c }{Poisson, $ f = 0.65 $} 
			& \multicolumn{2}{ c }{Binomial, $ f = 0.35 $} \\
			\rowcolor{gray!50}
			$ v $ &  Head-Left & Balanced & Head-Left & Balanced & Head-Left & Balanced\\
			\hline
			$ o $ 
			& 5  & \ding{217} 4 & 4 & \ding{217} 3 & \ding{217} 4 & \ding{217} 4\\
			$ r $ 
			& \ding{217} 5 & \ding{217} 5 & \ding{217} 2 & 3 & \ding{217} 2 & 3\\
			$ s $ 
			& \ding{217} 4 & 6 & \ding{217} 3 &  4 & \ding{217} 3& 5\\
			\hline
		\end{tabular}
		\caption{Critical Heights table. Each corresponds to the expected values of efficiencies $ \mathcal{E} $ coming from the experiments with parameters $ P = (512, dist, o, r, \text{lfT}, s, f, 4) $ ($ f = 0.35 $ for $ \dist \in \{\text{Ud, Bd} \}$, $ f = 0.65 $ for $ \dist = Pd $) or $ P = (512, dist, o, r, \text{blT}, s, 4)   $. The heights pointed with arrows are the values valid for comparison between the  $ \text{lfT} $ and $ \text{blT} $ tree generation methods.
		}
		\label{Tbl Critical Heights Table Optimal Head Fraction}
	\end{centering}
\end{table}

\noindent Once the heights' comparison values are found, we proceed to compare both methods in analogous conditions i.e., when the remaining variables are equal. The results for the demand-capacity fraction variable $ o $, running through its full domain, are summarized in \textsc{Table} \ref{Tbl Occupancy Left Head Balanced Methods Comparison Head Optimal}. Similar tables were constructed for the price-capacity rate $ r \in \{34, 44, 54, 64, 74 \}$ and the sorting $ s \in \{\p, \c, \boldsymbol{\gamma}, \text{random}  \}$ variables running through their respective full domains which we omit here for the sake of brevity.\\

\noindent It is important to notice that in \textsc{Table} \ref{Tbl Occupancy Left Head Balanced Methods Comparison Head Optimal} all the values corresponding to the $ \text{blT} $ method are lower than its corresponding analogous for the $ \text{hlT} $ algorithm. The same phenomenon can be observed for the table running through the price-capacity rate $ r $. The table running through the sorting variable also shows clear predominance of the $ \text{blT} $ over the $ \text{hlT} $ method, though it is not absolute (10 out of 12 cases) as in the previous cases. Furthermore, noticing the differences of values, it follows that $ \text{blT} $ produces significant better results than $ \text{hlT} $. Therefore, this choice of strategy when using the D\&C approach is clear and the remaining strategies need to be decided based on $ \text{blT} $ tree generation algorithm results. 
%
%
\begin{remark}[Head Fraction $ f $ values]
	With regard to the optimal head fraction values it is important to notice the following
	\begin{enumerate}[(i)]
		\item As expected, the optimal values tend to be on the extremes $ f = 0.35 $ or $ f = 0.65 $, since $ f = 0 $ or $ f = 1 $ would imply that no D\&C pair has been introduced and therefore the efficiency should be 100\%.
		
		\item Notice that $ \text{hlT} $ generated D\&C tree for $ f \neq 0.5 $ will be deeper than its analogous for $ \text{blT} $, see \textsc{Figures} \ref{Fig Tree Generated by Head-Left Algorithm}, \ref{Fig Tree Generated by Head-Left Algorithm Biased} and \ref{Fig Tree Generated by Balanced Left-Right Algorithm}. In particular, the $ \text{hlT} $ method with optimal head fraction values ($ f = 0.35 $, $ f = 0.65 $) has higher complexity than its $ \text{blT} $ analogous.
		
		\item A similar comparison procedure was done between $ \text{hlT} $ and $ \text{blT} $, when $ f = 0.5 $ i.e., the standard. As expected, the $ \text{hlT} $ yields poorer results than using the optimal head fraction values and $ \text{blT} $ is remarkably superior.
	\end{enumerate}
\end{remark}
\begin{table}[h!]
	\begin{centering}
		\rowcolors{2}{gray!25}{white}
		\begin{tabular}{c c c c c c c }
			\hline
			\rowcolor{gray!50}
			Occupancy &  \multicolumn{2}{ c }{Uniform}  & \multicolumn{2}{ c }{Poisson} 
			& \multicolumn{2}{ c }{Binomial} \\
			\rowcolor{gray!50}
			$ o $ &  \multicolumn{2}{ c }{Comparison Height $ \tilde{h} = 4 $ }  & 
			\multicolumn{2}{ c }{Comparison Height $ \tilde{h} = 3$} 
			& \multicolumn{2}{ c }{Comparison Height $ \tilde{h} = 4$} \\
			\rowcolor{gray!50}
			&  Head-Left & Balanced & Head-Left & Balanced & Head-Left & Balanced\\
			\rowcolor{gray!50}
			&  $ f = 0.35 $ &  & $ f = 0.65 $ &  & $ f = 0.35 $ & \\
			\hline
			0.50 
			& 98.66 & 7.19 & 35.51 & 3.12 & 49.59 & 7.00\\
			0.55 
			& 86.04 & 7.69 & 32.95 & 2.65 & 44.37 & 5.64\\
			0.60 
			& 74.47 & 7.57 & 29.20 & 2.58 & 40.19 & 5.86\\
			0.65 
			& 64.25 & 6.73  & 25.35 & 2.42 & 37.71 & 5.33\\
			0.70 
			& 55.42 & 6.16 & 21.83 & 2.01 & 34.41 & 4.46\\
			0.75 
			& 47.39 & 5.56 & 18.61 &  1.77 & 30.73 & 4.56\\
			0.80 
			& 40.53 & 5.13 & 15.58 & 2.09 & 27.25 & 3.81\\
			0.85 
			& 34.35 & 3.74 & 12.57 & 2.30 & 23.93 & 4.00\\
			0.90 
			& 26.45 & 3.81 & 8.81 & 2.08 & 18.83 &  4.11\\
			\hline
		\end{tabular}
		\caption{Occupancy Fraction Comparison. Table registering values of $ \ell^{1} $-norms of arrays $  \{\bar{\ups}(\eff, P_{\talg}, o)(h): \eff \in \mathcal{E}, \, h = 1, 2, \ldots, \tilde{h}\} $ for $ o \in \text{full domain} $, $ \talg \in \{\text{hlT}, \text{blT} \}$ and $\dist \in \{\text{Ub, Pd, Bd} \} $. The remaining parameters are $ r = 54 $, $ s = \boldsymbol{\gamma} $, Left Head Fraction $ f = 0.35 $ for $ \dist \in \{\text{Ud, Bd} \}$, $ f = 0.65 $ for $ \dist = Pd $ (if $ \talg = \text{hlT} $) and Minimum List Size $ m = 4 $ i.e., $ P = (512, dist, o, 54, \talg, \boldsymbol{\gamma}, 
			f \in \{0.35, 0.65\}, 4) $. Observe that in all the instances of the problems, the $ \text{blT} $ method gives better results than the $ \text{hlT} $. 
		}
		\label{Tbl Occupancy Left Head Balanced Methods Comparison Head Optimal}
	\end{centering}
\end{table}
\subsection{Optimal Strategies 
}\label{Sec Best Strategies}
%
%
In the previous section, it was determined that $ \text{blT} $ produces better results than $ \text{hlT} $. Consequently, from now on, we focus on finding the best values for the remaining parameters: $ o $, $ r  $ and $ s $ conditioned to the $ \text{blT} $ tree generation method. 

\noindent First we revisit the pruning height of the tree: given that $ \tilde{h} \leq h_{v}(P_{\text{blT}})  $ (as introduced in \textsc{Definition} \ref{Def Critical Height} (v)) and the analysis is now narrowed down to the $ \text{blT} $ method, the computations will be done for these heights because is desirable to stretch the D\&C method as far as possible but within the quality deterioration control established by $  h_{v}(P_{\text{blT}}) $ (see \textsc{Definition} \ref{Def Critical Height} (iv)). 

\noindent Second, now we analyze the method from two points of view. A global one, as it has been done so far accounting for the overall efficiency of the variables in $ \mathcal{E} $ by computing the $ \ell^{1} $-norm of the array $ \{\bar{\ups}(\eff, P_{\text{blT}}, v)(h): \eff \in \mathcal{E}, \, h = 1, 2, \ldots, \tilde{h}\} $ as introduced in \textsc{Definition} \ref{Def Critical Height} (v). A specialized and second point of view, uses only the $ \ell^{1} $-norm of the array $ \{\bar{\ups}(GbE_{DPS}, P_{\text{blT}}, v)(h): h = 1, 2, \ldots, \tilde{h}\} $ i.e., regarding only the behavior of the efficiency $ GbE_{DPS} $,
through the variables $ o $, $ r $ and $ s $.  This specialized measurement is presented because the efficiency of the Exact Solution ($ DPS $) is the most important parameter, given that it contains the behavior of the exact solution. 

\noindent In \textsc{Table} \ref{Tbl Occupancy Balanced Methods} below the $ GbE_{DPS} $ and the global efficiency $ \eff \in \mathcal{E} $ are presented for the demand-capacity fraction variable $ o $, running through its full domain. The pointing arrows indicate the optimal strategy within its column  or family of comparable experiments. As in the previous stage, similar tables were built for the price-capacity rate $ r \in \{34, 44, 54, 64, 74 \}$ and the sorting $ s \in \{\p, \capac, \boldsymbol{\gamma}, \,\text{random}  \}$ variables, running through their respective full domains, which we omit here for the sake of brevity. Finally, in the tables \ref{Tbl Strategies Summary DPS} and \ref{Tbl Strategies Summary Global} below we summarize the optimal strategies from both points of view, the specialized $ GbE_{DPS} $ and the global one $ \eff \in \mathcal{E} $. 
\begin{table}[h!]
	\begin{centering}
		\rowcolors{2}{gray!25}{white}
		\begin{tabular}{c c c c c c c }
			\hline
			\rowcolor{gray!50}
			Occupancy &  \multicolumn{2}{ c }{Uniform}  & \multicolumn{2}{ c }{Poisson} 
			& \multicolumn{2}{ c }{Binomial} \\
			\rowcolor{gray!50}
			$ o $ 
			&  \multicolumn{2}{ c }{Height $ h_{o}(P_{\text{blT}}) = 4 $ }  
			& \multicolumn{2}{ c }{Height $ h_{o}(P_{\text{blT}}) = 3 $} 
			& \multicolumn{2}{ c }{Height $ h_{o}(P_{\text{blT}}) = 4 $} \\
			\rowcolor{gray!50}
			& $ GbE_{DPS}$ &  $ \eff \in \mathcal{E} $ 
			& $ GbE_{DPS}$ &  $ \eff \in \mathcal{E} $  
			& $ GbE_{DPS}$ &  $ \eff \in \mathcal{E} $  \\
			\hline
			0.50 
			& 1.28  & 7.19 & 0.67 & 3.12 & 1.83 & 7.00\\
			0.55 
			& 1.14 & 7.69 & 0.58 & 2.65 & 1.35 & 5.64\\
			0.60 
			& 1.08 & 7.57 & 0.50 & 2.58 & 1.30& 5.86\\
			0.65 
			& 0.95 & 6.73  & 0.54 & 2.42 & 1.29 & 5.33\\
			0.70 
			& 0.87 & 6.16 & \ding{217} 0.33 & 2.01 & 1.00 & 4.46\\
			0.75 
			& 0.79 & 5.56 & 0.42 & \ding{217} 1.77 & 0.98  & 4.56\\
			0.80 
			& \ding{217} 0.75 & 5.13 & 0.50 & 2.09 & 0.82 & \ding{217} 3.81\\
			0.85 
			& \ding{217} 0.75 & \ding{217} 3.74 & 0.52 & 2.30 & 0.78 & 4.00\\
			0.90 
			& 0.84 & 3.81 & 0.56 & 2.08 & \ding{217} 0.62 & 4.11\\
			\hline
		\end{tabular}
		\caption{Occupancy Fraction Comparison. Table registering values of $ \ell^{1} $-norms of arrays $  \{\bar{\ups}(\eff, P_{\text{blT}}, o)(h): \eff \in \mathcal{E}, \, h = 1, 2, \ldots, h_{o}(P_{\text{blT}} \} $ and  $ \{\bar{\ups}(GbE_{DPS}, P_{\text{blT}}, o)(h): h = 1, 2, \ldots, h_{o}(P_{\text{blT}} )\} $. The values are displayed for $ o \in \text{full domain} $ and $\dist \in \{\text{Ub, Pd, Bd} \} $. The remaining parameters are $ r = 54 $, $ s = \boldsymbol{\gamma} $ and Minimum List Size $ m = 4 $ i.e., $ P = (512, dist, o, 54, \text{blT}, \boldsymbol{\gamma}, 4) $. The pointing arrows indicate the optimal strategy within its column  or family of comparable experiments.}
		\label{Tbl Occupancy Balanced Methods}
	\end{centering}
\end{table}
\begin{table}[h!]
	\begin{centering}
		\rowcolors{2}{gray!25}{white}
		\begin{tabular}{c c c c c c c c c c }
			\hline
			\rowcolor{gray!50}
			Var &  \multicolumn{3}{ c }{Uniform}  & \multicolumn{3}{ c }{Poisson} 
			& \multicolumn{3}{ c }{Binomial} \\
			\rowcolor{gray!50}
			$ v $ 
			& Strategy & Height & Error
			& Strategy & Height & Error
			& Strategy & Height & Error \\
			\rowcolor{gray!50}
			&  & $ h_{v}(P_{\text{blT}}) $ & $ \ell^{1} $
			&  & $ h_{v}(P_{\text{blT}}) $ & $ \ell^{1} $
			&  & $ h_{v}(P_{\text{blT}}) $ & $ \ell^{1} $ \\
			\hline
			$ o $ 
			& 0.80/0.85  & 4 & 0.75 & 0.70 & 3 & 0.33 & 0.90 &  4 & 0.62 \\
			$ r $ 
			& 34 & 5 & 1.08 & 34 & 3 & 0.22 & 34 & 3 & 0.15 \\
			$ s $ 
			& $ \capac$ & 6 & 2.85 
			& $ \text{random} $ & 4 & 1.07
			& $ \capac$ & 5 & 1.01
			\\
			\hline
		\end{tabular}
		\caption{Chosen Strategies Table. Summary of best strategies. The expected errors are measured with the $ \ell^{1} $-norms of arrays $  \{\bar{\ups}(GbE_{DPS}, P_{\talg}, v)(h):  h = 1, 2, \ldots,  h_{v}(P_{\text{blT}})\} $, for each of the variables $ v \in \{o, r, s \}$. These were used as decision parameters; given that the norms are computed only for the $ BgE_{DPS} $ efficiency, this point of view only considers the exact solution. The tree generation method is the $ \text{blT} $ since it was determined as the best tree generation strategy. 
		}
		\label{Tbl Strategies Summary DPS}
	\end{centering}
\end{table}
\begin{table}[h!]
	\begin{centering}
		\rowcolors{2}{gray!25}{white}
		\begin{tabular}{c c c c c c c c c c }
			\hline
			\rowcolor{gray!50}
			Var &  \multicolumn{3}{ c }{Uniform}  & \multicolumn{3}{ c }{Poisson} 
			& \multicolumn{3}{ c }{Binomial} \\
			\rowcolor{gray!50}
			$ v $ 
			& Strategy & Height & Error
			& Strategy & Height & Error
			& Strategy & Height & Error \\
			\rowcolor{gray!50}
			&  & $ h_{v}(P_{\text{blT}}) $ & $ \ell^{1} $
			&  & $ h_{v}(P_{\text{blT}}) $ & $ \ell^{1} $
			&  & $ h_{v}(P_{\text{blT}}) $ & $ \ell^{1} $ \\
			\hline
			$ o $ 
			& 0.85  & 4 & 3.74 & 0.75 & 3 & 1.77 & 0.80 &  4 & 3.81 \\
			$ r $ 
			& 44 & 5 & 6.77 & 34 & 3 & 1.69 & 64 & 3 & 1.06 \\
			$ s $ 
			& $ \boldsymbol{\gamma} $ & 6 & 12.29 
			& $ \boldsymbol{\gamma} $ & 4 & 3.79
			& $ \boldsymbol{\gamma} $ & 5 & 8.52
			\\
			\hline
		\end{tabular}
		\caption{Summary of Chosen Strategies. In this case the expected errors are measured with the $ \ell^{1} $-norms of arrays $  \{\bar{\ups}(\eff, P_{\talg}, v)(h): \eff \in \mathcal{E}, \, h = 1, 2, \ldots,  h_{v}(P_{\text{blT}})\} $, for each of the variables $ v \in \{o, r, s \}$. These were used as decision parameters; given that the norms are computed through all the efficiencies in $ \mathcal{E} $, this a global point of view. The tree generation method is the $ \text{blT} $ since it was determined as the best tree generation strategy. 
		}
		\label{Tbl Strategies Summary Global}
	\end{centering}
\end{table}
%
%
%
%
%
%
\subsection{Computational Time}\label{Sec Computational Time}
%
%
%
%
In this section, we discuss the computational time needed for the Divide \& Conquer method. To that end, we present the relative times rather than the absolute computational times, as the latter values can greatly vary from one computer to another. More specifically we focus on the relative computational global time (GbT) and stepwise time (SwT), introduced in \textsc{Definition} \ref{Def DC Tree efficiency}, equations \eqref{Eqn Global Times}, \eqref{Eqn Stepwise Times}. \\

\noindent In the tables \ref{Tbl Computational Times DPS} we display the expected values of $ GbT $ and $ SwT $ (after 50 realizations) of the Exact Solution ($ DPS $), for the datasets generated by the three random distributions and taking the problems in standard setting (see \textsc{Definition} \ref{Def Standard Setting}), the corresponding graphs are depicted in \textsc{Figures} \ref{Fig Averaged Computational Times} (a), (b). In the same fashion, \textsc{Table} \ref{Tbl Computational Times LRS} and \textsc{Figures} \ref{Fig Averaged Computational Times} (c), (d), summarize the expected values for $ GbT $ and $ SwT $ when measuring the computational time of the Linear Relaxation Solution ($ LRS $). The Greedy Approximation Solution ($ GAS $) presents an analogous behavior to the $ LRS $, which we omit here for brevity. As it can be observed in both, the table and figures, the difference in computational time is marginal, but is strongly tied to the algorithm that must be solved along the D\&C tree. Moreover, a similar phenomenon will be observed when moving away from the standard setting to the other problem instances explored before (see \textsc{Figure} \ref{Fig Branch Strategies Tree} ), i.e., the computational time is essentially indifferent with respect to the D\&C strategies and the data distribution. \\
%
%
\begin{table}[h!]
	\begin{centering}
		\rowcolors{2}{gray!25}{white}
		\begin{tabular}{c c c c c c c  }
			\hline
			\rowcolor{gray!50}
			Height &  \multicolumn{2}{ c }{Uniform}  & \multicolumn{2}{ c }{Poisson} 
			& \multicolumn{2}{ c }{Binomial} \\
			\rowcolor{gray!50} 
			& $ GbT $ & $ SwT $
			& $ GbT $ & $ SwT $
			& $ GbT $ & $ SwT $ \\
			\hline
			0 &	100	& &	100	& &	100	& \\
			1 &	46.93 &	46.93 &	47.42 &	47.42 &	45.69 &	45.69 \\
			2 &	25.47 &	54.29 &	26.26 &	55.37 &	24.61 &	53.86 \\
			3 &	15.97 &	62.70 &	16.74 &	63.76 &	15.05 &	61.15 \\
			4 &	11.51 &	72.05 &	12.35 &	73.82 &	10.65 &	70.80 \\
			5 &	8.94 &	77.68 &	9.66 &	78.23 &	8.09 &	76.00 \\
			6 &	7.88 &	88.21 &	8.66 &	89.59 &	7.10 &	87.79 \\
			7 &	8.39 &	106.39 & 9.23 &	106.64 & 7.53 &	106.03 \\
			\hline
		\end{tabular}
		\caption{Summary of Computational Times Exact Solution ($ DPS $). See equations \eqref{Eqn Global Times}, \eqref{Eqn Stepwise Times} for the definitions of $ GbT $ and $ SwT $ respectively. See \textsc{Figures} \ref{Fig Averaged Computational Times} (a), (b) for their corresponding depiction.
		}
		\label{Tbl Computational Times DPS}
	\end{centering}
\end{table}
\begin{table}[h!]
	\begin{centering}
		\rowcolors{2}{gray!25}{white}
		\begin{tabular}{c c c c c c c  }
			\hline
			\rowcolor{gray!50}
			Height &  \multicolumn{2}{ c }{Uniform}  & \multicolumn{2}{ c }{Poisson} 
			& \multicolumn{2}{ c }{Binomial} \\
			\rowcolor{gray!50} 
			& $ GbT $ & $ SwT $
			& $ GbT $ & $ SwT $
			& $ GbT $ & $ SwT $ \\
			\hline
			0 &	100	& &	100	& &	100	& \\
			1 &	42.60 &	42.60 &	46.76 &	46.76 &	46.23 &	46.23 \\
			2 &	24.95 &	58.59 &	25.81 &	55.20 &	25.41 &	54.97 \\
			3 &	17.66 &	70.80 &	17.75 &	68.77 &	17.74 &	69.85 \\
			4 &	13.66 &	77.37 &	13.96 &	78.65 &	13.97 &	78.76 \\
			5 &	12.72 &	93.16 &	13.20 &	94.58 &	13.28 &	95.09 \\
			6 &	12.83 &	100.90 & 14.13 & 107.04 & 14.20	& 106.92 \\
			7 &	14.79 &	115.20 & 15.60 & 110.44	& 15.78	& 111.16 \\
			\hline
		\end{tabular}
		\caption{Summary of Computational Times Linear Relaxation Solution. See equations \eqref{Eqn Global Times}, \eqref{Eqn Stepwise Times} for the definitions of $ GbT $ and $ SwT $ respectively. See \textsc{Figures} \ref{Fig Averaged Computational Times} (c), (d) for their corresponding depiction.
		}
		\label{Tbl Computational Times LRS}
	\end{centering}
\end{table}

\noindent In the numerical results, we observe that the $ DPS $ shows an exponential decay for the $ BgT $, which is consistent with the almost linear behavior of the $ SwT $. The critical $ GbT $ point is $ h = 6 $ because taking $ h \geq 7 $ would produce bigger computational time and a lower quality solution, i.e., deterioration in both features with respect to $ h = 6 $. On the other hand, while the $ LRS $ shows also an exponential decay, it has wilder behavior in the $ SwT $ which scales up to a shift in the critical point $ h = 5 $ in its $ GbT $, i.e., for $ h \geq 6 $ there is no longer a trade-off between solution quality and computational time.\\

\noindent The existence of critical points in the $ GbT $ mentioned above, occurs because some sizes of the problem are small enough for the algorithm ($ DPS, LRS $ or $ GAS $) to become quite efficient, therefore, the decomposition of a given problem in multiple parts such as the D\&C tree generation, together with the reassembling of the problems' results (computation of pruned trees, leaves and sums, see \textsc{Definition} \ref{Def DC Tree efficiency}), add up to higher computational times. Further experiments with different size for the original 0-1 Minimization KP are summarized in the table \ref{Fig Critical Heights}. From there, it follows that the D\&C will continue to trade-off computational time vs. solution's quality, until the size of the subproblems is 16 for the Exact Solution $ DPS $ and 32 for the for the bounds $ LRS, GAS $.
\begin{table}[h!]
	\begin{centering}
		\rowcolors{2}{gray!25}{white}		
		\begin{tabular}{ccccccc}
			\rowcolor{gray!50}
			Problem
			& \multicolumn{2}{ c }{$ LRS $}  & \multicolumn{2}{ c }{$ DPS $}  & \multicolumn{2}{ c }{$ GAS $} \\
			\rowcolor{gray!50}
			Size
			& Height & Size  & Height & Size  & Height & Size \\
			\hline
			128 & 3 & 32 & 4 & 16 & 3 & 32\\
			256 & 4 & 32 & 5 & 16 & 4 & 32\\
			512 & 5 & 32 & 6 & 16 & 5 & 32\\
			1024 & 6 & 32 & 7 & 16 & 6 & 32\\
			\hline		
		\end{tabular}
		\caption{Critical D\&C tree Heights and associated subproblem sizes, for several sizes of the original 0-1 Minimization KP.}
		\label{Fig Critical Heights} 
	\end{centering}
\end{table}
\begin{figure}[t]
	\centering
	\begin{subfigure}[Global Computational Time $ GbT $, $ DPS $. See \textsc{Table} \ref{Tbl Computational Times DPS}.]
		{\includegraphics[scale = 0.350]{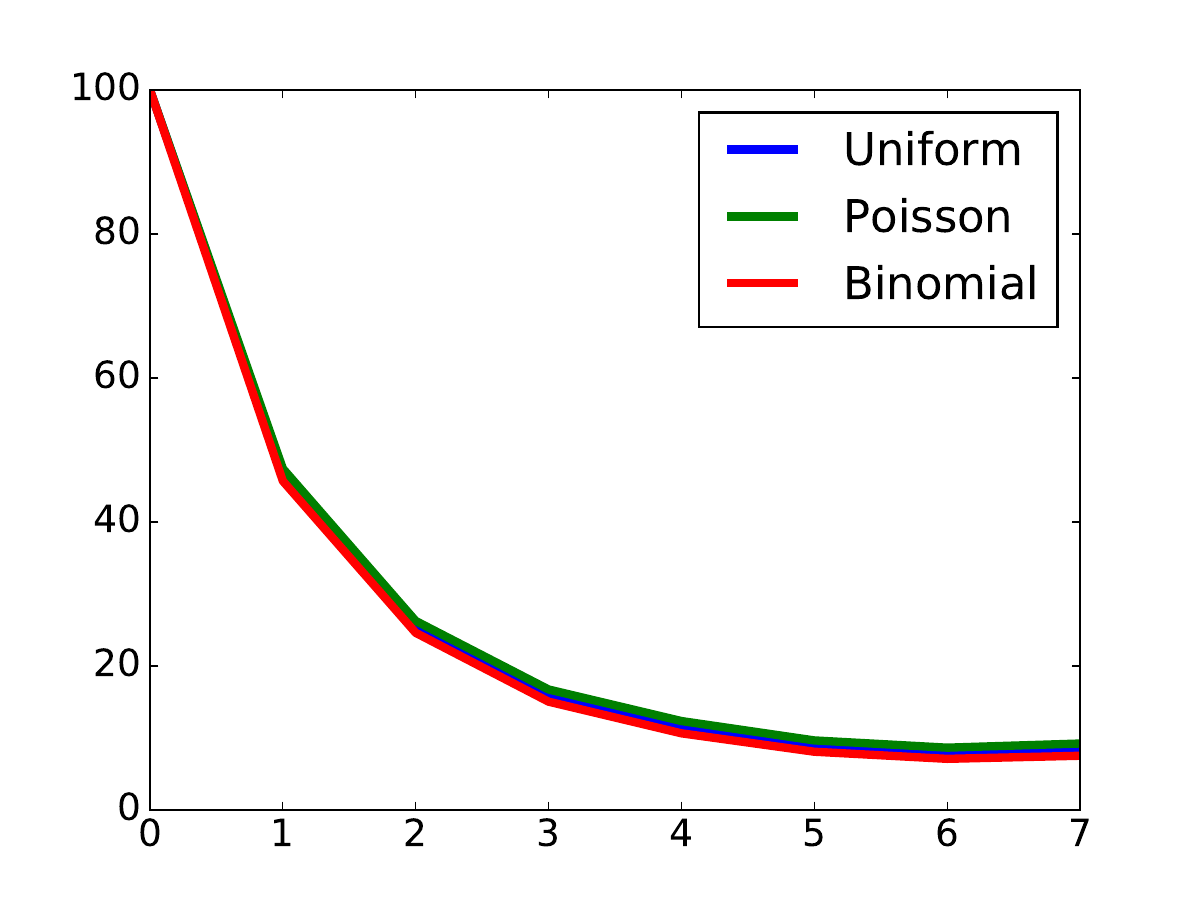} } 
	\end{subfigure}
	~ 
	\begin{subfigure}[Stepwise Computational Time $ SwT $, $ DPS $. See \textsc{Table} \ref{Tbl Computational Times DPS}.]
		{\includegraphics[scale = 0.350]{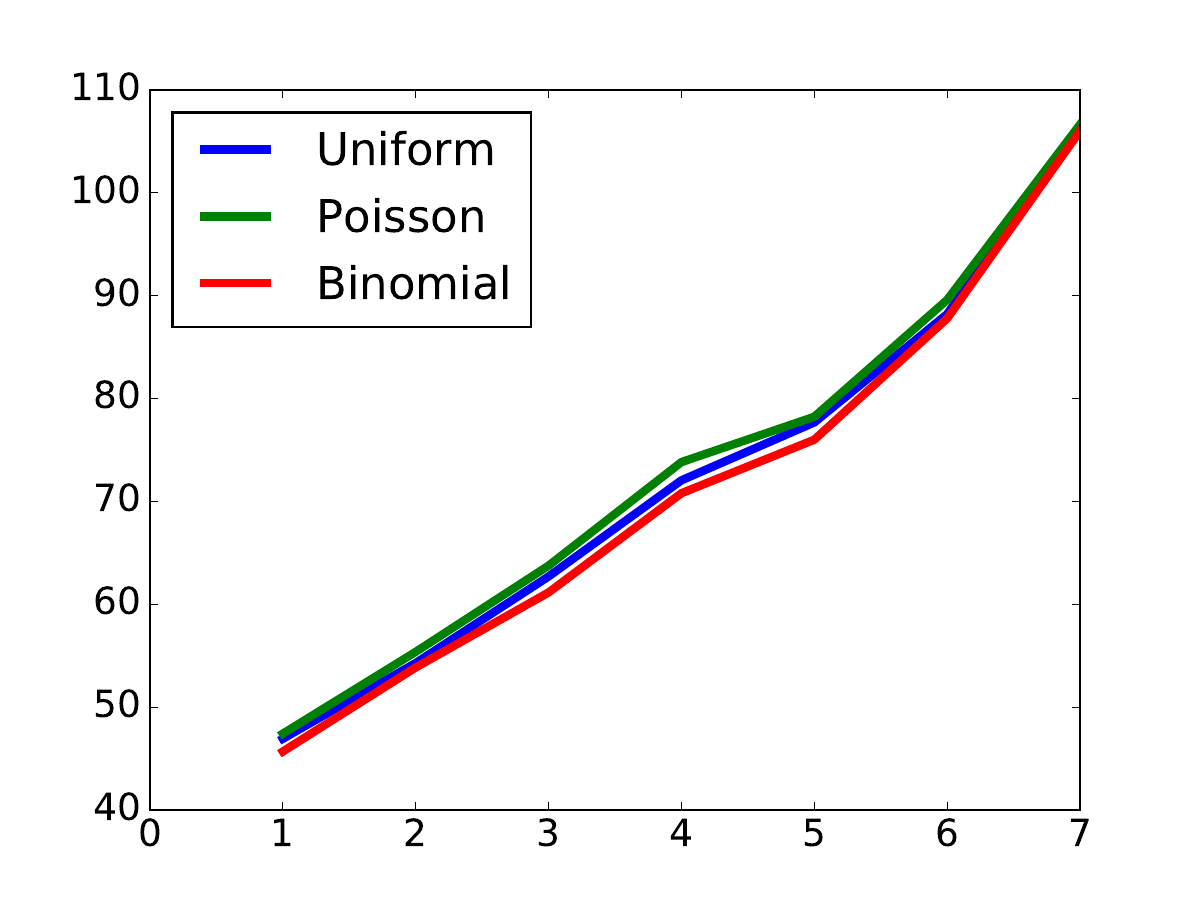} } 
	\end{subfigure}
	\begin{subfigure}[Global Computational Time $ GbT $, $ LRS $. See \textsc{Table} \ref{Tbl Computational Times LRS}.]
		{\includegraphics[scale = 0.350]{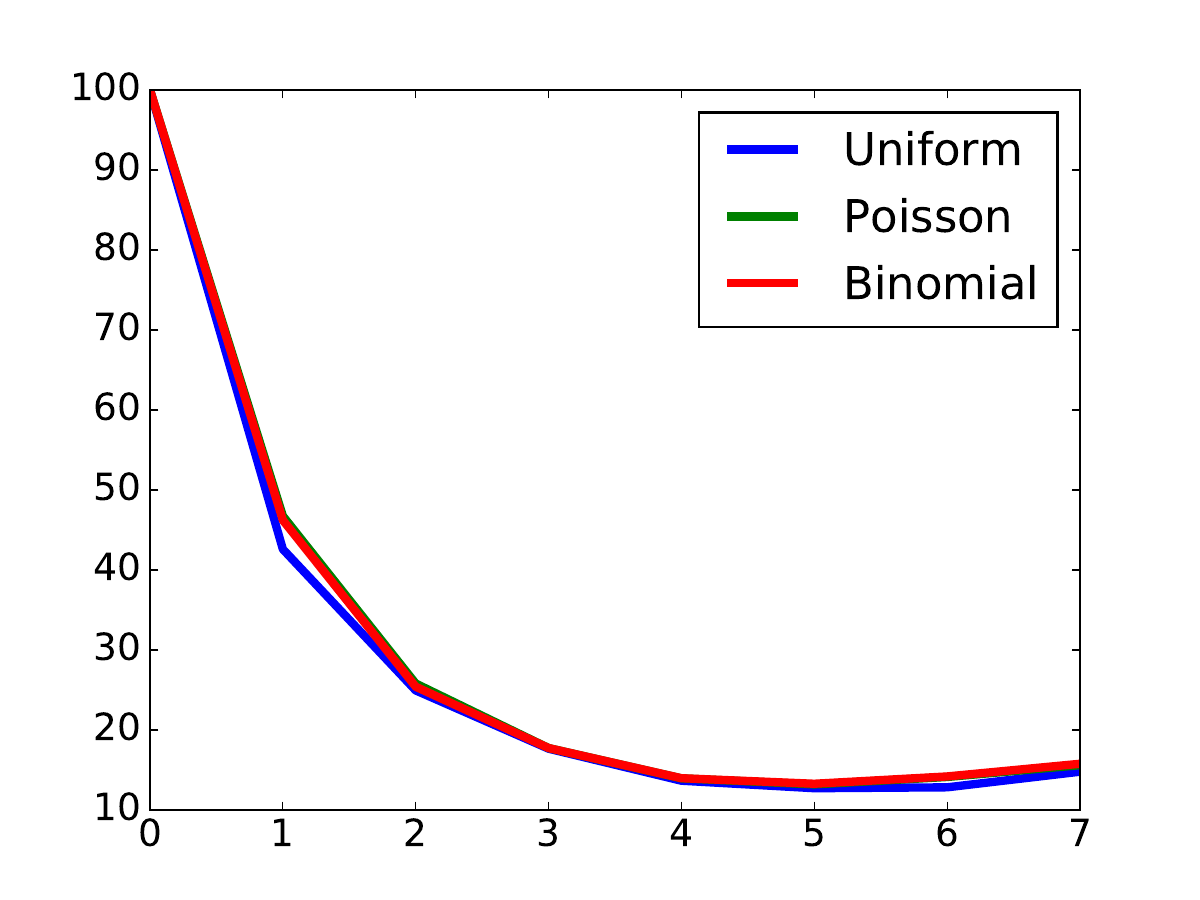} } 
	\end{subfigure}
	~ 
	\begin{subfigure}[Stepwise Computational Time $ SwT $, $ LRS $. See \textsc{Table} \ref{Tbl Computational Times LRS}.]
		{\includegraphics[scale = 0.350]{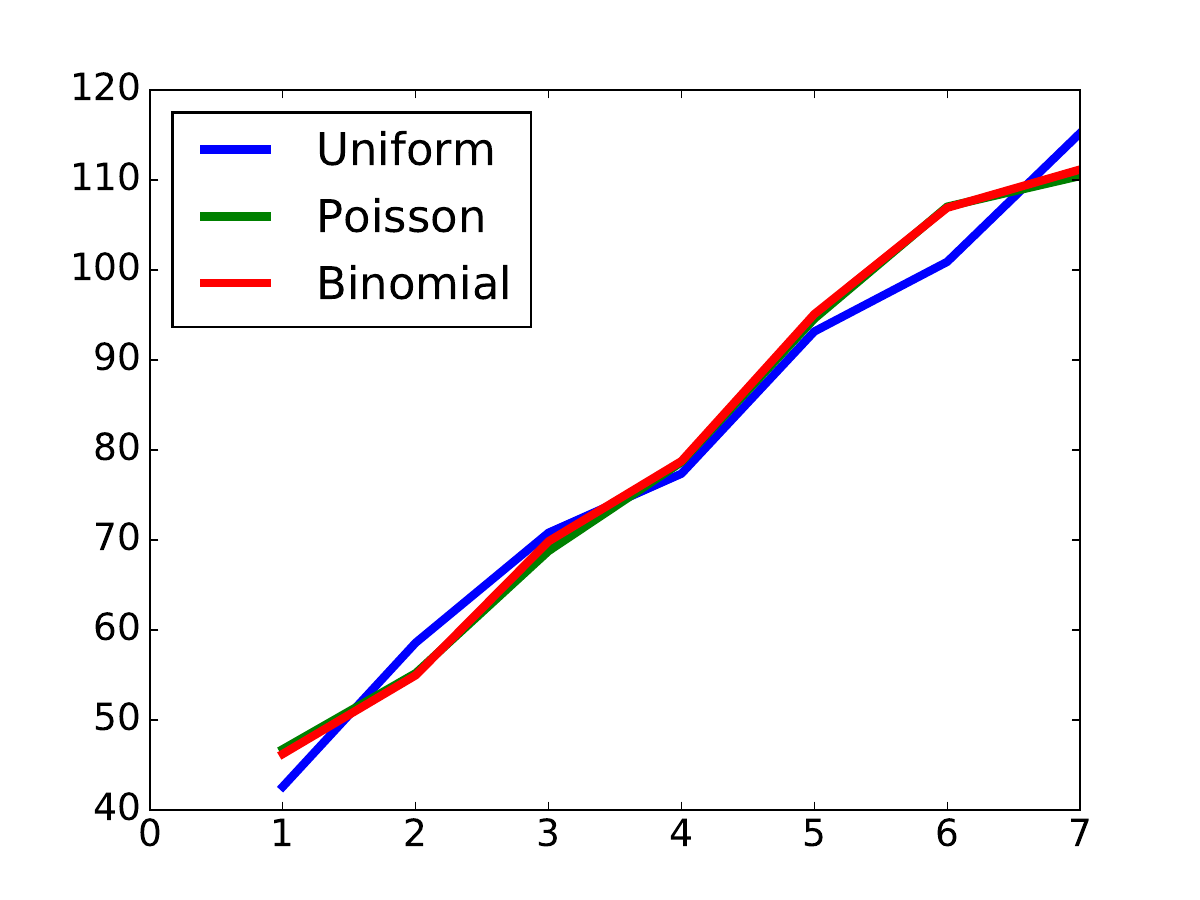} } 
	\end{subfigure}
	\caption{The figures display the average computational time values for 50 random realizations: global $ GbT $ (see \textsc{Equation} \eqref{Eqn Global Times}) and Stepwise $ SwT $ (see \textsc{Equation} \eqref{Eqn Stepwise Times}). The Exact Solution $ DPS $ (see \textsc{Table} \ref{Tbl Computational Times DPS}) is depicted in figures (a), (b), while the Linear Relaxation Solution $ LRS $ (see \textsc{Table} \ref{Tbl Computational Times LRS}) is displayed in figures (c), (d). The experiments were done in the standard setting (see \textsc{Definition} \ref{Def Standard Setting}). Other instances of the problem show similar behavior to its analogous in standard setting. The Greedy Approximation Solution (GAS) shows similar behavior to $ LRS $.}
	\label{Fig Averaged Computational Times} 
\end{figure}
%
%
%
%
%
%
\section{Conclusions and Final Discussion}\label{Sec Conclusions}
%
%
%
%
The present work yields the following conclusions. The heuristics of the method can be summarized as
\begin{enumerate}[(i)]
	\item We have proposed A Divide and Conquer method to solve the Knapsack Problem at large scale. The method reduces the computational time at the expense of loosing quality in the solution. Consequently, the central goal of the paper is to minimize the quality loss by finding the optimal strategies to use the method.
	
	\item The deterioration of the solution's accuracy and/or other parameters of control (such as upper and lower bounds) is defined as the efficiency of the method, and it is the main quantity to asses the quality of the method.  
	
	\item The method is heuristic therefore, several scenarios need to be explored in order to asses its efficiency. The scenarios are modeled using intermediate variables, some deterministic and some probabilistic, e.g. $ \text{lhT, blT} $ tree generation methods, distribution of capacities $ \dist\in \{\text{Ud, Pd, Ud} \} $ respectively, see \textsc{Figure} \ref{Fig Branch Strategies Tree}.

	\item The assessment of strategies is done statistically, using random realizations, computing the respective averages and appealing to the Law of Large Numbers \ref{Th the Law of Large Numbers} to approximate the expected behavior.  
	
	\item 
	It is important to stress that the D\&C method is not \textbf{directly comparable} with previous algorithms, because it does \textbf{not compete} with them, it \textbf{complements} them. In particular, approximation algorithms (such as those presented here or others included in \cite{kellerer2005knapsack} and \cite{martello1990knapsack}), essentially exact algorithms (such as COMBO from \cite{martello1999dynamic}) or exact algorithms (such as a naive Dynamic Programming implementation, or MT1 from \cite{martello1990knapsack}) can be combined with it. Matter of fact, it must be combined with a solution algorithm at certain level of branching if it is to produce an approximate solution at all. 
\end{enumerate}
From the results point of view
\begin{enumerate}[(i)]
	\item The D\&C method can be applied several times to the original KP and generate a tree of subproblems, as those depicted in \textsc{Figures} \ref{Fig Tree Generated by Head-Left Algorithm}, \ref{Fig Tree Generated by Head-Left Algorithm Biased}, \ref{Fig Tree Generated by Balanced Left-Right Algorithm}. However, it is not reasonable to branch the problem producing subproblems smaller than $ n = 32 $ (see \textsc{Section} \ref{Sec Computational Time}) due to the trade-off between computational time and quality deterioration. Such limit is denoted by $ h_{v}(P_{\text{blT}}) $ and it constitutes the first strategy in applying the D\&C method within a reasonable range of efficiency. 
	
	\item Two methods have been introduced to iterate the D\&C heuristics, namely $ \text{lhT, blT} $. They are compared after a common limit for the branching has been established: $ \tilde{h} \defining \min\{ h_{v}(P_{\text{blT}}), h_{v}(P_{\text{hlT}})\} $. Next, the efficiencies of both methods are compared from three points of view: demand-capacity fraction (e.g. \textsc{Table} \ref{Tbl Occupancy Left Head Balanced Methods Comparison Head Optimal}), price-capacity rate and sorting method. It follows that the $ \text{blT} $ furnishes significantly better results than the $ \text{hlT} $ in most of the possible scenarios.
	
	\item Once the $ \text{blT} $ algorithm has been determined as the best tree/branching generation method, the remaining optimal strategies are searched from two points of view: a specialized one, focused on the exact solution only $ GbE_{DPS} $, and a global one analyzing also the decay of the bounds of control $ GAE, LRE $. The results are
	summarized in the tables \ref{Tbl Strategies Summary DPS} and \ref{Tbl Strategies Summary Global} above.
	
	\item
	As it can be seen, the optimal strategies disagree from one point of view to the other for most of the cases. It is useful to have these information for both cases because in practice, depending on the method to be used in solving the family of subproblems derived from successive applications of the D\&C branching, it may be more convenient to prioritize one point of view over the other. For instance, if the family of subproblems will be solved using Exact, then $ GbE_{DPS} $ is more important. On the other hand, if the method includes bounds control (quantified in $ GAE $ and $ LRE $) the global point of view may be preferable. 
	
	\item It is also important to stress that in most of the cases $ GbE_{DPS} $ represents, in average, a fraction of  33\% of the global efficiency. This shows that when applying the D\&C method, the deterioration of the exact solution's quality is important with respect to the deterioration of the bounds' quality.
	
	\item A paramount feature is that the D\&C method deteriorates within reasonable values. In the case of $ GbE_{DPS} $, a maximum expected error of $ 2.85\% $ is observed. However, such an error occurs after the 6th D\&C iteration, which drastically reduces the computational time. On the other hand, the global quantification $ \eff \in \mathcal{E} $, presents a quality decay of 12.29\% in the worst case scenario but again, 6 D\&C iterations were used and this value encompasses all the efficiencies. It follows that the proposed method is efficient.
	
	\item The computational time is indifferent with respect to the strategies for the D\&C tree design as well as the data probabilistic distribution. 
\end{enumerate}
The present paper opens up new research lines to be explored in future work
\begin{enumerate}[(i)]
	\item The reduction of computational time and critical problem sizes, discussed in \textsc{Section} \ref{Sec Computational Time}, were quantified considering a serial algorithm implementation. A parallel implementation, on the other hand may furnish better results, because a D\&C iteration produces two fully decoupled optimization problems. The assessment of computational time for a parallel scheme will be pursued in future work.

	\item As mentioned above, currently a D\&C iteration produces two fully decoupled subproblems. However, another scheme with partial coupling can be proposed namely introducing a pair of problems like that presented in \textsc{Definition} \ref{Def Divide and Conquer Setting} (ii), \textsc{Problem} \ref{Pblm DC subproblems}, but such that $ A_{0} \cup A_{1} = [N] $ and $ A_{0} \cap A_{1} \neq \emptyset $; with assigned demands $ D_{0} $, $ D_{1} $, computed by rules analogous to \textsc{Equation} \eqref{Eqn Demand Fraction} i.e., construct artificially an integer problem with the structure 
	\begin{problem}[$ \Pi^{b}, \, b = 0, 1 $]\label{Pblm Dantzig-Wolfe}
		\begin{subequations}\label{Eqn Dantzig-Wolfe Problem}
			\begin{equation}\label{Eqn Dantzig-Wolfe Problem Objective Function}
			\min \Bigg[\sum\limits_{b\, \in \, \{0,1\} }\sum\limits_{i\, \in \, A^{b} } p_{i}x_{i} 
			- \sum\limits_{j\, \in \, A^{0}\cap A^{1} } p_{j}x_{j}\Bigg] ,
			\end{equation}
			subject to
			\begin{align}\label{Eqn Dantzig-Wolfe Problem Capacity Constraint}
			& \sum\limits_{i\, \in \, A^{0} } c_{i}x_{i} \geq D^{0},&
			& \sum\limits_{i\, \in \, A^{1} } c_{i}x_{i} \geq D^{1}, &
			& x_{i} \in \{0,1\}, \; \forall\, i \in [N] .
			\end{align}
		\end{subequations}
	\end{problem}
	A future line of research is the optimal choice of coupling/overlapping sets $ A_{0} \cap A_{1} \neq \emptyset $ and exploit the structure of the integer programming problem \ref{Pblm Dantzig-Wolfe} (analogously to the Dantzig-Wolfe decomposition for linear problems with the same structure). Furthermore, the optimality has to be analyzed from the perspective quality vs. computing time.  
	
	\item In this work, the method used a static choice of strategies, e.g. if the sorting method was $ s = \boldsymbol{\gamma} $, it remained constant through all the nodes of the D\&C tree (as \textsc{Table} \ref{Tbl Balanced Tree for Particular Example}, \textsc{Figure} \ref{Fig Tree Generated by Balanced Left-Right Algorithm} illustrate). A future line of research is to investigate the effect of mixing the strategies, e.g. the sorting parameter $ s $ taking different  values from $ \{ \p, \capac, \boldsymbol{\gamma}, \text{random} \} $ from one node to its children, or from one height (tree level) to the next.
	
	\item The $ \text{blT} $ algorithm is significantly superior to the $ \text{hlT} $ method; the numerical evidence suggests that an analytic proof of this conjecture is plausible. A future line of research is to look for a rigorous mathematical proof, which of course, would use probability theory and furnish its results in terms of expected efficiencies.   
	
	\item Finally, a future line of research is the implementation and assessment of the D\&C method for the optimization of general linear integer programs. However, such a step should be done only once the aforementioned aspects have been deeply studied.  
\end{enumerate}
%
%
%
%
%
\section*{Acknowledgments}
%
%
\noindent The first Author wishes to thank Universidad Nacional de Colombia, Sede Medell\'in for supporting the production of this work through the project Hermes 45713 as well as granting access to Gauss Server, financed by ``Proyecto Plan 150x150 Fomento de la cultura de evaluaci\'on continua a trav\'es del apoyo a planes de mejoramiento de los programas curriculares". (\url{gauss.medellin.unal.edu.co}), where the numerical experiments were executed. The second Author wishes to thank Universidad EAFIT for its financial support as MSc student, through the Internal Grant 819156 ``Modelos matem\'aticos y m\'etodos de soluci\'on a un tipo de problema log\'istico que involucha agrupamiento de clientes, distribuci\'on y ruteo". The authors also wish to thank the anonymous referees whose meticulous review and insightful suggestions enhanced substantially the quality of this work. Special thanks to Professor Daniel Cabarcas from Universidad Nacional de Colombia, Sede Medell\'in for his help in understanding and running the code COMBO from \cite{martello1999dynamic}.
\bibliographystyle{plain}

%
%
%
%
%
%
\end{document}